\documentclass[12pt]{article}
\usepackage{amsmath}
\usepackage{amsthm}
\usepackage{amsfonts}
\usepackage{color}
\usepackage{graphicx}
\usepackage{float}
\usepackage{subfig}
\usepackage{amsmath}
\usepackage{amsfonts}
\usepackage{multirow}
\usepackage{setspace}
\usepackage{array}
\usepackage{amssymb}
\usepackage{lscape}
\usepackage[authoryear]{natbib}
\usepackage{color}
\usepackage[usenames,dvipsnames]{xcolor}
\usepackage[colorlinks]{hyperref}
\usepackage{rotating}
\usepackage{tikz}
\usepackage{framed}
\usepackage{cleveref}
\usepackage{autonum}
\usepackage{bbm}
\usepackage{subfig} 

\usepackage{microtype}

\usetikzlibrary{backgrounds}
\usetikzlibrary{arrows,shapes}
\usetikzlibrary{tikzmark}
\usetikzlibrary{calc}

\usepackage{tcolorbox}

\usepackage[margin=1in]{geometry}
\usepackage{caption}

\usepackage{attrib} 

\newtheorem{theorem}{{Theorem}}[section]
\newtheorem{proposition}[theorem]{Proposition}
\newtheorem{lemma}[theorem]{Lemma}
\newtheorem{definition}[theorem]{Definition}

\newtheorem{axiom}[theorem]{Axiom}

\newtheorem{example}[theorem]{{Example}}

\hypersetup{
    unicode=false,          
    pdftoolbar=true,        
    pdfmenubar=true,        
    pdffitwindow=false,     
    pdfstartview={FitH},    
    pdftitle={My title},    
    pdfauthor={Author},     
    pdfsubject={Subject},   
    pdfcreator={Creator},   
    pdfproducer={Producer}, 
    pdfkeywords={keyword1, key2, key3}, 
    pdfnewwindow=true,      
    colorlinks=true,       
    linkcolor=blue,          
    citecolor=blue,        
    filecolor=magenta,      
    urlcolor=cyan           
}

\def\la{q}

\def\ME{\mathcal{E}}
\def\UE{\bar{U}}
\def\compset{\Psi}

\DeclareMathOperator*{\argmin}{\arg\!\min}
\DeclareMathOperator*{\argmax}{\arg\!\max}

\newif\ifcomments
\commentstrue   

\usepackage{subfiles}

\author{Erik Eyster, Shengwu Li, and Sarah Ridout\thanks{Eyster: UC Santa Barbara, \protect\url{erikeyster@ucsb.edu}; Li: Harvard University, \protect\url{shengwu_li@fas.harvard.edu}; Ridout: Vanderbilt University, \protect\url{sarah.e.ridout@vanderbilt.edu}. We thank Mohammad Akbarpour, Nicholas Barberis, Douglas Bernheim, Ben Bushong, Sam Gershman, Ben Golub, David Laibson, Muriel Niederle, Matthew Rabin, Gautam Rao, Jesse Shapiro, Philipp Strack, Bruno Strulovici, Tomasz Strzalecki, and Richard Thaler for valuable comments.  All errors remain our own.}}
\title{A Theory of \textit{Ex Post} Rationalization}

\begin{document}

\maketitle

\abstract{People rationalize their past choices, even those that were mistakes in hindsight.  We propose a formal theory of this behavior.  The theory predicts that sunk costs affect later choices. Its model primitives are identified by choice behavior and it yields tractable comparative statics. (JEL: D11, D81, D91)}

\newpage

\section{Introduction} \label{sec:intro}

People rationalize their past choices.  We look back on our lives and try to make sense of what we have done.  Upon realizing that, in hindsight, we have made a mistake, we can adapt our goals, attitudes, or beliefs to justify the decision. 

Classical economic theory rules out rationalization. It assumes that people make forward-looking choices according to fixed preferences, rather than adapting their preferences to rationalize past decisions. Introspection, common sense, and psychological research all suggest that the classical approach omits a key aspect of human decision-making.\footnote{We review psychological research on rationalization in \Cref{sec:psych_evidence}.}

How can economic models capture \textit{ex post} rationalization? We develop a theory that accommodates this behavior. To fix ideas, consider the following example from \cite{thaler1980toward}.\footnote{The ticket cost \$40 in Thaler's example; we have raised the price due to inflation.}
\begin{example}\label{ex:baseline}
Bob pays \$100 for a ticket to a basketball game to be played 60 miles from his home.  On the day of the game there is a snowstorm.  He decides to go anyway.  If the ticket had been free-of-charge, he would have stayed home.
\end{example}
The \$100 that Bob paid is a sunk cost.  It is not worth going to the basketball game during a snowstorm.  In hindsight, it was a mistake to have bought the ticket.\footnote{Of course, it may have been \textit{ex ante} optimal to buy the ticket. We use ``mistake" as shorthand for ``sub-optimal action".}  But if Bob goes to the game, then he can avoid acknowledging the mistake, by exaggerating his enthusiasm for basketball or by downplaying the hazards of driving through a snowstorm.  If he stays home, then he is inarguably worse off than if he had not bought a ticket in the first place.

There are two key ingredients for this behavior.  First, Bob must have made a choice that was an \textit{ex post} mistake.  Hence, this modified example is far-fetched:
\begin{example}\label{ex:nochoice}
Bob receives a free ticket to a basketball game and loses \$100 due to an unusually high utility bill.  On the day of the game there is a snowstorm. He decides to go anyway.  If he had not lost the \$100, he would have stayed home.
\end{example}
Second, there must be plausible preferences that, if adopted, would justify Bob's earlier decision.  To illustrate this, we replace the physical consequences in \Cref{ex:baseline} with monetary gains and losses.
\begin{example}\label{ex:money}
Bob pays \$100 for a financial option that can only be exercised on the day of the basketball game.  It yields \$180 if exercised in good weather and loses \$20 if exercised in a snowstorm.  On the day of the game there is a snowstorm.  He decides to exercise the option anyway, for a net loss of \$100 + \$20.
\end{example}
\Cref{ex:money} is unnatural because there is no way for Bob to rationalize his initial purchase.  Letting the option expire results in a net loss of \$100, whereas exercising the option results in a net loss of \$120.  More money is better, so Bob has to acknowledge the mistake and cut his losses.

Even at high stakes, decision-makers sometimes rationalize sunk costs rather than acknowledge error.  For instance, a senior Irish Republican Army leader was asked in 1978 whether the years of violent resistance had been worth it.  He replied, ``Virtually nothing has been achieved. We can't give up now and admit that men and women who were sent to their graves died for nothing." \citep[p. 225]{smith2002fighting}\footnote{For further reading on how rationalizations by combatants prolonged the Troubles, see Chapter 3 of \cite{alonso2007}.}

Motivated by these examples, we propose a theory about agents who seek to rationalize their past choices by adapting their preferences. We model an agent facing a decision problem with this structure:
\begin{enumerate}
    \item The agent chooses action $a_1$ from menu $A_1$.
    \item The agent learns the state of the world $s \in S$.
    \item The agent chooses action $a_2$ from menu $A_2(a_1,s)$, which can depend on his first action and the state.
\end{enumerate}

A utility function takes as arguments $a_1$, $a_2$, and $s$.  The agent's \emph{material utility function} is denoted $u$.  The agent may adopt any utility function in the set $\mathcal{V}$, which we call \emph{rationales}.  $u$ and $\mathcal{V}$ are primitives of the model. We assume that $u \in \mathcal{V}$.

We start by describing the agent's choice from menu $A_2(a_1,s)$, after having chosen $a_1$ from menu $A_1$ and learned that the state is $s$.  The agent chooses action $a_2$ from menu $A_2(a_1,s)$ and rationale $v$ from $\mathcal{V}$ to maximize \textit{total utility}, that is
\vspace{\baselineskip}
\begin{equation}\label{eq:intro_rationalizer}
    (1-\gamma)\underbrace{u(a_1,\tikzmarknode{a21}a_2,s)}_{\substack{\text{\scriptsize material} \\ \text{\scriptsize utility}}} + \gamma  \underbrace{\left[\tikzmarknode{v1}{v}(a_1,\tikzmarknode{a22}a_2,s) - \max_{\substack{\hat{a}_1 \in A_1 \\ \hat{a}_2 \in A_2(\hat{a}_1,s)}} \tikzmarknode{v2}{v}(\hat{a}_1,\hat{a}_2,s)\right]}_{\text{\scriptsize rationalization utility}},
\end{equation}
\begin{tikzpicture}[overlay,remember picture,>=stealth,nodes={align=left,inner ysep=1pt},<-]
        \path (a21.north) ++ (-.1,1.15) node[anchor=south west,color=red!67] (scalep){\scriptsize chosen action};
        \draw[<->,color=red!57] ([shift={(0,0.1)}]a21.north) -- ++(0,1)  -|  ([shift={(0,0.1)}]a22.north);
        \path (v1.north) ++ (1.65,.85) node[anchor=south west,color=blue!67] (scalep){\scriptsize chosen rationale};
        \draw[<->,color=blue!57] ([shift={(0,0.1)}]v1.north) -- ++(0,.7)  -| node[] {} ([shift={(0,0.1)}]v2.north);
\end{tikzpicture}
where $\gamma \in [0,1]$ is the weight on \emph{rationalization utility}. Rationalization utility measures how close $(a_1,a_2)$ is to the \textit{ex post} optimum under the chosen rationale $v$. By construction, rationalization utility is no more than $0$.

Observe that if $a_1$ is \textit{ex post} materially optimal, that is
$$\max_{\hat{a}_2 \in A_2(a_1,s)} u(a_1, \hat{a}_2,s) = \max_{\substack{\hat{a}_1 \in A_1 \\ \hat{a}_2 \in A_2(\hat{a}_1,s)}} u(\hat{a}_1,\hat{a}_2,s),$$
then the theory predicts that the agent chooses $a_2$ to maximize material utility, which yields rationalization utility $0$ for rationale $v = u$. Hence, the theory departs from the classical prediction only when the agent has made an \textit{ex post} mistake.

When the agent has made an \textit{ex post} mistake, he may be able to increase rationalization utility by choosing rationale $v \neq u$.  By construction, $a_2$ maximizes a weighted sum of his material utility $u(a_1,a_2,s)$ and his chosen rationale $v(a_1,a_2,s)$, which distorts his choices compared to the classical benchmark.

We now apply the model to the earlier examples.  The rationales $\mathcal{V}$ are parameterized by $\theta \in [0, 400]$.  Utility function $v_{\theta}$ specifies that the agent gets $\theta$ utils for attending the game, $-200$ utils for driving through a snowstorm, and $-p$ utils for paying $p$ dollars.  Material utility is $u = v_{180}$, so a classical agent ($\gamma = 0$) is willing to pay \$180 to attend the game in good weather, but will stay home in a snowstorm.

In \Cref{ex:baseline}, the menu $A_1$ has two alternatives; the agent can get a ticket and lose \$100 or he can decline.  After buying the ticket for \$100, the agent learns that there is a snowstorm.  If he stays home, then his material utility is $-100$, and his rationalization utility is $-100 - \max\{\theta - 200 - 100, 0\}$.  It is optimal to choose $\theta \leq 300$, yielding total utility $(1-\gamma)(-100) + \gamma(-100) = -100$.  By contrast, if the agent attends the game, then his material utility is $180 - 200 - 100 = -120$ and his rationalization utility is $\theta - 200 - 100 - \max\{\theta - 200 - 100, 0\}$. Now it is optimal to choose $\theta \geq 300$, for rationalization utility of $0$ and total utility of $(1-\gamma)(-120)$. By attending the game and exaggerating his enthusiasm, the agent is able to increase rationalization utility at the cost of material utility.  For $\gamma > \frac{1}{6}$, his total utility is strictly higher when he attends the game.

Suppose instead that the ticket was free-of-charge.  Then staying home in a snowstorm leads to no regret under the agent's material utility function $u$.  Hence, the agent maximizes total utility by adopting rationale $v = u$ and staying home.  The agent's behavior exhibits sunk-cost effects; his choice on the day of the basketball game depends on upfront costs that he cannot recover.

In \Cref{ex:nochoice}, the agent has no choice initially, so the menu $A_1$ contains only one alternative: the agent gets a ticket and loses \$100.  This is trivially \textit{ex post} optimal, so the agent maximizes total utility by staying home in the snowstorm.  Hence, removing unchosen alternatives from the menu $A_1$ can alter the agent's later choice from $A_2$.

In \Cref{ex:money}, we have taken the agent's material utility for the outcomes in \Cref{ex:baseline} and converted utils to dollars.  The agent bought the financial option for \$100.  Exercising the option in good weather yields \$180, and exercising it in a snowstorm loses \$20.  But every rationale agrees about money, so there is no room to reduce regret and the agent does not exercise the option in a snowstorm.

Perhaps Bob's behavior is not due to rationalization, but due to an error in reasoning. Bob falls for the sunk cost fallacy: If he buys the ticket and stays home, then he has wasted \$100, whereas if he goes to the game then it was not a waste. But the same logic applies to \Cref{ex:money}---if Bob buys the financial option and does not use it, then he has wasted \$100.  Why cannot he avoid waste by exercising the option (and losing a further \$20)? In the present theory, sunk-cost effects arise when one can rationalize incurring past costs. This provides an explanation for why people commit the sunk cost fallacy in some situations and not in others.\footnote{Of course, this is only a partial account of the psychology at work. \cite{imas2016realization} finds that lab subjects take on more risk after paper losses and less risk after realized losses. The present theory neglects this dependence on framing.}

To complete the model, we specify the agent's behavior when choosing from the first menu $A_1$.  At this point, the agent has no earlier choices to rationalize, so we assume that the agent evaluates choices from $A_1$ according to his expected material utility under some prior on the states $S$.  This depends on the agent's beliefs about his future choice from $A_2$.  A \textit{na\"if} believes he will maximize material utility when choosing from $A_2$; a \textit{sophisticate} correctly foresees his own choices from $A_2$.

Does the theory make testable predictions that do not rely heavily on functional form assumptions? How do the choices of \textit{ex post} rationalizers compare to the classical benchmark?  We study a broader class of decision problems, in which the first actions, second actions, and rationales are complements.  Let the first actions, the second actions, and parameter set $\Theta$ be totally ordered sets. We assume that the rationales $\mathcal{V}$ have the form $\{w(a_1,a_2,\theta,s) : \theta \in \Theta\}$, for some function $w$ that is supermodular in $(a_1,a_2,\theta)$.  For instance, this includes the rationales we posited for \Cref{ex:baseline}, if we impose that \underline{buy a ticket} is a higher action than \underline{don't buy a ticket}, and that \underline{go to the game} is a higher action than \underline{stay home}.  It also includes time-separable utility functions, of the form $w(a_1,a_2,\theta,s) = w_1(a_1,\theta,s) + w_2(a_2,\theta,s)$, with $w_t$ supermodular in $(a_t,\theta)$.  We assume that the menu $A_2(a_1)$ is monotone non-decreasing in $a_1$.

We prove that if the rationalizer's first action was \textit{ex post} too high, then his second action is distorted upwards compared to the classical benchmark. Symmetrically, if the rationalizer's first action was \textit{ex post} too low, then his second action is distorted {downwards} compared to the classical benchmark.

This result yields comparative statics for a variety of settings.  It predicts sunk-cost effects under risky investment---when time-$1$ investment and time-$2$ investment are complements, the agent responds to time-$1$ cost shocks by exaggerating the expected profits of the project and raising time-$2$ investment.  It predicts that agents repeatedly facing identical decisions will have `sticky' choice behavior, responding too little to new information.  In particular, lab subjects who make incentivized reports of priors and posteriors will report posteriors biased towards their priors, and will underweight informative signals compared to subjects who report only posteriors.

We study the effect of \textit{unchosen} time-$1$ alternatives on time-$2$ choice. In the classical model, such alternatives are irrelevant for time-$2$ behavior. By contrast, for a rationalizer facing a supermodular decision problem, raising the unchosen time-$1$ alternatives lowers the agent's time-$2$ choices. This prediction compares cleanly to the zero effect predicted by the classical model.

Next, we study the problem of identification. In general, how is the modeler to specify material utility $u$ and the rationales $\mathcal{V}$?  In some situations, we can use standard restrictions on the preferences agents may plausibly hold.  For instance, for an agent choosing between money lotteries, we could assume that $\mathcal{V}$ is a class of preferences with constant relative risk aversion.  Similarly, for an agent bidding in an auction, we could assume that the rationales $\mathcal{V}$ have different valuations for the object, but are all quasi-linear in money.

If we do not make \textit{a priori} restrictions on $\mathcal{V}$, can we nonetheless deduce the rationales? We prove that the model primitives $u$ and $\mathcal{V}$ are identified by choice behavior.  That is, suppose we start with finitely many outcomes, and each utility function depends on the outcome and the state.  We then construct objective lotteries over outcomes, and extend utility functions by taking expectations.  The agent faces decision problems of this form:
\begin{enumerate}
    \item The agent selects a menu $M$ from a collection of menus of lotteries.
    \item The agent learns the state of the world $s \in S$.
    \item The agent chooses a lottery from $M$.
\end{enumerate}
We find that the agent's choice correspondence pins down material utility $u$ and the rationales $\mathcal{V}$; essentially, these are unique up to a positive affine transformation. Hence, statements about the agent's rationales can be reduced to statements about the agent's choice behavior. We also prove a representation theorem, providing necessary and sufficient conditions for a choice correspondence to be consistent with the theory.

The paper proceeds as follows: \Cref{sec:litreview} reviews related economic theories. \Cref{sec:model_setup} states the theory and discusses interpretations. \Cref{sec:evidence} relates the theory to data; it surveys evidence from psychology about rationalization and contrasts \textit{ex post} rationalization with other explanations of sunk-cost effects. \Cref{sec:comp_stat} provides comparative statics for supermodular decision problems, yielding more testable predictions. \Cref{sec:id} provides a representation theorem and an identification theorem. \Cref{sec:extensions} discusses extensions.
\section{Literature review}\label{sec:litreview}

Various economic theories posit that agents change their beliefs or preferences to align with past actions, in particular settings such as belief updating, voting, and consumption \citep{yariv2005seeit,acharya2018explaining,bernheim2021theory,suzuki2019choice,nagler2021thoughts}.\footnote{This is part of a broader literature that studies agents who choose their beliefs or preferences \citep{akerlof1982economic,rabin1994cognitive,rotemberg1994human,becker1997endogenous,brunnermeier2005optimal}.} The present theory contributes to this literature in two ways: it applies to general two-stage decision problems and its parameters are identified from choice data. There are other key differences. In \cite{yariv2005seeit}, \cite{acharya2018explaining}, and \cite{bernheim2021theory}\footnote{The retrospective motive appears in the preprint of \cite{bernheim2021theory}, but not in the published version.}, the agent's choices do not depend on earlier forgone alternatives, whereas for an \textit{ex post} rationalizer they do. In \cite{suzuki2019choice} the agent suppresses signals about the state, and in \cite{nagler2021thoughts} the agent pays a cost to raise his marginal utility of consumption. By contrast, an \textit{ex post} rationalizer adopts different utility functions to justify his actions. 

Regret theory posits that the agent makes choices today so as to reduce regret tomorrow \citep{savage1951theory, loomes1982regret,loomes1987some, bell1982regret, sarver2008anticipating}. By contrast, an \textit{ex post} rationalizer distorts today's choices so as to justify yesterday's choices. This retrospective motive does not arise in standard regret theory.

Most directly, we build on ideas from \cite{eyster2002rationalizing} and \cite{ridout2020model}.

\cite{eyster2002rationalizing} studies a  two-period model with an agent who wishes to reduce \textit{ex post} regret, assessed according to \textit{material} utility, but limits attention to alternative first actions that are `consistent with' the chosen second action.  To illustrate, the theory of \cite{eyster2002rationalizing} would explain \Cref{ex:baseline} by positing that if Bob attends the game, then only \underline{buy a ticket} is consistent, so he feels no regret.  On the other hand, if Bob stays home, then both \underline{buy a ticket} and \underline{don't buy a ticket} are consistent, so he feels regret for having bought the ticket. One limitation of this approach is that the modeler's intuitions about consistency may vary with how the actions are framed -- \underline{stay home} seems consistent with \underline{don't buy a ticket}, but \underline{stay home with a ticket in hand} does not seem consistent with \underline{don't buy a ticket}, so Bob can also avoid regret by staying home with a ticket in hand. Our present approach
overcomes this framing objection. Instead of a frame-dependent consistency relation, our main primitive is a set of rationales, \textit{i.e.} \textit{post hoc} reasons for the agent's choice, and these are identified from choice data. 

\cite{ridout2020model} studies a model of one-shot choice, with an agent who has a set of `justifiable' preferences and a material preference that is not justifiable.  The agent is constrained to choose only alternatives that maximize some justifiable preference.  The key contrast between these theories is that Ridout's agent desires only material satisfaction but must conceal his true motives, whereas an \textit{ex post} rationalizer seeks to reduce regret from past mistakes.  This retrospective motive does not appear in \cite{ridout2020model}.
\section{Statement of theory}\label{sec:model_setup}

In our model, an agent chooses an action from a menu, then learns the state of the world, and then finally chooses an action from a second menu, which can depend on the first action and the state.

We now define the model primitives.  $\mathcal{A}_1$ denotes the \textbf{first actions}; $\mathcal{A}_2$ denotes the \textbf{second actions}; and $S$ denotes the \textbf{states of the world}, with representative elements $a_1$, $a_2$, and $s$, respectively.

A \textbf{decision problem} $D \equiv (A_1,A_2,F)$ consists of
\begin{enumerate}
    \item a first-period menu $A_1 \subseteq \mathcal{A}_1$,
    \item and a second-period menu correspondence $A_2: A_1 \times S \rightrightarrows \mathcal{A}_2$.
    \item a prior over states $F \in \Delta S$,
\end{enumerate}
We require that $A_1$ and $A_2$ be non-empty.

A \textbf{utility function} is a function $v: \mathcal{A}_1 \times \mathcal{A}_2 \times S \rightarrow \mathbb{R}$.  The \textbf{rationales} are denoted $\mathcal{V}$; these are a set of utility functions that the agent may adopt to justify her actions.  The agent's \textbf{material utility function} is denoted $u$, and we assume that $u \in \mathcal{V}$.

The set of rationales $\mathcal{V}$ captures the preferences that the agent regards as reasonable.  For instance, the rationales could specify the agent's utility from consuming a good or service.  Alternatively, rationales could specify the agent's subjective beliefs about some payoff-relevant event, with the observed state $s$ being a noisy signal about that event.

We start by describing choice in the second period.  The agent facing decision problem $D$ has chosen $a_1$ from menu $A_1$ and learned that the state is $s$.  She chooses $a_2 \in A_2(a_1,s)$ and $v \in \mathcal{V}$ to maximize
\begin{equation}\label{eq:full_rationalizer}
    U_D(a_2, v \mid a_1, s) \equiv
    (1-\gamma)\underbrace{u(a_1,a_2,s)}_{\substack{\text{material} \\ \text{utility}}} + \gamma  \underbrace{\left[ v(a_1,a_2,s) - \max_{\substack{\hat{a}_1 \in A_1 \\ \hat{a}_2 \in A_2(\hat{a}_1,s)}} v(\hat{a}_1,\hat{a}_2,s)\right]}_{\text{rationalization utility}}
\end{equation}
for parameter $\gamma \in [0,1]$.   Rationalization utility measures how close the agent's course of action is to the \textit{ex post optimum} under her chosen rationale $v$.  When $a_1$ was \textit{ex post} sub-optimal according to $u$, rationalization utility might be increased by adopting rationale $v \neq u$.  This distorts the agent's choice of $a_2$, which maximizes $(1-\gamma) u(a_1,a_2,s) + \gamma v(a_1,a_2,s)$.

We restrict attention to decision problems for which the relevant maxima are well-defined. This is implied, for instance, if every $v \in \mathcal{V}$ is continuous in the actions, and the sets $A_2(a_1,s)$, $\{(a'_1,a'_2): a'_1 \in A_1 \text{ and } a'_2 \in A_2(a'_1) \}$, and $\mathcal{V}$ are compact.

We discuss some natural benchmarks for first-period behavior.  A \textbf{na\"if} chooses $a_1$ to maximize $\mathbb{E}_F[u(a_1,a^*_2(a_1,s),s)]$ where $a^*_2(a_1,s)$ is a selection from
\begin{equation}
    \argmax_{a_2 \in A_2(a_1,s)} u(a_1,a_2,s).
\end{equation}
A \textbf{sophisticate} chooses $a_1$ to maximize $\mathbb{E}_F[u(a_1,\tilde{a}_2(a_1,s),s)]$ where $\tilde{a}_2(a_1,s)$ is a selection from
\begin{equation}\label{eq:soph_obj}
    \argmax_{a_2 \in A_2(a_1,s)} \max_{v \in \mathcal{V}} U_D(a_2,v \mid a_1,s).
\end{equation}
Na\"ifs and sophisticates both maximize expected material utility \textit{ex ante}, albeit with different beliefs about \textit{ex post} behavior. If choice correspondence \eqref{eq:soph_obj} is non-singleton, one may select between personal equilibria as in \cite{kHoszegi2006model}.

Another natural benchmark is the \textbf{empathetic sophisticate}, whose first action maximizes expected total utility, \textit{i.e.}
\begin{equation}
    \mathbb{E}_F\left[\max_{a_2 \in A_2(a_1,s)}\max_{v \in \mathcal{V}} U_D(a_2,v\mid a_1, s)\right].
\end{equation}
This agent both correctly foresees his \textit{ex post} behavior and desires to reduce regret when choosing \textit{ex ante}.

\subsection{Discussion of modeling choices}\label{sec:modeling_choices}

Plausibly, the agent's rationalization motive depends on the kind of \textit{ex ante} uncertainty she faced.  Choosing a risky investment is not like choosing a bet in roulette. It is easier to remember the \textit{ex ante} perspective when evaluating choices with objective risks.  By contrast, people are more likely to say, ``I should have known it!" for decisions that involved subjective uncertainty or required deliberation to weigh competing considerations.  Our model abstracts from this nuance, representing  uncertainty using only a distribution over states.  However, we interpret the scope of the theory to be confined to those kinds of uncertainty which seem predictable in hindsight.\footnote{We suggest that experimental tests of the theory use forms of uncertainty that require the subject to exercise judgment, rather than objective risks such as coin flips or dice rolls. Experimenters might also consider designs that give subjects the illusion of control \citep{langer1975illusion, presson1996illusion}.}

For the theory to depart from the classical prediction, the available rationales $\mathcal{V}$ must be limited.  For instance, if $\mathcal{V}$ includes a `stoic' rationale that is indifferent between all action sequences, then the second term in \eqref{eq:full_rationalizer} can always be set to zero, and the theory predicts material utility maximization.  Thus, the theory's novel predictions depend on plausible restrictions on the rationales that the agent can adopt.

In applying the theory, the rationales $\mathcal{V}$ should capture the preferences that it is psychologically plausible that the decision-maker could hold. Bob can convince himself that he will enjoy the game enough to pay \$100 and drive through the snowstorm; he cannot convince himself that he will get \$100 of pleasure from staying home and literally eating the ticket. These facts about Bob rest on our commonsense understanding of basketball games, snowstorms, and the human diet.

As a rule of thumb, we recommend that the rationales be limited to those that are consistent with the actual choices made by the relevant population under full information. Some fans would pay \$100 to attend a basketball game, even knowing for sure that it takes place during a snowstorm. Virtually no one would pay \$100 to eat a basketball ticket, or would choose to lose \$120 for sure instead of losing \$100 for sure. If nobody would choose according to $v$ with full information about $s$, then it is implausible that the agent can convince himself of rationale $v$ in state $s$.

Even in classical economic analyses, we make judgments about the set of plausible preferences. For instance, in mechanism design, the type space captures the \textit{a priori} plausible limits on the agent's preferences, in the canonical interpretation proposed by \cite{hurwicz1972}. Similar judgments are required to make functional form restrictions for structural models. These judgments are informed partly by introspection and partly by observing the choices made by other people.

Whenever possible, we suggest that the theory should be applied by importing standard preference restrictions from classical models. This serves to prevent \textit{ad hoc} explanations and to make the theory a portable extension of existing models, in the sense of \cite{rabin2013approach}. For instance, it is standard in auction theory to assume that the agent's preferences have the form $\theta \mathbbm{1} - \tau$ for $\theta \in [\underline{\theta}, \overline{\theta}]$, where $\mathbbm{1}$ is an indicator for whether the agent gets the object and $\tau$ is his net transfer. When applying the present theory to auctions, it is natural to adopt this same class of preferences as the set of rationales. Similarly, when studying belief updating, it is often assumed that the agent only cares about his money payments, but may have a range of prior beliefs about the underlying state. This is a natural structure to impose on $\mathcal{V}$, and we illustrate the construction in \Cref{sec:rationales_as_beliefs}. 

Nonetheless, because the theory's predictions depend on $\mathcal{V}$, its general applicability depends on whether $\mathcal{V}$ can be identified from choice behavior. We take up this challenge in \Cref{sec:id}.

\subsection{Rationales as subjective beliefs}\label{sec:rationales_as_beliefs}

What if, rather than adapting her tastes, the agent adapts her beliefs?  Let the states of the world be some finite set $\Omega$ and the agent's payoff in state $\omega \in \Omega$ be $\tilde{v}(a_1,a_2,\omega)$.  After choosing $a_1$, the agent observes a signal $X$ about the state, a random variable with full-support distribution conditional on $\omega$. Then she chooses $a_2$.

The agent adapts her beliefs by distorting her prior on $\Omega$.  Formally, the agent has a set of plausible priors $\Pi \subseteq \Delta \Omega$. One of these is the `material' prior $\pi^*$. For each prior $\pi \in \Pi$ and each signal realization $x$, we define
\begin{equation}
    v_\pi(a_1,a_2,x) \equiv E_\pi\left[ \tilde{v}(a_1,a_2,\omega) \mid X = x \right].
\end{equation}
From this perspective, belief adaptation is a special case of the model, with material utility and available rationales as follows:
\begin{equation}
    u(a_1,a_2,x) \equiv E_{\pi^*}\left[ \tilde{v}(a_1,a_2,\omega) \mid X = x \right],
\end{equation}
\begin{equation}
    \mathcal{V} \equiv \left\{ v_\pi: \pi \in \Pi \right\}.
\end{equation}
On this interpretation, the model yields predictions about how the agent's beliefs respond to her past choices. We explore some of these in \Cref{sec:beliefs}.

\subsection{Rationales as justifications to others}\label{sec:justif_to_others}

We have interpreted the theory as capturing individual psychological motives, but it has another interpretation in the context of organizations. On this interpretation, the theory describes a rational agent who is rewarded for past performance, but can influence the principal's criteria for performance evaluation.  The function $u$ represents the principal's default criterion and $\mathcal{V}$ represents the criteria that the principal would find acceptable. The parameter $\gamma$ captures the agent's degree of influence. The resulting performance criterion is
\begin{equation}
    \rho(a_1,a_2,s) = (1-\gamma)u(a_1,a_2,s) + \gamma v(a_1,a_2,s),
\end{equation}
and the agent's time-$2$ reward is some positive affine transformation of
\begin{equation}\label{eq:perf}
    \rho(a_1,a_2,s) - \max_{\substack{\hat{a}_1 \in A_1 \\ \hat{a}_2 \in A_2(\hat{a}_1,s)}}  \rho(\hat{a}_1,\hat{a}_2,s).
\end{equation}
The actions $a_2$ and rationales $v$ that maximize \eqref{eq:perf} are the same as those that maximize \eqref{eq:full_rationalizer}.
Thus, the agent's need to defend past decisions generates sunk-cost effects.\footnote{\cite{fujino2016neural} find that people who tend to adhere to social rules and regulations are more likely to exhibit the sunk-cost effect.}  This formalizes an observation by \cite{staw1980rationality} about the perverse incentives of retrospective performance evaluation.
\section{Evidence on \textit{ex post} rationalization}\label{sec:evidence}

\subsection{Evidence from psychology}\label{sec:psych_evidence}

We review psychological research pertinent to the present theory. We cover only a fraction of the vast literature on rationalization; as \cite{cushman2020rationalization} writes, ``Among psychologists, [rationalization] is one of the most exhaustively documented and relentlessly maligned acts in the human repertoire." 

The present theory posits that the agent evaluates the past from the \textit{ex post} perspective, as is consistent with the literature on hindsight bias. This literature finds that after an event has occurred, people are overconfident that they could have predicted it in advance \citep{fischhoff1975hindsight,blank2007hindsight}, and even misremember their own \textit{ex ante} predictions, falsely believing that they predicted what came to pass \citep{fischhoff1975knew, fischhoff1977perceived}. Hence, it is plausible that people adopt the \textit{ex post} rather than the \textit{ex ante} perspective when evaluating their past choices.

In the present theory, rather than sticking to her original motives, the agent adopts \textit{post hoc} rationales that justify her actions. The psychology literature on confabulation finds that people generate \textit{post hoc} rationales and sincerely believe them. Confabulation has been documented in split-brain patients \citep{gazzaniga1967split, gazzaniga2005forty} and in ordinary people manipulated to misremember what they chose \citep{johansson2005failure, johansson2006something}. \cite{nisbett1977telling} find that subjects are often unaware of the effects of experimental stimuli on their behavior, and offer spurious explanations when queried by experimenters.

Cognitive dissonance theory posits that people adapt their cognitions so as to achieve internal consistency \citep{festinger1957theory}.  In particular, they adapt their attitudes and beliefs to justify their past choices.  A variety of experiments find evidence for this hypothesis.  When a reading group conducts an unexpectedly humiliating initiation ritual, this causes new members to evaluate the group more positively, a result due to \cite{aronson1959effect} and replicated by \cite{gerard1966effects}.  The mere fact that an alternative was chosen in the past causes lab subjects to evaluate it more positively and causes them to be more likely to choose it in future \citep{arad2013past}.\footnote{This result originates from the experiment of \cite{brehm1956postdecision}, which is confounded by self-selection bias \citep{chen2010choice,risen2010study}. \cite{arad2013past} modified the paradigm to remove the confound.} \cite{harmon2007cognitive} review the substantial experimental evidence on cognitive dissonance.

In summary, the mental mechanisms behind \textit{ex post} rationalization have been studied by psychologists for decades. Our contribution is to formalize these in a tractable economic model.

\subsection{Sunk-cost effects}\label{sec:sunk_costs}

Sunk-cost effects have been documented in many settings, including  business decisions \citep{staw1976knee, mccarthy1993reinvestment,schoorman1988escalation,staw1997escalation}, consumption decisions \citep{arkes1985psychology, ho2018sunk}, professional sports \citep{staw1995sunk,camerer1999econometrics,keefer2017sunk}, and auctions \citep{herrmann2015beating,augenblick2016sunk}.  Some studies find no evidence of sunk-cost effects \citep{ashraf2010can,friedman2007searching,ketel2016tuition,negrini2020still}.\footnote{For a meta-analysis, see \cite{roth2015sunk}}

Despite the intuitive pull of sunk-cost effects and their clear relevance to economic decisions, there is no standard theory of sunk-cost effects with broad scope. Most empirical analyses of sunk costs either do not test a formal theory of sunk-cost effects or use bespoke theories designed for particular contexts.\footnote{Of the studies cited above, \cite{staw1976knee}, \cite{mccarthy1993reinvestment}, \cite{schoorman1988escalation}, \cite{staw1997escalation}, \cite{staw1995sunk}, \cite{camerer1999econometrics}, \cite{keefer2017sunk}, \cite{herrmann2015beating}, \cite{friedman2007searching}, and \cite{ketel2016tuition} do not test a formal theory of sunk costs. \cite{ho2018sunk} and \cite{augenblick2016sunk} test theories specialized to durable goods and dynamic auctions respectively. \cite{arkes1985psychology} and \cite{negrini2020still} test the predictions of prospect theory. \cite{ashraf2010can} test the taste-for-consistency theory of \cite{eyster2002rationalizing}.}

\textit{Ex post} rationalization is a candidate for a workhorse theory.  It predicts that sunk-cost effects will occur when the agent has made an \textit{ex post} mistake that can be rationalized by doubling down on the original course of action.

Sunk-cost effects do not rely on the functional form we chose for \Cref{ex:baseline}. Instead, consider an agent making a risky investment decision to invent a new product.  He chooses investment levels $a_1,a_2 \in [0,1]$.  The project succeeds with probability $\phi(a_1,a_2)$, where $\phi$ is increasing in both arguments, differentiable, and has a positive cross-partial derivative. The agent's rationales involve adapting his belief about the expected profits conditional on success; this is captured by $\theta \in \mathbb{R}$.  His payoff under rationale $v_\theta$ is
\begin{equation}
    \theta \phi(a_1,a_2) - s a_1 - a_2.
\end{equation}
$s$ is a non-negative shock to time-$1$ investment costs. These costs are sunk at time $2$, so a classical agent's time-$2$ choice does not depend on $s$. By contrast, for an \textit{ex post} rationalizer, high time-$1$ cost shocks lead to higher time-$2$ investments.

More broadly, suppose that the agent's rationales involve adapting his value $\theta$ for consuming some good or achieving some goal. In \Cref{sec:comp_stat}, we find that sunk-cost effects are implied by a general comparative statics result, that applies whenever the agent's payoffs are supermodular in $a_1$, $a_2$, and $\theta$.  Supermodularity is a standard condition that captures complementarity between choice variables.

Reputation concerns can give rise to sunk-cost effects. For instance, a project manager might respond to sunk costs so as to convince the market that he is competent \citep{prendergast1996impetuous}, and agents in bilateral relationships might respond to sunk costs so as to improve their reputation for future partners \citep{mcafee2010sunk}. The present model can be seen as a tractable reduced form for reputation concerns, as in \Cref{sec:justif_to_others}.  Its simple linear form abstracts from the details of signaling equilibria, and in return allows the theory to apply to a wider class of decision problems.

Self-signaling under limited memory can lead to sunk-cost effects, because investing at time $1$ despite high costs signals the agent's time-$1$ information to her time-$2$ self \citep{baliga2011mnemonomics, hong2019sunk}.  However, if a cost shock occurs only after the time-$1$ decision, then unexpectedly high costs are unrelated to the agent's time-$1$ information.  Consequently, self-signaling explains some but not all of the data. For instance, \cite{arkes1985psychology} and \cite{guenzel2021too} find sunk-cost effects from cost shocks that postdate the time-$1$ decision.\footnote{\cite{arkes1985psychology} study people who were at the ticket window, had already announced their intention to buy a season ticket, and then were given an unexpected discount. \cite{guenzel2021too} studies shocks to corporate acquisition costs that occur only after the acquisition decision.}

Reputation concerns and memory limitations are real and often important.  However, the evidence of \Cref{sec:psych_evidence} suggests that rationalization is a basic psychological process, operating even in situations where those features are absent.

Reference-dependent preferences are another explanation for sunk-cost effects \citep{thaler1980toward, arkes1985psychology}. Returning to \Cref{ex:baseline}, if Bob buys the ticket and stays home in the snowstorm, then this results in a sure loss relative to his reference point. On the other hand, attending the game in a snowstorm is a gain in some dimensions and a loss in others.  The standard prospect theory value function is convex over losses \citep{kahneman1979prospect}, which can lead Bob to attend the game after buying the ticket, but to stay home if the ticket was free.

\textit{Ex post} rationalization and reference-dependence make different predictions about \Cref{ex:nochoice}, in which Bob receives a ticket and loses \$100 through no choice of his own.  The usual way to close the reference-dependent model is to use the agent's expectations to set the reference point, as in \cite{kHoszegi2006model}.  For any personal equilibrium in which Bob buys the ticket for \$100 in \Cref{ex:baseline}, his reference point is the same as in \Cref{ex:nochoice}. Thus, expectations-based reference dependence (EBRD) is constrained to predict the same time-$2$ behavior regardless of whether Bob made a choice to buy the ticket. By contrast, \textit{ex post} rationalization predicts that Bob's time-$2$ decision depends on whether he could have chosen otherwise at time $1$.

Recent work finds that sunk-cost effects depend on whether the decision-maker was responsible for incurring those costs, as predicted by \textit{ex post} rationalization. \cite{martens2021escalating} study the effect of sunk costs on follow-up investment decisions, in a laboratory experiment. They find that sunk-cost effects increase substantially when subjects are responsible for the initial investment decision. Similarly, \cite{guenzel2021too} studies corporate acquisitions and finds that exogenous acquisition cost shocks (occurring {after} the acquisition decision) decrease the company's willingness to divest, but this effect is much reduced if the CEO who led the acquisition steps down.

\section{Comparative statics for complements}\label{sec:comp_stat}

When do rationalizers react to sunk costs? How does the behavior of rationalizers differ from the classical prediction? To answer these questions, we derive comparative statics for decision problems in which time-$1$ actions, time-$2$ actions, and rationales are complements.

We now assume that the rationales $\mathcal{V}$ have the form $\{w(a_1,a_2,\theta,s) : \theta \in \Theta\}$, for some set $\Theta$ and some function $w$.\footnote{On its own, this assumption is without loss of generality.}  We use $\theta^*$ to denote the parameter value that corresponds to material utility, so $u(a_1,a_2,s) = w(a_1,a_2,\theta^*,s)$. To ease notation, we suppress the dependence of $A_2$ on $s$.\footnote{To extend these results, one replaces the requirement that $A_2(a_1)$ is monotone non-decreasing with the requirement that for each $s$, $A_2(a_1,s)$ is monotone non-decreasing in $a_1$. A similar extension holds for \eqref{eq:strong_nondec}.} Hence, the agent facing some decision problem $D$, having chosen $a_1$ from menu $A_1$ and observed state $s$, chooses $a_2 \in A_2(a_1)$ and $\theta \in \Theta$ to maximize
\begin{equation}
    U_D(a_2,\theta \mid a_1, s) \equiv 
    (1 - \gamma) w(a_1,a_2,\theta^*,s) + \gamma \left[w(a_1,a_2,\theta,s) - \max_{\substack{\hat{a}_1 \in A_1 \\ \hat{a}_2 \in A_2(\hat{a}_1)}} w(\hat{a}_1,\hat{a}_2,\theta,s) \right].
\end{equation}
We assume that $w$ and $U_D$ have non-empty maxima with respect to $(a_1,a_2,\theta)$, and similarly for subsets of these arguments.

We will assume that the choice variables are complements---the marginal return of raising one variable is non-decreasing in the other variables. All our results cover the case of $\mathcal{A}_1$, $\mathcal{A}_2$, and $\Theta$ totally ordered, and $w$ supermodular in $(a_1,a_2,\theta)$. If additionally $w$ is differentiable, then this reduces to the requirement that the cross partial derivatives are all non-negative. However, we state our results under weaker order-theoretic assumptions to expand their scope.

We now state some standard definitions; for more detail see \cite{milgrom1994monotone}. Suppose $X$ and $Y$ are partially ordered sets.  Function $f:X \times Y \times S \rightarrow \mathbb{R}$ has \textbf{increasing differences} between $x$ and $y$ if for all $\tilde{x} \leq \tilde{x}'$, all $\tilde{y} \leq \tilde{y}'$, and all $s$, we have
\begin{equation}
    f(\tilde{x},\tilde{y}',s) - f(\tilde{x},\tilde{y},s) \leq f(\tilde{x}',\tilde{y}',s) - f(\tilde{x}',\tilde{y},s).
\end{equation}
Suppose $X$ is a lattice and $Y$ is an arbitrary set. Function $f:X \times Y \rightarrow \mathbb{R}$ is \textbf{supermodular} in $x$ if for all $\tilde{x}$, all $\tilde{x}'$, and all $y$, we have
\begin{equation}
    f(\tilde{x},y) + f(\tilde{x}',y) \leq f(\tilde{x} \wedge \tilde{x}',y) + f(\tilde{x} \vee \tilde{x}',y).
\end{equation}
Given two lattices $X$ and $Y$, we order $X \times Y$ according to the component-wise order. Given any lattice, we order subsets $X$ and $Y$ with the \textbf{strong set order}, writing $X \ll Y$ if for any $x \in X$ and $y \in Y$, we have $x \wedge y \in X$ and $x \vee y \in Y$.  Given a partially ordered set $X$ and a lattice $Y$, we say that a correspondence $J: X \rightrightarrows Y$ is \textbf{monotone non-decreasing} if $x \leq x'$ implies that $J(x) \ll J(x')$.\footnote{The relation $\leq$ is reflexive, so this implies that for all $x$, $J(x)$ is a sublattice of $Y$.}

In decision problems with complements, Topkis's theorem implies that raising the first action, \textit{ceteris paribus}, raises the classical agent's second action. For rationalizers, we find that raising the first action raises both the second action and the rationale.

\begin{proposition}\label{prop:topkis_too}
Let $\mathcal{A}_1$ be a partially ordered set, let $\mathcal{A}_2$ be a lattice and let $\Theta$ be totally ordered. Suppose that $w$ has increasing differences between $a_1$ and $(a_2,\theta)$ and is supermodular in $(a_2,\theta)$.  Suppose that $A_2(a_1)$ is monotone non-decreasing. For any decision problem $D$ and any state $s$, the correspondence
\begin{equation}\label{eq:topkis_too}
     \argmax_{\substack{a_2 \in A_2({a}_1) \\ \theta \in \Theta}}  U_D(a_2,\theta \mid {a}_1, s)
\end{equation}
is monotone non-decreasing in ${a}_1$.
\end{proposition}
The proof is in \Cref{app:topkis_too}. The conclusion of \Cref{prop:topkis_too} implies that the agent's chosen actions,
\begin{equation}
    \argmax_{a_2 \in A_2({a}_1)} \left\{ \max_{\theta \in \Theta}  U_D(a_2,\theta \mid {a}_1, s) \right\},
\end{equation}
are monotone non-decreasing in ${a}_1$.

The next theorem shows that the theory yields systematic deviations from the classical benchmark. Sunk-cost effects are part of a larger class of phenomena predicted by \textit{ex post} rationalization, that involve distorting later actions upwards when past actions were \textit{ex post} too high, and distorting them downwards when past actions were \textit{ex post} too low.
\begin{theorem}\label{thm:comp_stat}
Let $\mathcal{A}_1$ be a partially ordered set, let $\mathcal{A}_2$ be a lattice and let $\Theta$ be totally ordered. Suppose that $w$ has increasing differences between $a_1$ and $(a_2,\theta)$ and is supermodular in $(a_2,\theta)$. Suppose that $A_2(a_1)$ is monotone non-decreasing.  If the agent's time-$1$ choice was \textit{ex post} weakly higher than optimal, \textit{i.e.} $\bar{a}_1 \geq a^*_1$ for
\begin{equation}
    a^*_1 \in  \argmax_{a_1 \in A_1} \left\{ \max_{a_2 \in A_2(a_1)} w(a_1,a_2,\theta^*,s) \right\},
\end{equation}
then the agent's time-$2$ choice is weakly higher than materially optimal, \textit{i.e.}
\begin{equation}\label{eq:comp_stat_choice}
   \argmax_{a_2 \in A_2(\bar{a}_1)} \left\{ \max_{\theta \in \Theta}  U_D(a_2,\theta \mid \bar{a}_1, s) \right\} \gg  \argmax_{a_2 \in A_2(\bar{a}_1)} w(\bar{a}_1,a_2,\theta^*,s),
\end{equation}
and for any selection $\bar{a}_2$ from the left-hand side of \eqref{eq:comp_stat_choice}, there exists $\bar{\theta} \geq \theta^*$ such that
\begin{equation}\label{eq:implied_rationale}
    \bar{a}_2 \in \argmax_{a_2 \in A_2(\bar{a}_1)} U_D(a_2,\bar{\theta} \mid \bar{a}_1, s).
\end{equation}
Symmetrically, if the agent's time-$1$ choice was \textit{ex post} weakly lower than optimal, then the agent's time-$2$ choice is weakly lower than materially optimal, and there exists $\bar{\theta} \leq \theta^*$ such that $\bar{a}_2 \in \argmax_{a_2 \in A_2(\bar{a}_1)} U_D(a_2,\bar{\theta} \mid \bar{a}_1, s)$.
\end{theorem}
The proof is in \Cref{app:comp_stat}.

\Cref{thm:comp_stat} is a tool to study the choices of \textit{ex post} rationalizers, relying only on standard conditions for monotone comparative statics. Its predictions extend beyond sunk-cost effects; distortions can occur even when each rationale regards the time-$1$ decision and the time-$2$ decision as additively separable, as we illustrate in \Cref{sec:sametwice}.

In addition to predictions about choices, \Cref{thm:comp_stat} makes predictions about rationales. When $a_1$ was \textit{ex post} too high, the agent behaves according to a rationale that is distorted upwards, and when $a_1$ was \textit{ex post} too low, the rationale is distorted downwards. In some experiments, this prediction can be directly tested by asking about the subject's attitudes or beliefs, as is done in the literature on effort justification \citep{aronson1959effect, gerard1966effects}.

Now we consider changing the unchosen alternatives at $t=1$. Observe that for a classical agent, unchosen alternatives from the first menu have no effect on second-period choice. Under complements, the next theorem predicts that increasing the first menu (in the strong set order), while leaving the first choice unchanged, \textit{decreases} the rationalizer's second-period choice as well as her rationale.

\begin{theorem}\label{thm:comp_stat_A1}
Let $\mathcal{A}_1$ and $\mathcal{A}_2$ be lattices and let $\Theta$ be totally ordered. Suppose that $w$ is supermodular in $(a_1,a_2,\theta)$.   Suppose that the correspondence $A_2$ satisfies
    \begin{equation}\label{eq:strong_nondec}
        a_2 \in A_2(a_1) \text{ and } a'_2 \in A_2(a'_1) \Longrightarrow a_2 \wedge a'_2 \in A_2(a_1 \wedge a'_1)  \text{ and } a_2 \vee a'_2 \in A_2(a_1 \vee a'_1).
    \end{equation}
    Take any $A_1, A'_1 \subseteq \mathcal{A}_1$ such that $A_1 \ll A'_1$.  Let $D$ and {\color{red} $D'$} denote the decision problems with $A_1$ and $A'_1$ respectively.  For any $\bar{a}_1 \in A_1 \cap A'_1$ and any $s$ we have
    \begin{equation}\label{eq:compstat_A1}
    \argmax_{\substack{a_2 \in A_2(\bar{a}_1) \\ \theta \in \Theta}} U_{D}(a_2,\theta \mid \bar{a}_1,s) \gg \argmax_{\substack{a_2 \in A_2(\bar{a}_1) \\ \theta \in \Theta}} U_{\color{red} D'}(a_2,\theta \mid \bar{a}_1,s).
    \end{equation}
\end{theorem}
The proof is in \Cref{app:comp_stat_A1}. 

An intuition for \Cref{thm:comp_stat_A1} is that when we add high actions to the first menu, the agent has to rationalize forgoing those actions, so he lowers $\theta$ and hence lowers $a_2$.  Similarly, when we remove low actions from the first menu, the agent no longer has to rationalize forgoing those actions, so he lowers $\theta$ and hence lowers $a_2$. 

Condition \eqref{eq:strong_nondec} is stronger that $A_2$ monotone non-decreasing. It is implied by $\mathcal{A}_1$ totally ordered and $A_2$ monotone non-decreasing. Alternatively, it is implied by constant $A_2$.

Note that \eqref{eq:compstat_A1} implies that the rationalizer's chosen actions decrease, \textit{i.e.}
    \begin{equation}
         \argmax_{a_2 \in A_2(\bar{a}_1)} \left\{ \max_{\theta \in \Theta}  U_D(a_2,\theta \mid \bar{a}_1, s) \right\} \gg \argmax_{a_2 \in A_2(\bar{a}_1)} \left\{ \max_{\theta \in \Theta}  U_{\color{red} D'}(a_2,\theta \mid \bar{a}_1, s) \right\}. 
    \end{equation}
This provides a testable prediction of the theory that does not require us to separately identify the adopted rationales or the materially optimal benchmark. 

\subsection{Applications of results}\label{sec:applications}

We examine some natural decision problems that satisfy the assumptions of \Cref{prop:topkis_too}, \Cref{thm:comp_stat}, and \Cref{thm:comp_stat_A1}.

\subsubsection{The sunk-cost effect for risky investments}
We return to the investment decision of \Cref{sec:sunk_costs}. The agent chooses investment levels $a_1, a_2 \in [0, 1]$.  The project succeeds with probability $\phi(a_1,a_2)$ for continuous supermodular $\phi$. Conditional on success, the project yields an expected profit of $\theta$; the agent's rationales consist in manipulating his beliefs about profit. He pays cost $s a_1 + a_2$, where $s$ is a cost shock for first-period investment.  Hence the agent's payoff is
\begin{equation}
    w(a_1,a_2,\theta) = \theta \phi(a_1,a_2) - s a_1 - a_2.
\end{equation}
The agent chooses $a_1$ before learning $s$, so the materially optimal choice of $a_2$ does not depend on the realized $s$.  \Cref{thm:comp_stat} implies that when $s$ has a high enough realization, so that $a_1$ was \textit{ex post} too high, the rationalizer's choice of $a_2$ is distorted upwards compared to the classical benchmark.  Hence, the time-$2$ investment is higher than the material optimum for projects that turn out to be unexpectedly costly.

If the agent is a publicly listed company, the forecasted profits $\theta$ may also be directly measurable, since firms are required to justify their decisions to shareholders. In that case, \Cref{thm:comp_stat} additionally predicts that high time-$1$ cost shocks cause higher profit forecasts.

\subsubsection{Encountering the same problem twice}\label{sec:sametwice}

The agent faces a decision problem, chooses an action, then learns the state, and then faces the same problem again.  That is, $A_1 = A_2 = A$ and 
\begin{equation}
    w(a_1,a_2,\theta,s) = \phi(a_1,\theta,s) + \phi(a_2,\theta,s),
\end{equation}
for some function $\phi: A \times \Theta \times S \rightarrow \mathbb{R}$ that is supermodular in $(a,\theta)$.  Let $a^*(s)$ be an \textit{ex post} optimal choice in state $s$, \textit{i.e.} $a^*(s) \in \argmax_a \phi(a,\theta^*,s)$.

Material utility is time-separable, so upon learning the state, a classical agent's time-2 choice is $\argmax_a \phi(a,\theta^*,s)$; his second-period choice does not depend on his first-period choice.  By contrast, \Cref{thm:comp_stat} implies that a rationalizer chooses $a_2$ to maximize $(1-\gamma)\phi(a_2,\theta^*,s) + \gamma \phi(a_2,\bar{\theta},s)$, with $\bar{\theta} \geq \theta^*$ when $a_1 \geq a^*(s)$ and $\bar{\theta} \leq  \theta^*$ when $a_1 \leq a^*(s)$. Attempting to rationalize the earlier decision creates a link between otherwise-separate decisions, pulling the rationalizer's second-period choice away from $a^*(s)$ in the direction of his initial choice $a_1$.  Thus, the rationalizer's choice is `stickier' than a classical agent's choice, responding less to learning about the state. This prediction is straightforward to test in laboratory experiments, because a rationalizer who learns $s$ and then faces the problem just once behaves exactly as a classical agent.

\subsubsection{Belief elicitation}\label{sec:beliefs}

In many laboratory experiments, subjects provide point estimates of some quantity, then learn some information, and finally report updated estimates. They are paid for one decision drawn at random, so they encounter the same problem twice, in the sense of \Cref{sec:sametwice}.\footnote{\cite{azrieli2018incentives} study the merits of paying one decision drawn at random.}

We apply the special case of the model with rationales as beliefs, as in \Cref{sec:rationales_as_beliefs}. $Y$ is a real-valued random variable with countable support.\footnote{A parallel construction works if $Y$ has support in some interval and each available rationale is an atomless distribution with strictly positive density.}  The subject makes an incentivized report of $\mathbb{E}[Y]$, then observes a signal $X$ with known conditional distribution $g(x \mid y)$, then makes an incentivized report of $\mathbb{E}[ Y \mid X = x]$.  
The available rationales are priors on $Y$; these are a set of probability mass functions indexed by $\theta$, denoted $(\pi_{\theta})_{\theta \in \Theta}$.  We assume that this set is totally ordered by the monotone likelihood ratio property (MLRP) \citep{milgrom1981good}, that is, for any $\theta > \theta'$ and any $y > y'$
\begin{equation}
    \pi_\theta(y) \pi_{\theta'}(y') >  \pi_{\theta'}(y) \pi_{\theta}(y').
\end{equation}
This restriction is without loss of generality if $Y$ is a Bernoulli random variable, \textit{i.e.} when the agent is being asked to report the probability of some event.  We assume that each $\pi_{\theta}$ has full support, so that no rationale is ruled out by some signal realization.

Given prior $\pi_\theta$ and signal realization $x$, we denote the posterior probability mass function $\pi_{\theta}(y \mid X = x)$.  The agent reports $a_1$, then observes the signal realization, then reports $a_2$. For each report $a_t$, the agent faces quadratic loss (conditional on the signal realization), resulting in the payoff
\begin{equation}
    \phi(a_t,\theta,x) = - \sum_y (a_t - y)^2 \pi_{\theta}(y \mid X = x).
\end{equation}
This captures the interim expected utility of a risk-neutral agent facing a quadratic scoring rule.  It also captures the interim expected utility of an agent with general risk preferences facing an appropriate binarized scoring rule \citep{hossain2013binarized}.

Given the same signal realization, MLRP-ordered priors induce posteriors that are ordered by first-order stochastic dominance \citep{milgrom1981good,klemens2007ordered}.  Thus, if $\theta > \theta'$ then $\pi_{\theta}(y \mid X =x)$ first-order stochastically dominates $\pi_{\theta'}(y \mid X =x)$.  It follows that $\phi$ is supermodular in $(a_t,\theta)$.

Our analysis in \Cref{sec:sametwice} implies that when $a_1 \ge \mathbb{E}[Y \mid X = x]$, then the agent's reported posterior beliefs are distorted upwards, $a_2 \geq \mathbb{E}[Y \mid X = x]$.\footnote{Also in that case, when $a_1 \in \argmax_{a \in A_1} \left\{ \phi (a,\theta ', x) \right\}$ for some $\theta '$, then $a_2 \le a_1$. (If $a_2 > a_1$, then $\phi (a_1,\theta^*, x)>\phi (a_2,\theta^*, x)$, in which case by reducing $a_2$ to $a_1$ and adopting the rationale $\theta'$, the agent can achieve higher material utility and zero rationalization utility, which contradicts the optimality of $a_2$.) This assumption is satisfied whenever the rationales include all full-support priors on $Y$.} Such preference for consistency in belief elicitation is folk wisdom amongst experimenters. \cite{falk2018} find that laboratory subjects report beliefs that are distorted towards their prior reports.

\subsubsection{Consumption under two-part tariffs}

Consider a consumer facing two-part tariffs, each consisting of a lump-sum payment $L$ and a per-unit price $p$ (both non-negative), as in \cite{thaler1980toward}.  The consumer faces a finite list of such tariffs, denoted $(L_k, p_k)_{k \in K}$.  Without loss of generality, we assume that the list contains no dominated tariffs and no duplicates.  We order the tariffs so that $L_1 < L_2 < \cdots$ and $p_1 > p_2 > \cdots$.

The timing is as follows:
\begin{enumerate}
    \item The consumer chooses a tariff $(L_k, p_k)$ from the list or declines.
    \item The consumer learns the taste shock $s \in [0,1]$.
    \item If the consumer chose a tariff, the consumer chooses quantity $q \in [0,1]$.
\end{enumerate}
The set of rationales is indexed by $\Theta = [0,1]$. The consumer's utility from tariff $(L,p)$ under rationale $\theta$ is
$$u(q,s,\theta,L,p) = s\psi(q,\theta) - pq -L$$
where $\psi$ is continuous and supermodular in both arguments, and $\psi(0,\theta) = 0$ for all $\theta$.  The utility from declining is $0$.

Consider a list comprised of tariffs $(L_1,p_1)$ and $(L_2,p_2)$, with $L_1 < L_2$ and $p_1 > p_2$.  By \Cref{prop:topkis_too}, changing the chosen tariff from $(L_1,p_1)$ to $(L_2,p_2)$ weakly raises the quantity consumed, for every realization of the taste shock. 

Observe that for a classical consumer, once a tariff has been chosen, the lump-sum $L$ is sunk and has no effect on the quantity demanded. By contrast, \cite{thaler1980toward} proposes that when a consumer responds to sunk costs, raising the lump-sum payment can increase the quantity demanded.  We formalize this observation in the context of our model.

Take any per-unit price $p$ and taste shock $s$. Take any $L<L'$ such that
\begin{equation}
    \max_{q \in [0,1]} \left\{ s\psi(q,\theta^*) - pq - L \right\} \geq 0 > \max_{q \in [0,1]} \left\{ s\psi(q,\theta^*) - pq - L'\right\}.
\end{equation}
If the consumer was offered only tariff $(L,p)$ and accepted, then in state $s$ this was not an \textit{ex post} mistake, so he demands the materially optimal quantity. By contrast, if he was offered only tariff $(L',p)$ and accepted, then his first choice was \textit{ex post} too high. \Cref{thm:comp_stat} implies that his demand is weakly higher after $(L',p)$ than after $(L,p)$.  Moreover, it is strictly higher for various simple functional forms, such as $\psi(q,\theta) = \theta\sqrt{q}$ with interior $\theta^*$. Thus, high enough lump-sum payments can raise demand compared to the material optimum, provided that they do not cause the consumer to decline the tariff.

Finally, let us take compare two lists, one of which is produced by truncating the other from above.  That is, we have $(L_k,p_k)_{k=1}^K$ and $(L_k,p_k)_{k=1}^{K'}$, for $K < K'$. Suppose we fix the agent's chosen tariff at $(L_j,p_j)$ for $j \leq K$ and switch from the full list $(L_k,p_k)_{k=1}^{K'}$ to the truncated list $(L_k,p_k)_{k=1}^K$. By \Cref{thm:comp_stat_A1}, this change weakly increases demand in every state.  This suggests that a firm selling to rationalizing consumers may find it beneficial to withdraw options with high lump-sums and low marginal prices, especially if those options are seldom chosen.
\section{Representation and identification}\label{sec:id}

Psychologists tend to measure rationalization by asking people directly about their attitudes and beliefs, as in the work surveyed in \Cref{sec:psych_evidence}. The direct approach is useful, but not always feasible. In some situations, a person's rationales may be too complicated to fully articulate. Moreover, the rationale that someone offers to a researcher may be different from the rationale they offer to themselves.

In this section, we instead study a revealed-preference approach to rationalization. If we rely only on behavior, what can we infer about the model parameters? Could it be that two different sets of rationales, $\mathcal{V}$ and $\mathcal{V}'$, nonetheless yield identical choices?

We find that the rationalizer's choice behavior is enough to identify the model. The primitives $u$ and $\mathcal{V}$ are essentially unique, up to a positive affine transformation.

We proceed by making only light assumptions about the permissible rationales. Thus, the identification result applies generally, and does not rely on quasi-linearity in money or on the complementarities assumed in \Cref{sec:comp_stat}. The cost of generality is that our identification procedure is quite abstract. The theorem implies that questions about model primitives can always be settled by choice data, but the best practical approach to elicit such data may vary from case to case. \Cref{app:quasiliner_ID} contains a simple example of identification under stronger assumptions.

\subsection{Decision problems with menus of lotteries}

Up to this point, we have described decision problems in terms of choosing a first action $a_1$ and then a second action $a_2$. However, we can equivalently conceive of the agent first choosing between menus of `outcomes', and then as selecting an outcome from the chosen menu.\footnote{This is structurally similar to the decision problems of \cite{gul2001temptation}, though the theories are different. In \cite{gul2001temptation}, the agent is tempted at time $2$, which affects the time-$1$ choice. Under the present theory, the agent desires to rationalize the time-$1$ choice, which affects the time-$2$ choice.} To illustrate, in \Cref{ex:baseline}, the action \underline{buy a ticket} is equivalent to a menu with two outcomes, \underline{attend the game \& pay \$100} and \underline{stay home \& pay \$100}. The action \underline{don't buy a ticket} is equivalent to the singleton menu consisting of the outcome $\text{\underline{stay home \& pay nothing}}$.\footnote{If the correspondence $A_2$ depends on the state $s$, we can represent this with state-contingent menus of outcomes.}

For the purposes of identification, it saves notation to work with outcomes instead of actions. Of course, if we map action sequences to outcomes and identify the model defined over outcomes, then we also identify the model defined over action sequences.

We start with a finite set of outcomes $Z$.  The lotteries over outcomes are denoted $\Delta(Z)$. For any set $B$, let $\mathcal{K}(B)$ denote the collection of nonempty subsets of $B$. Let $\mathcal{K}_f(B)$ denote the collection of finite nonempty subsets of $B$.

The agent faces decision problems of this form:
\begin{enumerate}
    \item At $t=1$, the agent selects a menu $M$ from a finite collection of menus $\mathcal{M} \subset \mathcal{K}_f(\Delta(Z))$.
    \item The agent learns the state $s$.
    \item At $t = 2$, the agent chooses a lottery $\la$ from the selected menu $M$.
\end{enumerate}

We use $\mathcal{U}$ to denote the set of all functions from $Z$ to $\mathbb{R}$, and extend these to $\Delta(Z)$ by taking expectations.

We now focus on the parameters that yield departures from the classical model. To exhibit rationalizing behavior, it must be that $\gamma \neq 0$, that the rationales include at least two distinct preferences, and that there is no `stoic' rationale that is a constant function. In the definition that follows, we treat each element of $\mathcal{U}$ as a point in $\mathbb{R}^{Z}$. 

\begin{definition}[Rationalization Model]
$(\gamma, u^s, \mathcal{V}^s) \in [0, 1] \times \mathcal{U} \times \mathcal{K}(\mathcal{U})$ is a rationalization model if
\begin{enumerate}
    \item $\gamma \in (0, 1)$;
    \item $\mathcal{V}^s$ is compact and convex, contains representations of at least two distinct preferences, and does not contain a constant utility;
    \item and $u^s \in \mathcal{V}^s$.
\end{enumerate}
\end{definition}

We take as data the agent's time-$2$ choice correspondence in each state. Having selected menu $M$ from collection $\mathcal{M}$ and learned that the state is is $s$, the agent's choices from $M$ are denoted $c_2^s(M \mid \mathcal{M})$, which satisfies $c_2^s(M \mid \mathcal{M}) \subseteq M$.  Observe that the agent can, by choosing differently at time-$1$, achieve any lottery in $\bigcup \mathcal{M}$. 

\begin{definition}[Rationalization Representation]
A rationalization model $(\gamma, u^s, \mathcal{V}^s)$ is a rationalization representation for $c_2^s$ if, for all $(M, \mathcal{M})$, 
\begin{equation}\label{eq:representation}
    c_2^s(M \mid \mathcal{M}) = \argmax_{\la \in M}\left\{(1-\gamma)u^s(\la) + \gamma\max_{v^s \in \mathcal{V}^s}\left\{v^s(\la) - \max_{\hat{\la} \in \bigcup \mathcal{M}} v^s(\hat{\la}) \right\}\right\}.
\end{equation}
\end{definition}

Our present exercise relies only on time-$2$ choice behavior, so the resulting theorems apply equally to na\"ifs, sophisticates, and empathetic sophisticates. We abuse notation and use $c_2^s(M \mid M')$ to denote the agent's choice from $M$ when $\bigcup \mathcal{M} = M'$.

An important caveat is that we observe $c_2^s(M \mid \mathcal{M})$ even if choosing $M$ from $\mathcal{M}$ is sub-optimal from the \textit{ex ante} perspective. Our interpretation of this is that the agent trembles at time $1$ and chooses each menu with some small probability, and in such cases $c_2^s(M \mid \mathcal{M})$ captures how the agent chooses after a slip of the hand, as in \cite{selten1975reexamination}. We discuss how to relax this requirement in \Cref{app:onpath_id}.

\subsection{Measuring total utility with material equivalents}

We now study the agent's time-$2$ choices holding the state $s$ fixed. To reduce clutter, we suppress the superscript $s$ in various notations, including $c_2^s$, $u^s$, and $\mathcal{V}^s$.

To simplify the definitions that follow, we take material utility $u$ as given. That is, we assume the existence of $u \in \mathcal{U}$ such that $c_2(M \mid M) = \argmax_M u$ for all $M$. Material utility is identified by the usual arguments, because the agent maximizes $u$ at time $2$ whenever there are no forgone time-$1$ alternatives.\footnote{The axioms could be stated without reference to $u$, at some cost in clarity.}

A key insight is that we can sometimes measure the agent's total utility by weighing their chosen lottery against a `material equivalent', that is, another lottery that is known to yield the same total utility but has a rationalization utility of $0$.

\Cref{fig:ID_1} depicts preferences over lotteries with three outcomes. The shaded area consists of lotteries that \textit{every} rationale regards as strictly worse than $x$. For a rationalizer, this set can be inferred from choices. It turns out that the shaded area is equivalent\footnote{For a proof that these sets are equivalent, see \Cref{lem:V2worse}.} to the set of lotteries that are \textbf{worse than} $x$, in the following sense:
\begin{definition}
Lottery $r$ is \textbf{worse than} lottery $x$ if for all $\tilde{r}$ close enough to $r$, all $y \in \text{co}(x, \tilde{r})$ distinct from $x$, and all $M, M'$ with $M \subseteq M'$,
\[x \in c_2(M \mid M') \quad \Longleftrightarrow \quad x \in c_2(M \setminus \{y\} \mid M' \setminus \{y\}).\]
\end{definition}

\begin{figure}
\centering
\parbox{7cm}{
\includegraphics[width=7cm]{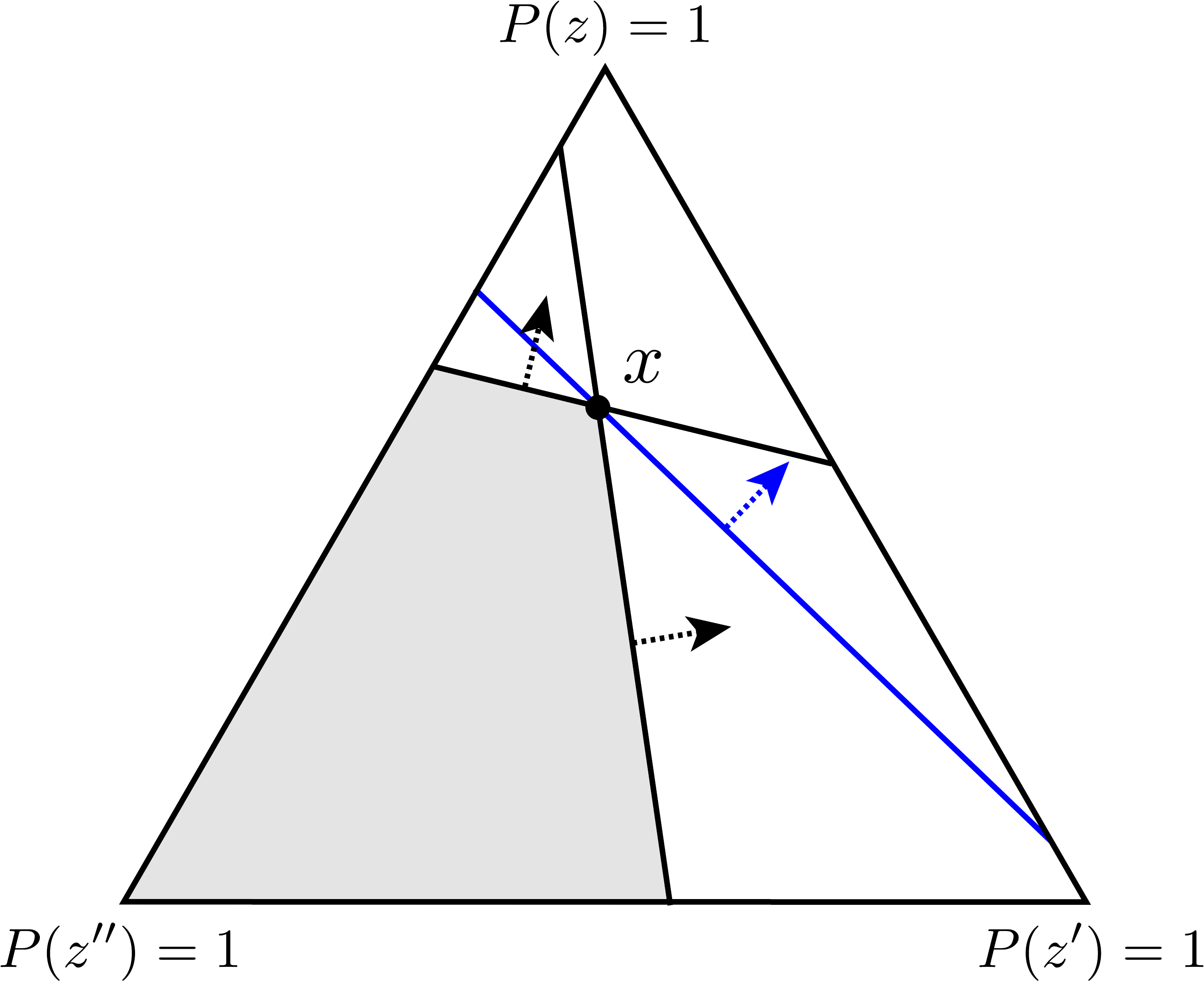}
\caption{Indifference curves for various rationales. Arrows indicate direction of increasing utility. Material indifference depicted in {\color{blue} blue}.}
\label{fig:ID_1}}
\qquad
\begin{minipage}{7cm}
\includegraphics[width=7cm]{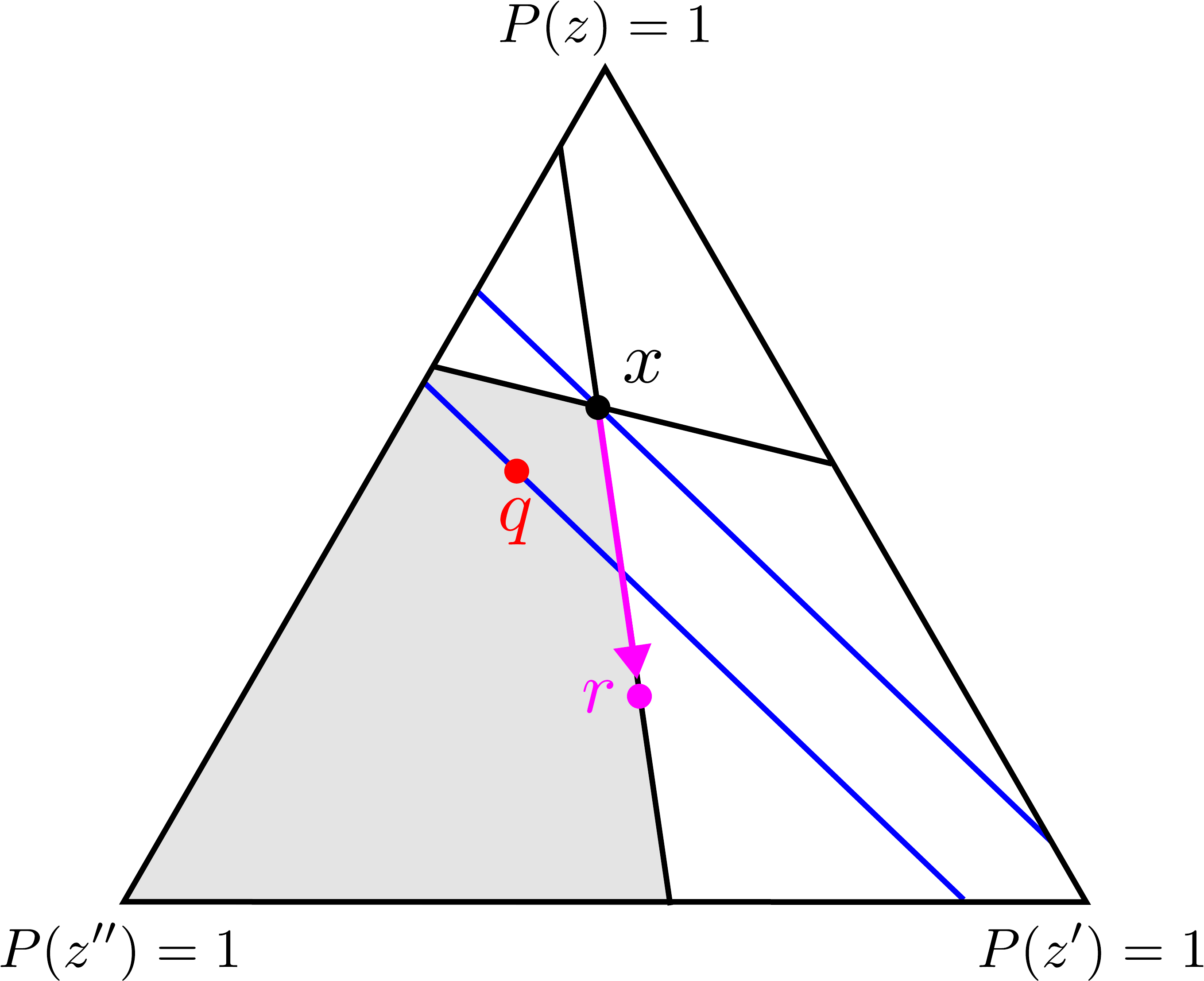}
\caption{Finding a material equivalent of $q$ given $x$. We start with $r = x$ and gradually move $r$ in the marked direction.}
\label{fig:ID_2}
\end{minipage}
\end{figure}

Suppose that the agent faced collection $\left\{ \{q\}, \{x\} \right\}$ at time $1$ and chose menu $\{q\}$. We wish to measure the total utility obtained, in state $s$, from choosing $q$ and forgoing $x$. If at least one rationale weakly prefers $q$ to $x$, then the total utility of $q$ is equal to its material utility. Suppose otherwise, so $q$ is in the shaded region of \Cref{fig:ID_2}.  We add a third lottery $r$. If $r = x$, then $c_2(q,r \mid q, r, x) = \{r\}$. We then move $r$ in the direction marked in \Cref{fig:ID_2}, reducing its material utility. Suppose we find a point such that the agent is indifferent between $q$ and $r$, so $c_2(q,r \mid q, r, x) = \{q, r\}$. Every rationale weakly prefers $x$ to $r$, so $r$ does not alter the rationalization utility of choosing $q$. At least one rationale weakly prefers $r$ to $x$, so if the agent chooses $r$ then rationalization utility is $0$. Thus, we can infer that the \textit{total} utility of $q$ after forgoing $x$ is equal to the \textit{material} utility of $r$. The next definition formalizes this idea.

\begin{definition}\label{def:ME}
The set of \textbf{material equivalents} of $q$ given $x$, denoted $\ME(q \mid x)$, is the set of lotteries that maximize $u$ over
\[\{r \in \text{int}(\Delta(Z)): r \text{ is not worse than any member of } \text{co}(q, x) \text{ and } q \in c_2(q, r \mid q, r, x)\}.\]
When $\ME(q \mid x)$ is nonempty, we define \textbf{revealed utility} $\UE(q \mid x)$ to be the material utility achieved by the members of $\ME(q \mid x)$. 
\end{definition}
For a rationalizer, if $q$ is not worse than $x$, then there exists $v \in \mathcal{V}$ such that $v(q) \geq v(x)$, so in that case revealed utility coincides with material utility, \textit{i.e.} $\UE(q \mid x) = u(q)$. If $q$ is worse than $x$, then we have $\UE(q \mid x) < u(q)$, and rationalization utility is equal to $\UE(q \mid x) - u(q)$.

\subsection{Representation result}

We study the empirical content of rationalization models, taking an axiomatic approach. When the meaning is clear, we omit universal quantifiers. 

\begin{axiom}[Linearity]\label{ax:linearity}
For all $\alpha \in (0, 1)$,
\[c_2(\alpha \{q\} + (1-\alpha)M \mid \alpha \{q\} + (1-\alpha)\mathcal{M}) = \alpha \{q\} + (1-\alpha)c_2(M \mid \mathcal{M}).\]
\end{axiom}

\begin{axiom}[Existence]\label{ax:exis}
There exists an open convex set $S \subset \Delta(Z)$ such that $\ME(q \mid x)$ is nonempty for all $q$ and $x$ in $S$.
\end{axiom}

Consider the ordered pairs of lotteries such that material equivalents exist.  We define these as
\[\compset \equiv \left\{(q, x) \in (\Delta(Z))^2: \ME(q \mid x) \neq \emptyset\right\}.\]

\begin{axiom}[Rationalization]\label{ax:rationalization}
If $M \times \text{co}(\bigcup \mathcal{M}) \subseteq \compset$, then
\[c_2(M \mid \mathcal{M}) = \argmax_{q \in M}\min_{y \in \text{co}(\bigcup \mathcal{M})}\UE(q \mid y).\]
\end{axiom}

\begin{axiom}[Monotonicity]\label{ax:monotonicity}
If $\{(q, x), (q, y)\} \subseteq \compset$, and if $q$ is worse than $x$ and $x$ is worse than $y$, then
\[\UE(q \mid y) < \UE(q \mid x).\]
\end{axiom}

\begin{axiom}[Quasiconvexity]\label{ax:quasiconvexity}
If $\{(q, \alpha x + (1-\alpha)y): \alpha \in [0, 1]\} \subseteq \compset$, then for all $\alpha \in (0, 1)$,
\[\UE(q \mid \alpha x + (1-\alpha)y) \leq \max\left\{\UE(q \mid x), \UE(q \mid y)\right\}.\]
\end{axiom}

\begin{axiom}[Continuity]\label{ax:continuity}
$\UE$ is Lipschitz continuous in both arguments.
\end{axiom}

The next result states that the axioms above fully describe the empirical content of the theory---they are necessary and sufficient for the existence of a rationalization representation. The theorem limits attention to the non-trivial cases, requiring that the agent has some material utility function $u$, but does not always maximize $u$.

\begin{theorem}\label{thm:representation}
Suppose that there exists $u \in \mathcal{U}$ such that
\[c_2(M \mid M) = \argmax_M u\]
for all $M$, but
\[c_2(M \mid \mathcal{M}) \neq \argmax_M u\]
for some $(M, \mathcal{M})$. The following are equivalent:
\begin{enumerate}
    \item $c_2$ satisfies Linearity, Existence, Rationalization, Monotonicity, Quasiconvexity, and Continuity conditional on $u$.
    \item $c_2$ has a rationalization representation $(\gamma, u, \mathcal{V})$. 
\end{enumerate}
\end{theorem}

The proof is in \Cref{app:representation}.

\Cref{thm:representation} sheds light on the bare empirical content of the theory itself. In practice, one expects to apply the model by making psychologically plausible restrictions on the set of rationales $\mathcal{V}$, as we discuss in \Cref{sec:modeling_choices} and illustrate in \Cref{sec:applications}. Nonetheless, \Cref{thm:representation} helps to distinguish between predictions that depend on such restrictions and predictions implied by the formal structure of the theory.

A key feature of the theory is that the agent's \textit{ex post} regret is assessed with respect to his chosen rationale. Given arbitrary rationales, what implications does this have for the objective function? \Cref{thm:representation} answers this explicitly: Revealed utility $\UE(q \mid x)$ captures the agent's total utility of choosing $q$ after forgoing $x$. By construction, we have $\UE(q \mid x) = u(q)$ if $q$ is not worse than $x$, and $\UE(q \mid x) < u(q)$ otherwise. Moreover, the function $\UE$ satisfies Monotonicity, Quasiconvexity, and Continuity. While $\UE$ is only defined on the set $\compset$, Linearity implies that it has a unique extension to all pairs in the simplex. 

Revealed utility $\UE$ is defined for pairwise comparisons, but these comparisons pin down choices in all decision problems. The Rationalization axiom states that when the agent faced some collection $\mathcal{M}$ at time $1$, he evaluates each time-$2$ alternative according to the worst pairwise comparison in $\text{co}(\bigcup \mathcal{M})$. An intuition for this axiom is that we can swap the order of operations when calculating rationalization utility:
\begin{align}
        \max_{v \in \mathcal{V}} \min_{y \in \bigcup \mathcal{M}} \left\{ v(q) - v(y) \right\} &= \max_{v \in \mathcal{V}} \min_{y \in \text{co}(\bigcup \mathcal{M})} \left\{ v(q) - v(y) \right\} & \text{(by linearity of $v$)}\\
        &= \min_{y \in \text{co}(\bigcup \mathcal{M})} \max_{v \in \mathcal{V}} \left\{ v(q) - v(y) \right\}. & \text{(by the minimax theorem)}
\end{align}
Thus, we can study for each pair $(q,y)$ the value $\max_{v \in \mathcal{V}} \left\{ v(q) - v(y) \right\}$, and then take the worst case over all $y \in  \text{co}(\bigcup \mathcal{M})$.

\subsection{Identification result}

In this section, we find that choice behavior identifies the model primitives.

Suppose that we have $v, v' \in \mathcal{V}$, with $v = 2v'$. The availability of $v$ makes no difference to the agent's choices---whenever she adopts rationale $v$, she could weakly increase rationalization utility by switching to $v'$. Thus, we can at best hope to identify $\mathcal{V}$ up to the non-redundant rationales, in the following sense:
\begin{definition}
$v \in \mathcal{V}$ is \textbf{redundant} if $\alpha v + \beta \in \mathcal{V}$ for $\alpha \in (0, 1)$ and $\beta \in \mathbb{R}$. Let $\underline{\mathcal{V}} \equiv \{v \in \mathcal{V}: v \text{ is non-redundant}\}$. 
\end{definition}
The choice correspondence \eqref{eq:representation} is unchanged if we restrict the rationales to $\underline{\mathcal{V}}$.

Suppose that material utility is redundant in model $(\gamma,u,\mathcal{V})$, so $\alpha u + \beta \in \mathcal{V}$ for $\alpha \in (0,1)$ and $\beta \in \mathbb{R}$. Then we can define another model $(\gamma',u',\mathcal{V})$ that yields the same choice behavior, with $\frac{\alpha(1 - \gamma')}{\gamma'} = \frac{1 - \gamma}{ \gamma}$ and $u' = \alpha u + \beta$. Consequently, we normalize material utility to be non-redundant.
\begin{definition}
Rationalization model $(\gamma,u,\mathcal{V})$ is \textbf{canonical} if $u$ is non-redundant.
\end{definition}

The next theorem states that the agent's time-$2$ choice behavior identifies the model. In particular, both material utility and the non-redundant rationales are unique up to a positive affine transformation.

\begin{theorem}\label{thm:id}
If $c_2$ has canonical rationalization representations $(\gamma, u, \mathcal{V})$ and $(\gamma', u', \mathcal{V}')$, then $\gamma = \gamma'$ and there exists $\alpha > 0$ such that
\begin{align}
&\alpha u + \beta = u' \text{ for some } \beta \in \mathbb{R},\\
&v \in \underline{\mathcal{V}} \quad \Longrightarrow \quad \alpha v + \beta \in \underline{\mathcal{V}'} \text{ for some } \beta \in \mathbb{R}.
\end{align}
\end{theorem}

The proof is in \Cref{app:id}.

\Cref{thm:id} establishes that every non-trivial claim about the model parameters can be reduced to a claim about choice behavior. That is, if two rationalization models are not identical (up to the above transformation), then there exist time-$2$ choices that distinguish them. Thus, while many researchers study rationalization by asking people about their reasons for action, in principle the rationalizer's choices can speak for themselves.

Because of its generality, the present exercise is quite abstract. Practical elicitation procedures could exploit structure that is specific to the problem at hand. For instance, one might assume that each rationale is quasi-linear in money. Alternatively, when $Z$ is a set of consumption paths, one might assume that each rationale is time-separable. Usefully, \Cref{thm:id} implies that those assumptions are themselves testable.

\section{Extensions}\label{sec:extensions}

In our model, all uncertainty is resolved before the second action.  If instead the agent observes a random variable $X$ correlated with the state, then it is natural to stipulate that she assesses rationalization utility conditional on the signal realization.  This can be accommodated by transforming the state space and utility functions, so our previous results also apply to decisions with noisy signals.  Namely, given each signal realization $x$ and each rationale $v$, we define
\begin{equation}
    \bar{v}(a_1,a_2,x) \equiv \mathbb{E}_s[v(a_1,a_2,s) \mid X = x],
\end{equation}
and transform the objective function \eqref{eq:full_rationalizer} by substituting $\bar{v}$ for $v$ and $x$ for $s$.

More subtly, the agent's first action could determine not only her payoffs, but also what signal she sees.  In this case, sophisticated rationalizers have a novel motive for information avoidance:  If the sophisticate avoids new information, then \textit{ex ante} optimal actions are also \textit{ex post} optimal, so there is no loss of material utility from rationalizing behavior.

The theory can be extended beyond two periods.  To do so, we stipulate that at each time $t$, the agent chooses some rationale $v$, and compares the expected utility of her strategy under $v$ to the expected utility of the interim-optimal strategy for $v$, with both expectations conditional on the information available at $t$.  At each $t$, she chooses a continuation strategy that maximizes a weighted sum of expected material utility and this generalized rationalization utility.  For sophisticates, we restrict the continuation strategies to be consistent with future rationalizing behavior.

Should we additionally require that the chosen rationales are stable over time?  That is, if the agent acts on Monday, Tuesday, and Wednesday, can she adopt one rationale on Tuesday and then a different rationale on Wednesday?  The theory does not forbid such inconsistent rationalizations, but it does weigh against them, because Tuesday's rationale affects Tuesday's action, which then affects Wednesday's rationale.  If we directly required Tuesday's rationale to be equal to Wednesday's rationale, then the theory's predictions would vary with the addition of `dummy periods' with singleton action menus.  In the interest of invariance, we do not impose this requirement.
\section{Conclusion}

Rationalization is not a new idea; it appears even in classical antiquity. In the Iliad, Odysseus urges the Achaean army to persist at the siege of Troy, arguing:
\begin{quotation}
\noindent
This is the ninth year come round, the ninth \\
we've hung on here. Who could blame the Achaeans \\
for chafing, bridling beside the beaked ships? \\
Ah but still---what a humiliation it would be \\
to hold out so long, then sail home empty-handed. \\
Courage, my friends, hold out a little longer.\\

\hfill Homer, \textit{The Iliad}, trans. \citet[2.345-50]{fagles1990homer}
\end{quotation}

Rationalization is a basic aspect of human behavior, exhaustively documented by psychologists but missing from standard economic models. The present theory bridges that gap. Having spent years writing it, we probably overestimate its contribution.

\bibliographystyle{ecta}
\bibliography{references}

\newpage
\appendix

\section{Proofs of comparative statics results}

\subsection{Proof of Proposition \ref{prop:topkis_too}}\label{app:topkis_too}

In this proof, we fix the state $s$ and the decision problem $D$ with $A_2(a_1)$ monotone non-decreasing, and we suppress the dependence of $w$ on $s$ to reduce notation.  We define $\tilde{U}(a_1,a_2,\theta) \equiv U_D(a_2,\theta \mid a_1,s)$.

\begin{lemma}\label{lem:total_util_incdiff}
If $w$ has increasing differences between $a_1$ and $(a_2,\theta)$, then $\tilde{U}$ has increasing differences between $a_1$ and $(a_2,\theta)$.
\end{lemma}
\begin{proof}
Take any $a'_1 \geq a_1$ and any $(a'_2, \theta') \geq (a_2, \theta)$.  Canceling terms, we have
\begin{equation}\label{eq:CS_incdiff}
\begin{split}
    &\tilde{U}(a'_1,a'_2,\theta') - \tilde{U}(a'_1,a_2,\theta) - \left[\tilde{U}(a_1,a'_2,\theta') - \tilde{U}(a_1,a_2,\theta)\right] \\
    =& (1-\gamma) \left[ w(a'_1,a'_2,\theta^*) - w(a'_1,a_2,\theta^*) - \left[w(a_1,a'_2,\theta^*) - w(a_1,a_2,\theta^*)\right] \right] \\
    & +\gamma \left[w(a'_1,a'_2,\theta') - w(a'_1,a_2,\theta) - \left[ w(a_1,a'_2,\theta') - w(a_1,a_2,\theta)\right] \right].
\end{split}
\end{equation}
By increasing differences for $w$, the term multiplied by $(1-\gamma)$ and the term multiplied by $\gamma$ are both non-negative, so the right-hand side of \eqref{eq:CS_incdiff} is non-negative.
\end{proof}

\begin{lemma}\label{lem:total_util_supermodular}
If $w$ is supermodular in $(a_2,\theta)$, then $\tilde{U}$ is supermodular in $(a_2,\theta)$.
\end{lemma}
\begin{proof}
Take any $(a_2, \theta)$ and $(a'_2, \theta')$.  By substitution, we have
\begin{equation}\label{eq:w_supermod_implies}
\begin{split}
    &\tilde{U}(a_1,(a_2,\theta) \wedge (a'_2,\theta')) + \tilde{U}(a_1, (a_2,\theta) \vee (a'_2,\theta')) \\
    &- \tilde{U}(a_1,a_2,\theta) - \tilde{U}(a_1,a'_2,\theta')  \\
    =& (1-\gamma) [w(a_1,a_2 \wedge a'_2,\theta^*) + w(a_1,a_2 \vee a'_2,\theta^*) \\
    &-  w(a_1,a_2,\theta^*) - w(a_1,a'_2,\theta^*) ]\\
    +& \gamma [w(a_1,(a_2,\theta) \wedge (a'_2,\theta')) + w(a_1,(a_2,\theta) \vee (a'_2,\theta')) \\
    &-  w(a_1,a_2,\theta) - w(a_1,a'_2,\theta') ] - \gamma \Upsilon,
\end{split}
\end{equation}
for
\begin{equation}
\begin{split}
    \Upsilon \equiv &\max_{\substack{\hat{a}_1 \in A_1  \\ \hat{a}_2\in A_2(\hat{a}_1)}} w(\hat{a}_1,\hat{a}_2, \theta \wedge \theta') + \max_{\substack{\hat{a}_1 \in A_1  \\ \hat{a}_2\in A_2(\hat{a}_1)}} w(\hat{a}_1,\hat{a}_2, \theta \vee \theta')\\
    -  &\max_{\substack{\hat{a}_1 \in A_1  \\ \hat{a}_2\in A_2(\hat{a}_1)}} w(\hat{a}_1,\hat{a}_2, \theta) - \max_{\substack{\hat{a}_1 \in A_1  \\ \hat{a}_2\in A_2(\hat{a}_1)}} w(\hat{a}_1,\hat{a}_2, \theta').
\end{split}
\end{equation}
By $\Theta$ totally ordered, $\Upsilon = 0$. Thus, supermodularity of $w$ in $(a_2,\theta)$ implies that the right-hand side of \eqref{eq:w_supermod_implies} is non-negative.
\end{proof}

\Cref{lem:total_util_incdiff}, \Cref{lem:total_util_supermodular}, and Topkis's theorem yield \Cref{prop:topkis_too}. \qed

\subsection{Proof of Theorem \ref{thm:comp_stat}}\label{app:comp_stat}

Suppose $\bar{a}_1$ was \textit{ex post} too high. We prove that if some action $\bar{a}_2$ and some rationale $\bar{\theta}$ maximizes total utility, then there exists $\theta \geq \theta^*$ such that $(\bar{a}_2,\theta)$ maximizes total utility.  Thus, when we are only concerned with actions $a_2$ that maximize total utility, it is without loss of generality to restrict the rationales to be at least $\theta^*$.  An argument using Topkis's theorem then yields \Cref{thm:comp_stat}.

In this proof, we fix the state $s$ and the decision problem $D$ with $A_2(a_1)$ monotone non-decreasing, and we suppress the dependence of $w$ on $s$ to reduce notation.  As before, we define $\tilde{U}(a_1,a_2,\theta) \equiv U_D(a_2,\theta \mid a_1,s)$.

Let $a^*_1$ be as defined in \Cref{thm:comp_stat} and let $\bar{a}_1 \geq a^*_1$. We now define
\begin{equation}
\begin{split}
    & \bar{\Theta} \equiv \argmax_{\theta \in \Theta} \max_{a_2 \in A_2(\bar{a}_1)}  \tilde{U}(\bar{a}_1,a_2,\theta),\\
    & \bar{\Theta}_{\geq} \equiv \left\{\theta \in \bar{\Theta} : \theta \geq \theta^* \right\}.
\end{split}    
\end{equation}
Observe that
\begin{equation}\label{eq:bar_theta_union}
    \argmax_{a_2 \in A_2(\bar{a}_1)} \max_{\theta \in \Theta} \tilde{U}(\bar{a}_1,a_2,\theta) = \bigcup_{\bar{\theta} \in \bar{\Theta}} \argmax_{a_2 \in A_2(\bar{a}_1)} \tilde{U}(\bar{a}_1,a_2,\bar{\theta}).
\end{equation}

\begin{lemma}\label{lem:consider_only_Theta_nonneg}
Under the assumptions of \Cref{thm:comp_stat}, we have
\begin{equation}\label{eq:consider_only_Theta_nonneg}
    \bigcup_{\bar{\theta} \in \bar{\Theta}} \argmax_{a_2 \in A_2(\bar{a}_1)} \tilde{U}(\bar{a}_1,a_2,\bar{\theta}) = \bigcup_{\bar{\theta} \in \bar{\Theta}_\geq} \argmax_{a_2 \in A_2(\bar{a}_1)} \tilde{U}(\bar{a}_1,a_2,\bar{\theta})
\end{equation}
\end{lemma}
\begin{proof}
  Take any
    \begin{equation}\label{eq:upper_optima}
     (\bar{a}_2,\bar{\theta}) \in \argmax_{\substack{a_2 \in A_2(\bar{a}_1) \\ \theta \in \Theta}} \tilde{U}(\bar{a}_1,a_2,{\theta}).
    \end{equation}
    We will show that
    \begin{equation}
     (\bar{a}_2,\bar{\theta}\vee\theta^*) \in \argmax_{\substack{a_2 \in A_2(\bar{a}_1) \\ \theta \in \Theta}} \tilde{U}(\bar{a}_1,a_2,{\theta}).
    \end{equation}
    Let us take any $a_2^* \in \argmax_{a_2 \in A_2(a_1^*)} w(a_1^*,a_2,\theta^*)$. By $a_1^*$ \textit{ex post} optimal, we have
    \begin{equation}\label{eq:lower_optima}
     \tilde{U}(a_1^*,a_2^*,\theta^*) =(1-\gamma) \max_{\substack{a_1 \in A_1 \\ a_2 \in A_2(a_1)}} w(a_1,a_2,\theta^*) = \max_{\substack{a_2 \in A_2(a_1^*) \\ \theta \in \Theta}} \tilde{U}(a_1^*,a_2,{\theta}).
    \end{equation}
    By $a_1^* \leq \bar{a}_1$, \eqref{eq:upper_optima}, \eqref{eq:lower_optima}, and \Cref{prop:topkis_too}, we have
    \begin{equation}\label{eq:meet_optima}
     (a_2^* \wedge \bar{a}_2 ,\theta^* \wedge \bar{\theta}) \in \argmax_{\substack{a_2 \in A_2(a_1^*) \\ \theta \in \Theta}} \tilde{U}(a_1^*,a_2,{\theta}).
    \end{equation}
    By \eqref{eq:meet_optima} and then \eqref{eq:lower_optima}, we have
    \begin{equation}
    \begin{split}
        (1-\gamma) w(a_1^*,a_2^* \wedge \bar{a}_2,\theta^*) & \geq \tilde{U}(a_1^*,a_2^* \wedge \bar{a}_2,\theta^* \wedge \bar{\theta}) \\ \geq \tilde{U}(a_1^*,a_2^*,\theta^*) & = (1-\gamma) \max_{\substack{a_1 \in A_1 \\ a_2 \in A_2(a_1)}}  w(a_1,a_2,\theta^*).
    \end{split}
    \end{equation}
    This implies that the action sequence $(a_1^*, a_2^* \wedge \bar{a}_2)$ yields no regret under rationale $\theta^*$, and thus
    \begin{equation}\label{eq:other_lower_optima}
     (a_2^* \wedge \bar{a}_2,\theta^*) \in \argmax_{\substack{a_2 \in A_2(a_1^*) \\ \theta \in \Theta}} \tilde{U}(a_1^*,a_2,{\theta}).
    \end{equation}
    By $a_1^* \leq \bar{a}_1$, \eqref{eq:upper_optima}, \eqref{eq:other_lower_optima}, and \Cref{prop:topkis_too}, we have
    \begin{equation}
     (\bar{a}_2,\bar{\theta} \vee \theta^*) = (\bar{a}_2,\bar{\theta}) \vee (a_2^* \wedge \bar{a}_2,\theta^*) \in \argmax_{\substack{a_2 \in A_2(\bar{a}_1) \\ \theta \in \Theta}} \tilde{U}(\bar{a}_1,a_2,{\theta}).
    \end{equation}
    Our argument holds for any $(\bar{a}_2,\bar{\theta})$ satisfying \eqref{eq:upper_optima}, which yields \eqref{eq:consider_only_Theta_nonneg}.
\end{proof}

\begin{lemma}\label{lem:actions_monotone_plus} Under the assumptions of \Cref{thm:comp_stat}, we have
\begin{equation}\label{eq:actions_monotone_plus}
        \bigcup_{\bar{\theta} \in \bar{\Theta}_\geq} \argmax_{a_2 \in A_2(\bar{a}_1)} \tilde{U}(\bar{a}_1,a_2,\bar{\theta}) \gg \argmax_{a_2\in A_2(\bar{a}_1)} w(\bar{a}_1,a_2,\theta^*).
\end{equation}
\end{lemma}
\begin{proof}
 Take any $\bar{\theta} \in \bar{\Theta}_\geq$. We define
 \begin{equation}
    g_{\bar{\theta}}(a_2,\gamma) \equiv (1 - \gamma) w(\bar{a}_1,a_2,\theta^*) + \gamma w(\bar{a}_1,a_2,\bar{\theta}).
\end{equation}
By $\bar{\theta} \geq \theta^*$ and $w$ supermodular, $g_{\bar{\theta}}(a_2,\gamma)$ is supermodular.  Thus, by Topkis's theorem, $\argmax_{a_2 \in A_2(\bar{a}_1)} g_{\bar{\theta}}(a_2,\gamma)$ is monotone non-decreasing in $\gamma$.  Thus we have,
\begin{equation}\label{eq:bar_theta_compare}
\begin{split}
    \argmax_{a_2 \in A_2(\bar{a}_1)} \tilde{U}(\bar{a}_1,a_2,\bar{\theta}) = \argmax_{a_2 \in A_2(\bar{a}_1)} g_{\bar{\theta}}(a_2,\gamma) &\gg \argmax_{a_2 \in A_2(\bar{a}_1)} g_{\bar{\theta}}(a_2,0) \\ &= \argmax_{a_2\in A_2(\bar{a}_1)} w(\bar{a}_1,a_2,\theta^*).
    \end{split}
\end{equation} \eqref{eq:bar_theta_compare} holds for all $\bar{\theta} \in \bar{\Theta}_\geq$, which implies \eqref{eq:actions_monotone_plus}.
\end{proof}

By \eqref{eq:bar_theta_union}, \Cref{lem:consider_only_Theta_nonneg}, and \Cref{lem:actions_monotone_plus}, we have
\begin{equation}\label{eq:comp_stat_compare}
    \argmax_{a_2 \in A_2(\bar{a}_1)} \max_{\theta \in \Theta} \tilde{U}(\bar{a}_1,a_2,\theta) \gg \argmax_{a_2\in A_2(\bar{a}_1)} w(\bar{a}_1,a_2,\theta^*).
\end{equation}
Moreover, by \Cref{lem:consider_only_Theta_nonneg}, for any $a_2$ in the left-hand side of \eqref{eq:comp_stat_compare}, there exists $\bar{\theta} \geq \theta^*$ such that
\begin{equation}
    \bar{a}_2 \in \argmax_{a_2 \in A_2(\bar{a}_1)} \tilde{U}(\bar{a}_1,a_2,\bar{\theta}).
\end{equation}
which completes the proof of \Cref{thm:comp_stat} for the case $\bar{a}_1 \geq a^*_1$.

Finally, note that if $w$ has increasing differences and is supermodular with respect to some order relations on $\mathcal{A}_1$, $\mathcal{A}_2$, and $\Theta$, then it has increasing differences and is supermodular with respect to the inverse orders on $\mathcal{A}_1$, $\mathcal{A}_2$, and $\Theta$.  Similarly, if $A_2(a_1)$ is monotone non-decreasing with respect to some order relations on $\mathcal{A}_1$ and $\mathcal{A}_2$, then it is also monotone non-decreasing with respect to the inverse orders. Thus, our proof covers the case $\bar{a}_1 \leq a^*_1$.  \qed

\subsection{Proof of Theorem \ref{thm:comp_stat_A1}}\label{app:comp_stat_A1}

In this proof, we fix the state $s$ and suppress the dependence of $U$ and $w$ on $s$ to reduce notation. 

We first establish that switching from $D'$ to $D$ increases the difference from raising $a_2$ and $\theta$.

\begin{lemma}\label{lem:incdiff_A1}
Under the assumptions of \Cref{thm:comp_stat_A1}, for any $(\bar{a}_2,\bar{\theta}) \leq (\bar{a}'_2,\bar{\theta}')$, we have
\begin{equation}\label{eq:incdiff_A1}
    U_{D'}(\bar{a}'_2,\bar{\theta}' \mid \bar{a}_1) - U_{D'}(\bar{a}_2,\bar{\theta} \mid \bar{a}_1) \leq U_D(\bar{a}'_2,\bar{\theta}' \mid \bar{a}_1) - U_D(\bar{a}_2,\bar{\theta} \mid \bar{a}_1).
\end{equation}
\end{lemma}
\begin{proof}
  For $\gamma = 0$, we have \eqref{eq:incdiff_A1} trivially. Otherwise, by substitution and then some algebra, \eqref{eq:incdiff_A1} reduces to
  \begin{align}
      & \max_{\substack{a_1 \in A'_1  \\ a_2\in A_2(a_1)}} w(a_1,a_2,\bar{\theta}) + \max_{\substack{a_1 \in A_1  \\ a_2\in A_2(a_1)}} w(a_1,a_2,\bar{\theta}') \\ \leq & \max_{\substack{a_1 \in A_1  \\ a_2\in A_2(a_1)}} w(a_1,a_2,\bar{\theta}) + \max_{\substack{a_1 \in A'_1  \\ a_2\in A_2(a_1)}} w(a_1,a_2,\bar{\theta}').
  \end{align}
  Let us define
  \begin{equation}
      (\hat{a}_1,\hat{a}_2) \in \argmax_{\substack{a_1 \in A'_1  \\ a_2\in A_2(a_1)}} w(a_1,a_2,\bar{\theta}),
  \end{equation}
\begin{equation}
      (\tilde{a}_1,\tilde{a}_2) \in \argmax_{\substack{a_1 \in A_1  \\ a_2\in A_2(a_1)}} w(a_1,a_2,\bar{\theta}'),
  \end{equation}
By $A_1 \ll A'_1$, we have $\hat{a}_1 \wedge \tilde{a}_1 \in A_1$ and $\hat{a}_1 \vee \tilde{a}_1 \in A'_1$.  By \eqref{eq:strong_nondec}, we have $\hat{a}_2 \wedge \tilde{a}_2 \in A_2(\hat{a}_1 \wedge \tilde{a}_1)$ and $\hat{a}_2 \vee \tilde{a}_2 \in A_2(\hat{a}_1 \vee \tilde{a}_1)$. Thus, we have
\begin{equation}
    w(\hat{a}_1 \wedge \tilde{a}_1,\hat{a}_2 \wedge \tilde{a}_2,\bar{\theta}) \leq \max_{\substack{a_1 \in A_1  \\ a_2\in A_2(a_1)}} w(a_1,a_2,\bar{\theta})
\end{equation}
\begin{equation}
    w(\hat{a}_1 \vee \tilde{a}_1,\hat{a}_2 \vee \tilde{a}_2,\bar{\theta}') \leq \max_{\substack{a_1 \in A'_1  \\ a_2\in A_2(a_1)}} w(a_1,a_2,\bar{\theta}'),
\end{equation}
Combining inequalities yields
  \begin{align}
       \max_{\substack{a_1 \in A'_1  \\ a_2\in A_2(a_1)}} w(a_1,a_2,\bar{\theta}) &+ \max_{\substack{a_1 \in A_1  \\ a_2\in A_2(a_1)}} w(a_1,a_2,\bar{\theta}') &
      \\ = w(\hat{a}_1,\hat{a}_2,\bar{\theta}) &+ w(\tilde{a}_1,\tilde{a}_2,\bar{\theta}') &
      \\ \leq w(\hat{a}_1\wedge \tilde{a}_1,\hat{a}_2 \wedge \tilde{a}_2,\bar{\theta})&+ w(\hat{a}_1\vee \tilde{a}_1,\hat{a}_2 \vee \tilde{a}_2,\bar{\theta}')& \text{ by supermodularity}
      \\ \leq \max_{\substack{a_1 \in A_1  \\ a_2\in A_2(a_1)}} w(a_1,a_2,\bar{\theta}) &+ \max_{\substack{a_1 \in A'_1  \\ a_2\in A_2(a_1)}} w(a_1,a_2,\bar{\theta}'),&
  \end{align}
which implies \eqref{eq:incdiff_A1}.
\end{proof}

By \Cref{lem:total_util_supermodular}, $U$ is supermodular in $(a_2,\theta)$.  Thus, by \Cref{lem:incdiff_A1} and Topkis's theorem, we have \Cref{thm:comp_stat_A1}. \qed
\section{Proof of Theorem \ref{thm:representation}}\label{app:representation}

Since we will only compare choices within a given state, we drop the superscript $s$ throughout. 

Let $\mathcal{U}_\text{pref}$ denote the set of expected-utility preferences on $\Delta(Z)$. 

\subsection{Necessity}

\begin{lemma}
\label{lem:closeworse_rep}
Suppose $c_2$ has a rationalization representation $(\gamma, u, \mathcal{V})$. For any interior $q$, there exist $p$ and $r$ arbitrarily close to $q$ such that $v(p) > v(q) > v(r)$ for all $v \in \mathcal{V}$. 
\end{lemma}

\begin{proof}
Suppose there do not exist any $x, y$ such that $v(x) > v(y)$ for all $v \in \mathcal{V}$. Fix any interior $q \in \Delta(Z)$ and any menu $\mathcal{M}$ such that $q \in \text{int}(\text{co}(\mathcal{M}))$. By the minimax theorem,
\begin{align}
U(q \mid \mathcal{M}) &= (1-\gamma)u(q) + \min_{x \in \text{co}(\mathcal{M})}\max_{v \in \mathcal{V}}\left(v(q) - v(x) \right)\\
&= (1-\gamma)u(q)\\
&= (1-\gamma)u(q) + \gamma\max_{v \in \mathcal{V}}\left(v(q) - \max_{x \in \mathcal{M}}v \right).
\end{align}
The final equality implies the existence of $v \in \mathcal{V}$ such that $v(q) = \max_{x \in \mathcal{M}}v$. Since $q \in \text{int}(\mathcal{M})$, only constant utilities can satisfy this condition. But $\mathcal{V}$ does not contain a constant utility. Conclude that there exist $x, y$ such that $v(x) > v(y)$ for all $v \in \mathcal{V}$. 

Since $q$ is interior, $q + \epsilon(x - y)$ and $q + \epsilon(y-x)$ exist for all $\epsilon > 0$ sufficiently small. For all such $\epsilon$ and all $v \in \mathcal{V}$, we have $v(q) < v(q + \epsilon(x - y))$ and $v(q) > v(q + \epsilon(y - x))$.
\end{proof}

\begin{lemma}
\label{lem:V2worse}
Suppose $c_2$ has a rationalization representation. For every rationalization representation $(\gamma, u, \mathcal{V})$, for every interior $q$,
\[\{r \in \Delta(Z): r \text{ is worse than } q\} = \{r \in \Delta(Z): v(q) > v(r) \text{ for all } v \in \mathcal{V}\}.\]
\end{lemma}

\begin{proof}
\textbf{First direction:} Suppose that $v(q) > v(r)$ for all $v \in \mathcal{V}$. Since $\mathcal{V}$ is compact, there exists $\epsilon > 0$ such that $v(q) > v(\tilde{r})$ for all $v \in \mathcal{V}$ and all $\tilde{r} \in B_\epsilon(r)$. Fix any $(M, \mathcal{M})$ such that $q \in M$ and any $y \in \text{co}(B_\epsilon(r) \cup \{q\})$. For each $v \in \mathcal{V}$, we have
\[\max_{p \in \mathcal{M}}v = \max_{p \in \mathcal{M} \setminus \{y\}}v.\]
For all $x \in M \setminus \{y\}$, we have
\begin{align}
\nonumber
U(x \mid \mathcal{M}) &= (1-\gamma)u(x) + \gamma\max_{v \in \mathcal{V}}\left(v(x) - \max_{p \in \mathcal{M}}v \right)\\
\nonumber
 &= (1-\gamma)u(x) + \gamma\max_{v \in \mathcal{V}}\left(v(x) - \max_{p \in \mathcal{M} \setminus \{y\}}v\right) \\
 \label{eq:minusy}
 &= U(x \mid \mathcal{M} \setminus \{y\}).
\end{align}
If $q \in c_2(M \mid \mathcal{M})$, we have that $U(q \mid \mathcal{M}) \geq U(x \mid \mathcal{M})$ for all $x \in M$. By (\ref{eq:minusy}), $U(q \mid \mathcal{M} \setminus \{y\}) \geq U(x \mid \mathcal{M} \setminus \{y\})$ for all $x \in M \setminus \{y\}$. This implies $q \in c_2(M \setminus \{y\} \mid \mathcal{M} \setminus \{y\})$. If $q \in c_2(M \setminus \{y\} \mid \mathcal{M} \setminus \{y\})$, we have that $U(q \mid \mathcal{M} \setminus \{y\}) \geq U(x \mid \mathcal{M} \setminus \{y\})$ for all $x \in M \setminus \{y\}$. By (\ref{eq:minusy}), $U(q \mid \mathcal{M}) \geq U(x \mid \mathcal{M})$ for all $x \in M \setminus \{y\}$. To conclude that $q \in c_2(M \mid \mathcal{M})$, it suffices to show that $U(q \mid \mathcal{M}) \geq U(y \mid \mathcal{M})$. Since $u \in \mathcal{V}$, we have that $u(y) < u(q)$. Fix any $v_y \in \argmax_{v \in \mathcal{V}}\left(v(y) - \max_{p \in \mathcal{M}}v\right)$. We have
\begin{align}
U(y \mid \mathcal{M}) = (1-\gamma)u(y) + \gamma\left(v_y(y) - \max_{p \in \mathcal{M}}v_y\right)
&< (1-\gamma)u(q) + \gamma\left(v_y(q) - \max_{p \in \mathcal{M}}v_y \right) \\
&\leq (1-\gamma)u(q) + \gamma \max_{v \in \mathcal{V}}\left(v(q) - \max_{p \in \mathcal{M}}v \right)\\
&= U(q \mid \mathcal{M})
\end{align}
as desired. Conclude that $r$ is worse than $q$.

\textbf{Second direction:} Suppose that $v^*(r) \geq v^*(q)$ for some $v^* \in \mathcal{V}$. We show that $r$ is not worse than $q$. Suppose that $u(r) \geq u(q)$. Then, for $\tilde{r}$ arbitrarily close to $r$, we have $q \in c_2(q \mid q)$ but $q \notin c_2(q, \tilde{r} \mid q, \tilde{r})$. Thus, $r$ is not worse than $q$. For the rest of the proof, we assume that $u(r) < u(q)$. 

Suppose first that $v^*(r) > v^*(q)$. By Lemma \ref{lem:closeworse_rep}, there exist $\bar{q}$ and $\underline{q}$ such that $v(\bar{q}) > v(q) > v(\underline{q})$ for all $v \in \mathcal{V}$. For each $\epsilon \in (0, 1]$, let $\bar{q}_\epsilon \equiv \epsilon \bar{q} + (1-\epsilon)q$. For some $\epsilon_1 > 0$, for all $\epsilon \in (0, \epsilon_1]$, we have $v^*(r) > v^*(\bar{q}_\epsilon)$. For all $\epsilon \in (0, \epsilon_1]$, let $\lambda_\epsilon \equiv (v^*(\bar{q}_\epsilon) - v^*(q))/(v^*(r) - v^*(\bar{q}_\epsilon)) + \epsilon$. There exists $\epsilon_2 \in (0, \epsilon_1]$ such that $y_\epsilon \equiv (1+\lambda_\epsilon)\bar{q}_\epsilon - \lambda_\epsilon r \in \text{int}(\Delta(Z))$ for all $\epsilon \in (0, \epsilon_2]$. Since $\lim_{\epsilon \to 0} y_\epsilon = q$ and since $v(q) > v(\underline{q})$ for all $v \in \mathcal{V}$, there exists $\epsilon_3 \in (0, \epsilon_2]$ such that $v(y_{\epsilon_3}) > v(\underline{q})$ for all $v \in \mathcal{V}$. Let $y \equiv y_{\epsilon_3}$. By construction, $v^*(r) > v^*(q) > v^*(y)$. 

Since $u \in \mathcal{V}$, we have $u(q) > u(\underline{q})$. Since $v^*(q) > \max\{v^*(y), v^*(\underline{q})\}$, we have that
\[U(q \mid q, y, \underline{q}) = (1-\gamma)u(q) > (1-\gamma)u(\underline{q}) > U(\underline{q} \mid q, y, \underline{q}).\]
Since $u(\bar{q}_\epsilon) > u(q) > u(r)$, and since $\bar{q}_\epsilon \in \text{co}(y, r)$, we have that $u(y) > u(q)$. Thus,
\[U(q \mid q, y).= (1-\gamma)u(q) < (1-\gamma)u(y) = U(y \mid q, y).\]
By continuity of $U$, there exists $p \in \text{int}(\text{co}(y, \underline{q}))$ such that
\[U(q \mid q, y, p) = U(p \mid q, y, p).\]
Since $v(y) > v(\underline{q})$ for all $v \in \mathcal{V}$ and since $p \in \text{int}(\text{co}(y, \underline{q}))$, we have that $v(y) > v(p)$ for all $v \in \mathcal{V}$. This implies
\[(1-\gamma)u(q) = U(q \mid q, y, p) = U(p \mid q, y, p) < (1-\gamma)u(p),\]
so $u(p) > u(q)$. 

By Lemma \ref{lem:closeworse_rep} and interiority of $y$, there exists $\bar{y}$ arbitrarily close to $y$ such that $v(\bar{y}) > v(y)$ for all $v \in \mathcal{V}$. Since $v^*(q) > v^*(y)$, we can require that $v^*(q) > v^*(\bar{y})$ by choosing $\bar{y}$ sufficiently close to $y$. Fix any 
\[\alpha \in \argmin_{\tilde{\alpha} \in [0, 1]}\max_{v \in \mathcal{V}}\left(v(p) - v\left(\tilde{\alpha} y + (1-\tilde{\alpha})q\right) \right).\] Since $u(p) > u(q)$ and since $v(y) > v(p)$ for all $v \in \mathcal{V}$, we have that
\[\max_{v \in \mathcal{V}}\left(v(p) - v(q)\right) = 0 > \max_{v \in \mathcal{V}}\left(v(p) - v(y)\right).\]
This implies $\alpha > 0$. Since $v(\bar{y}) > v(y)$ for all $v \in \mathcal{V}$, we have
\[v(p) - v(\alpha \bar{y} +(1-\alpha)q) < v(p) - v(\alpha y + (1-\alpha)q)\]
for all $v \in \mathcal{V}$. This implies
\begin{equation}
\label{eq:yvsybar}
\max_{v \in \mathcal{V}}\left(v(p) - v(\alpha \bar{y} + (1-\alpha)q) \right) < \max_{v \in \mathcal{V}}\left(v(p) - v(\alpha y + (1-\alpha)q) \right).
\end{equation}
Using (\ref{eq:yvsybar}) and the definition of $\alpha$, we have
\begin{align}
\frac{1}{\gamma}\left(U(p \mid q, \bar{y}, p) - (1-\gamma)u(p) \right) &= \min_{x \in \text{co}(q, \bar{y})}\max_{v \in \mathcal{V}}\left(v(p) - v(x) \right) \\
&\leq \max_{v \in \mathcal{V}}\left(v(p) - v\left(\alpha \bar{y} + (1-\alpha)q\right) \right) \\
&< \max_{v \in \mathcal{V}}\left(v(p) - v\left(\alpha y + (1-\alpha)q\right)\right) \\
&= \min_{x \in \text{co}(q, y)}\max_{v \in \mathcal{V}}\left(v(p) - v(x)\right) \\
&= \frac{1}{\gamma}\left(U(p \mid q, y, p) - (1-\gamma)u(p) \right)\\
&\Longrightarrow \quad U(p \mid q, \bar{y}, p) < U(p \mid q, y, p).
\end{align}
Since
\[U(p \mid q, \bar{y}, p) < U(p \mid q, y, p) < U(p \mid q, p),\]
there exists $y^* \in \text{int}(\text{co}(q, \bar{y}))$ such that
\[U(p \mid q, y^*, p) = U(p \mid q, y, p).\]
For any $x^* \in \argmin_{x \in \text{co}(q, y^*)}U(p \mid p, x)$, we have that
\[U(p \mid q, y^*, p) = U(p \mid p, x^*) = U(p \mid q, x^*, p).\]
Thus, it is without loss to assume that 
\[U(p \mid p, y^*) = U(p \mid q, y^*, p).\]

Let $L \equiv \{x \in \Delta(Z): U(p \mid p, x) < U(p \mid p, y^*)\}$. We show that $L$ is convex. It suffices to show that for any $\alpha \in (0, 1)$ and any $x_1$ and $x_2$ such that $\max\{U(p \mid p, x_1), U(p \mid p, x_2)\} < (1-\gamma)u(p)$, 
\[U(p \mid p, \alpha x_1 + (1-\alpha)x_2) \leq \max\left\{U(p \mid p, x_1), U(p \mid p, x_2) \right\}.\]
Since $U(p \mid p, x_i) < (1-\gamma)u(p)$ for $i = 1, 2$, we have that $v(p) < v(x_i)$ for all $v \in \mathcal{V}$ for $i = 1, 2$. Thus, $v(p) < v(\alpha x_1 + (1-\alpha)x_2)$ for all $v \in \mathcal{V}$. We have
\begin{align}
U(p \mid \alpha x_1 + (1-\alpha)x_2) &= (1-\gamma)u(p) + \gamma \max_{v \in \mathcal{V}}\left(v(p) - v(\alpha x_1 + (1-\alpha)x_2) \right)\\
&\leq (1-\gamma)u(p) + \gamma \left(\alpha \max_{v \in \mathcal{V}}\left(v(p) - v(x_1)\right) + (1-\alpha)\max_{v \in \mathcal{V}}\left(v(p) - v(x_2)\right) \right) \\
&= \alpha U(p \mid p, x_1) + (1-\alpha)u(p \mid p, x_2)\\
&\leq \max\left\{U(p \mid p, x_1), U(p \mid p, x_2)\right\}.
\end{align}
Conclude that $L$ is convex.

Since
\[U(p \mid q, \bar{y}, p) < U(p \mid q, y^*, p) = U(p \mid p, y^*),\]
and since $y^* \in \text{int}(\text{co}(\bar{y}, q))$, there exists $\hat{y} \in \text{co}(y^*, \bar{y}) \setminus \{y^*\}$ such that $\hat{y} \in L$. By continuity of $U$, there exists $\epsilon_1 > 0$ such that $\tilde{y} \in L$ for all $\tilde{y} \in B_{\epsilon_1}(\hat{y})$. Since $y^* \in \text{int}(\text{co}(q, \hat{y}))$, there exists $\epsilon_2 > 0$ such that any member of $\text{co}(B_{\epsilon_2}(q) \cup \{y^*\}) \setminus \{y^*\}$ can be written $y^* + \lambda(y^* - \tilde{y})$ for some $\lambda > 0$ and some $\tilde{y} \in B_{\epsilon_1}(\hat{y})$. Since $\tilde{y} \in L$, $y^* \notin L$, and $L$ is convex, we have that no member of $\text{co}(B_{\epsilon_2}(q) \cup \{y^*\})$ belongs to $L$. Fix any $\epsilon \in (0, 1]$ such that $r_\epsilon \equiv \epsilon r + (1-\epsilon)q \in B_{\epsilon_2}(q)$. Since $\text{co}(q, r_\epsilon, y^*) \subset \text{co}(B_{\epsilon_2}(q) \cup \{y^*\})$, we have that
\begin{align}
U(p \mid q, r_\epsilon, y^*, p) &= U(p \mid p, y^*) \\
&= U(p \mid q, y^*, p) \\
&= U(p \mid q, y, p) \\
&= U(q \mid q, y, p) \\
&= (1-\gamma)u(q).
\end{align}

Since $v^*(q) > \max\{v^*(y), v^*(\underline{q})\}$, and since $p \in \text{co}(y, \underline{q})$, we have that $v^*(q) > v^*(p)$. Since $v^*(q) > v^*(\bar{y})$, and since $y^* \in \text{int}(\text{co}(q, \bar{y}))$, we have that $v^*(q) > v^*(y^*)$. Thus,
\[U(q \mid q, y^*, p) = (1-\gamma)u(q) = U(p \mid q, y^*, p).\]
Conclude that $q \in c_2(p, q \mid p, q, y^*)$.

We show that there exists $x \in \text{co}(y^*, r_\epsilon)$ such that $v(x) > v(q)$ for all $v \in \mathcal{V}$. Recall that there exists $\alpha \in (0, 1)$ such that $v(\alpha y + (1-\alpha)r) > v(q)$ for all $v \in \mathcal{V}$. Since $v(\bar{y}) > v(y)$ for all $v \in \mathcal{V}$, we have $v(\alpha \bar{y} + (1-\alpha)r) > v(q)$ for all $v \in \mathcal{V}$. Since $y^* = \beta \bar{y} + (1-\beta)q$ for some $\beta \in (0, 1)$, we have that
\begin{multline}
    v\left(\frac{\alpha}{\alpha + \beta(1-\alpha)}y^*+\frac{\beta(1-\alpha)}{\alpha + \beta(1-\alpha)}r \right) \\= v\left(\frac{\beta}{\alpha + \beta(1-\alpha)}\left(\alpha \bar{y} + (1-\alpha)r\right) + \frac{\alpha(1-\beta)}{\alpha + \beta(1-\alpha)}q \right) > v(q)
    \end{multline}
for all $v \in \mathcal{V}$. Thus, there exists $\hat{\alpha} \in (0, 1)$ such that $v(\hat{\alpha}y^* + (1-\hat{\alpha})r) > v(q)$ for all $v \in \mathcal{V}$. We have
\begin{multline}
v\left(\frac{\epsilon\hat{\alpha}}{1-\hat{\alpha}(1-\epsilon)}y^* + \frac{1-\hat{\alpha}}{1-\hat{\alpha}(1-\epsilon)}r_\epsilon\right) \\
=v\left(\frac{\epsilon}{1-\hat{\alpha}(1-\epsilon)}\left(\hat{\alpha}y^* + (1-\hat{\alpha})r\right) + \frac{(1-\epsilon)(1-\hat{\alpha})}{1-\hat{\alpha}(1-\epsilon)}q\right)> v(q)
\end{multline}
for all $v \in \mathcal{V}$. This implies
\[U(q \mid q, r_\epsilon, y^*, p) < (1-\gamma)u(q) = U(p \mid q, r_\epsilon, y^*, p).\]
Conclude that $q \notin c_2(q, p \mid q, p, y^*, r_\epsilon)$. Since $q \in c_2(q, p \mid q, p, y^*)$, we have that $r$ is not worse than $q$. 

Finally, consider the case that $v(q) \geq v(r)$ for all $v \in \mathcal{V}$, but $v(q) = v(r)$ for some $v \in \mathcal{V}$. Then, there are $\tilde{r}$ arbitrarily close to $r$ such that $v(\tilde{r}) > v(q)$ for some $v \in \mathcal{V}$. For all such $\tilde{r}$, we can find $y \in \text{co}(q, \tilde{r}) \setminus \{q\}$ and $(M, \mathcal{M})$ such that $q \in c_2(M \setminus \{y\} \mid \mathcal{M} \setminus \{y\})$, but $q \notin c_2(M \mid \mathcal{M})$. Thus, $r$ is not worse than $q$. 
\end{proof}

\begin{lemma}
\label{lem:UE2U}
For any $q \in \text{int}(\Delta(Z))$: if $\ME(q \mid x)$ is nonempty, then $(1-\gamma)\UE(q \mid x) = U(q \mid q, x)$.
\end{lemma}

\begin{proof}
Suppose $q$ is not worse than $x$. By Lemma \ref{lem:V2worse}, $q$ is not worse than any member of $\text{co}(q, x)$, so $q \in \ME(q \mid x)$. We have 
\[(1-\gamma)\UE(q \mid x) = (1-\gamma)u(q).\]
By Lemma \ref{lem:V2worse},
\[(1-\gamma)u(q) = U(q \mid q, x).\]

Now suppose that $q$ is worse than $x$. By Lemma \ref{lem:V2worse} and $u \in \mathcal{V}$, we have $u(q) < u(x)$. Fix $r \in \ME(q \mid x)$. Since $r$ is not worse than any member of $\text{co}(q, x)$, Lemma \ref{lem:V2worse} implies
\[U(r \mid q, r, x) = (1-\gamma)u(r).\]
Since $q \in c_2(q, r \mid q, r, x)$,
\[U(q \mid q, r, x) \geq (1-\gamma)u(r).\]
Since 
\[(1-\gamma)u(r) \leq U(q \mid q, r, x) \leq U(q \mid q, x) < (1-\gamma)u(q),\]
we have $u(r) < u(q)$. 

Suppose that $U(q \mid q, r, x) > (1-\gamma)u(r)$. For each $\epsilon \in (0, 1]$, let $r_\epsilon \equiv \epsilon r + (1-\epsilon)x$. Since $r$ is interior, $r_\epsilon$ is interior for all $\epsilon \in (0, 1]$. Since $u(r) < u(q) < u(x)$, we have $u(r_\epsilon) > u(r)$ for all $\epsilon \in (0, 1)$. Since $r$ is not worse than any member of $\text{co}(q, x)$, Lemma \ref{lem:V2worse} implies the existence of $v \in \mathcal{V}$ such that $v(r) \geq v(x) > v(q)$. This same $v \in \mathcal{V}$ satisfies $v(r_\epsilon) \geq v(x) > v(q)$ for all $\epsilon \in (0, 1]$. By Lemma \ref{lem:V2worse}, no $r_\epsilon$ is worse than any member of $\text{co}(q, x)$. Since $U(q \mid q, r, x) > (1-\gamma)u(r)$, we can ensure $U(q \mid q, r_\epsilon, x) \geq (1-\gamma)u(r_\epsilon)$ by choosing $\epsilon$ sufficiently close to 1. This implies $q \in c_2(q, r_\epsilon \mid q, r_\epsilon, x)$. Since $r_\epsilon$ is interior and is not worse than any member of $\text{co}(q, x)$, and since $u(r_\epsilon) > u(r)$, we have a contradiction to $r \in \ME(q \mid x)$. Conclude that $U(q \mid q, r, x) = (1-\gamma)u(r)$.

Suppose that $U(q \mid q, r, x) < U(q \mid q, x)$. This could not be the case if $v(x) \geq v(r)$ for all $v \in \mathcal{V}$. Thus, $v(r) > v(x)$ for some $v \in \mathcal{V}$. By Lemma \ref{lem:closeworse_rep}, there exists $\underline{r}$ arbitrarily close to $r$ such that $v(r) > v(\underline{r})$ for all $v \in \mathcal{V}$. By choosing $\underline{r}$ sufficiently close to $r$, we can ensure $\underline{r} \in \text{int}(\Delta(Z))$ and $v(\underline{r}) > v(x)$ for some $v \in \mathcal{V}$. For any $\beta \in (0, 1]$, for all $v \in \mathcal{V}$, 
\[v(q) - v(\beta \underline{r} + (1-\beta)x) > v(q) - v(\beta r + (1-\beta)x).\]
This implies
\begin{equation}
\label{eq:beta1}
\max_{v \in \mathcal{V}}\left(v(q) - v(\beta \underline{r} + (1-\beta)x)\right) > \max_{v \in \mathcal{V}}\left(v(q) - v(\beta r + (1-\beta)x)\right) \text{ for all } \beta \in (0, 1].
\end{equation}
Since $U(q \mid q, r, x) < U(q \mid q, x)$, 
\begin{equation}
\label{eq:beta2}
\max_{v \in \mathcal{V}}\left(v(q) - v(x)\right) > \max_{v \in \mathcal{V}}\left(v(q) - v(\beta r + (1-\beta)x)\right) \text{ for some } \beta \in (0, 1].
\end{equation}
By (\ref{eq:beta1}) and (\ref{eq:beta2}),
\begin{equation}
\label{eq:beta3}
\min_{y \in \text{co}(\underline{r}, x)}\max_{v \in \mathcal{V}}\left(v(q) - v(y) \right) > \min_{y \in \text{co}(r, x)}\max_{v \in \mathcal{V}}\left(v(q) - v(y)\right).
\end{equation}
Since $u \in \mathcal{V}$, we have $u(\underline{r}) < u(r) < u(q) < u(x)$. Thus, there exists $r^* \in \text{int}(\text{co}(x, \underline{r}))$ such that $u(r^*) = u(r)$. Since $\underline{r}$ is interior, so is $r^*$. Since $v(\underline{r}) > v(x)$ for some $v \in \mathcal{V}$, we have $v(r^*) > v(x) > v(q)$ for some $v \in \mathcal{V}$. By Lemma \ref{lem:V2worse}, $r^*$ is not worse than any member of $\text{co}(q, x)$. 

For any $y \in \text{co}(\underline{r}, r^*)$, we have $u(q) > u(y)$, so $v(q) \geq v(y)$ for some $v \in \mathcal{V}$. Thus, no member of $\text{co}(\underline{r}, r^*)$ can belong to $\argmin_{y \in \text{co}(\underline{r}, x)}\max_{v \in \mathcal{V}}\left(v(q) - v(y) \right)$. This allows us to replace $\underline{r}$ with $r^*$ in (\ref{eq:beta3}):
\[\min_{y \in \text{co}(r^*, x)}\max_{v \in \mathcal{V}}\left(v(q) - v(y) \right) > \min_{y \in \text{co}(r, x)}\max_{v \in \mathcal{V}}\left(v(q) - v(y)\right).\]
This implies 
\[U(q \mid q, r^*, x) > U(q \mid q, r, x) = (1-\gamma)u(r) = (1-\gamma)u(r^*) = U(r^* \mid q, r^*, x),\]
Thus, $\{q\} = c_2(q, r^* \mid q, r^*, x)$ and $r^* \in \ME(q \mid x)$. But we showed above that this is impossible. Conclude that, for any $r \in \ME(q \mid x)$, 
\[U(q \mid q, r, x) = U(q \mid q, x).\]

Since 
\[U(q \mid q, r, x) = (1-\gamma)u(r)\]
for any $r \in \ME(q \mid x)$, we have
\[(1-\gamma)\UE(q \mid x) = U(q \mid q, r, x) = U(q \mid q, x).\]
\end{proof}

\textbf{Existence:} Fix any $q \in \text{int}(\Delta(Z))$. Fix any $x$ such that $q$ is not worse than $x$. We show that $\ME(q \mid x) \neq \emptyset$. We have $q \in c_2(q \mid q, x)$. By Lemma \ref{lem:V2worse}, $q$ is not worse than any member of $\text{co}(q, x)$. Suppose there exists $r$ such that $u(r) > u(q)$, $r$ is not worse than any member of $\text{co}(q, x)$, and $q \in c_2(q, r \mid q, r, x)$. By Lemma \ref{lem:V2worse}, 
\[U(r \mid q, r, x) = (1-\gamma)u(r).\]
Since $q \in c_2(q, r \mid q, r, x)$, we have
\[U(q \mid q, r, x) \geq (1-\gamma)u(r) > (1-\gamma)u(q),\]
which is impossible. Conclude that $q \in \ME(q \mid x)$.

Since $u \in \mathcal{V}$ but $\mathcal{V}$ contains at least one preference distinct from $u$, there exists $\underline{q} \in \Delta(Z)$ such that $u(q) > u(\underline{q})$, $v(q) \geq v(\underline{q})$ for all $v \in \mathcal{V}$, and $v(q) = v(\underline{q})$ for some $v \in \mathcal{V}$. It is without loss to assume that $\underline{q} \in \text{int}(\Delta(Z))$. For some $\epsilon_1 > 0$, three conditions are satisfied. First, the distance between $q$ and the boundary of $\Delta(Z)$ is greater than $\epsilon_1$. 
Second, for all $x \in B_{\epsilon_1}(q)$, the distance between $\underline{q} + x - q$ and the boundary of $\Delta(Z)$ is at least $\epsilon_1$. Third, $\max_{x \in \bar{B}_{\epsilon_1}(q)}u(\underline{q} + x - q) < u(q)$. 

We claim that, for any $\alpha \in (0, 1)$,
\begin{equation}
\label{eq:shrink}
\min_{x \in \bar{B}_{\alpha\epsilon_1}(q)}U(q \mid q, x) = \alpha \min_{x \in \bar{B}_{\epsilon_1}(q)}U(q \mid q, x) + (1-\alpha)u(q).
\end{equation}
Fix any $x^* \in \argmin_{x \in \bar{B}_{\epsilon_1}(q)}U(q \mid q, x)$. We have
\begin{align}
\min_{x \in \bar{B}_{\alpha\epsilon_1}(q)}U(q \mid q, x) &\leq U(q \mid q, \alpha x^* + (1-\alpha)q) \\
&= \alpha U(q \mid q, x^*) + (1-\alpha)u(q) \\
&= \alpha \min_{x \in \bar{B}_{\epsilon_1}(q)}U(q \mid q, x) + (1-\alpha)u(q).
\end{align}
Now fix any $x^*_\alpha \in \argmin_{x \in \bar{B}_{\alpha\epsilon_1}(q)}U(q \mid q, x)$. Since the distance between $q$ and the boundary of $\Delta(Z)$ is greater than $\epsilon_1$, 
\[x^* \equiv \frac{1}{\alpha}x^*_\alpha - \frac{1-\alpha}{\alpha}q \in \Delta(Z).\]
Since $x^* \in B_\epsilon(q)$, we have
\begin{align}
\min_{x \in \bar{B}_{\epsilon_1}(q)}U(q \mid q, x) &\leq U(q \mid q, x^*)\\
&= \frac{1}{\alpha}U(q \mid q, x^*_\alpha) - \frac{1-\alpha}{\alpha}u(q)\\
&= \frac{1}{\alpha}\min_{x \in \bar{B}_{\alpha\epsilon_1}(q)}U(q \mid q, x) - \frac{1-\alpha}{\alpha}u(q) \\
&\Longrightarrow \quad \alpha \min_{x \in \bar{B}_{\epsilon_1}(q)}U(q \mid q, x) + (1-\alpha)u(q) \leq \min_{x \in \bar{B}_{\alpha\epsilon_1}(q)}U(q \mid q, x).
\end{align}
By (\ref{eq:shrink}) and $u(q) > \max_{x \in \bar{B}_{\epsilon_1}(q)}u(\underline{q} + x - q)$, there exists $\epsilon_2 \in (0, \epsilon_1)$ such that
\[\min_{x \in \bar{B}_{\epsilon_2}(q)}U(q \mid q, x) > (1-\gamma)\max_{x \in \bar{B}_\epsilon(q)}u(\underline{q} + x - q).\]
For each $x \in B_{\epsilon_2}(q)$ such that $q$ is worse than $x$, we have
\[(1-\gamma)u(x) > (1-\gamma)u(q) > U(q \mid q, x) > (1-\gamma)u(\underline{q} + x - q).\]
Thus, there exists $r \in \text{int}(\text{co}(x, \underline{q} + x - q))$ such that $(1-\gamma)u(r) = U(q \mid q, x)$.

Since $q$ is worse than $x$, Lemma \ref{lem:V2worse} implies $v(x) > v(q)$ for all $v \in \mathcal{V}$. Recall that there exists $v^* \in \mathcal{V}$ such that $v^*(q) = v^*(\underline{q})$, so $v^*(\underline{q}+x-q) = v^*(x) > v^*(q)$. Since $r \in \text{co}(x, \underline{q} + x - q)$, we have $v^*(r) \geq v^*(y)$ for all $y \in \text{co}(q, x)$. By Lemma \ref{lem:V2worse}, $r$ is not worse than any member of $\text{co}(q, x)$. This implies
\[U(r \mid q, r, x) = (1-\gamma)u(r).\]
Since $v(\underline{q} + x - q) \leq v(x)$ for all $v \in \mathcal{V}$, we have $v(r) \leq v(x)$ for all $v \in \mathcal{V}$. Since $v(x) > v(q)$ for all $v \in \mathcal{V}$, 
\begin{align}
U(q \mid q, r, x) &= (1-\gamma)u(q) + \max_{v \in \mathcal{V}}\left(v(q) - \max_{q, r, x}v\right) \\
&= (1-\gamma)u(q) + \max_{v \in \mathcal{V}}\left(v(q) - v(x)\right) \\
&= U(q \mid q, x) \\
&= (1-\gamma)u(r) \\
&= U(r \mid q, r, x).
\end{align}

Now fix any $p, x \in B_{\epsilon_2/2}(q)$. We show that $\ME(p \mid x)$ is nonempty. If $p$ is not worse than $x$, we already showed that $p \in \ME(p \mid x)$. Suppose $p$ is worse than $x$. Since $d(p, x) < \epsilon_2$, we have that $q + x - p \in B_{\epsilon_2}(q)$. Thus, there exists $r \in \text{co}(q + x - p, \underline{q} + x - p)$ such that $v(r) \leq v(q + x - p)$ for all $v \in \mathcal{V}$ and 
\[U(r \mid q, r, q + x - p) = (1-\gamma)u(r) = U(q \mid q, q + x - p) = U(q \mid q, r, q + x - p).\]
Since $q + x - p \in B_{\epsilon_2}(q)$, the distance between $\underline{q} + x - p$ and the boundary of $\Delta(Z)$ is greater than $\epsilon_2$. Since $p \in B_{\epsilon_2/2}(q)$, we have that $(\underline{q} + x - p)+ (p - q) \in \text{int}(\Delta(Z))$. Also, $(q + x - p) + (p - q) = x \in \text{int}(\Delta(Z))$. Since $r \in \text{co}(q + x - p, \underline{q} + x - p)$, we have $r + p - q \in \text{int}(\Delta(Z))$. 

Since
\[U(r \mid q, r, q + x - p) = (1-\gamma)u(r),\]
we have
\[U(r + p - q \mid p, r + p - q, x) = (1-\gamma)u(r + p - q).\]
By Lemma \ref{lem:V2worse}, $r+p-q$ is not worse than any member of $\text{co}(p, x)$. We also have
\begin{align}
U(p \mid p, r + p - q, x) &= (1-\gamma)\left(u(p) - u(q)\right) + U(q \mid q, r, q + x - p) \\
&= (1-\gamma)\left(u(p) - u(q)\right) + U(q \mid q, q + x - p) \\
&= (1-\gamma)u(r + p - q) \\
&= U(r + p - q \mid p, r + p - q, x).
\end{align}
Thus, $p \in c_2(p, r + p - q \mid p, r + p - q, x)$. 

Since $p$ is worse than $x$, we have $v(x) > v(p)$ for all $v \in \mathcal{V}$. Since $v(q + x - p) \geq v(r)$ for all $v \in \mathcal{V}$, we have $v(x) \geq v(r + p - q)$ for all $v \in \mathcal{V}$. This implies
\[U(p \mid p, r + p - q, x) = U(p \mid p, x).\]
Take any $r^*$ that is not worse than any member of $\text{co}(p, x)$ and such that $u(r^*) > u(r + p - q)$. We have
\begin{align}
U(r^* \mid p, r^*, x) &= (1-\gamma)u(r^*) \\
&> (1-\gamma)u(r + p - q) \\
&= U(p \mid p, r + p - q, x)\\
&= U(p \mid p, x) \\
&\geq U(p \mid p, r^*, x).
\end{align}
Thus, $p \notin c_2(p, r^* \mid p, r^*, x)$. Conclude that $r+p-q \in \ME(p \mid x)$, so $\ME(p \mid x) \neq \emptyset$. Since $p$ and $x$ were arbitrary members of $B_{\epsilon_2/2}(q)$, we can take $S \equiv B_{\epsilon_2/2}(q)$. 

\textbf{Rationalization:} By Lemma \ref{lem:UE2U}, Rationalization can be rewritten
\[c_2(M \mid \mathcal{M}) = \argmax_{q \in M}\min_{y \in \text{co}(\mathcal{M})}U(q \mid q, y).\]
By the minimax theorem, 
\begin{align}
\argmax_{q \in M}\min_{y \in \text{co}(\mathcal{M})}U(q \mid q, y) &=
\argmax_{q \in M}\left((1-\gamma)u(q) + \gamma \min_{y \in \text{co}(\mathcal{M})}\max_{v \in \mathcal{V}}\left(v(q) - v(y)\right) \right)\\
&= \argmax_{q \in M}\left((1-\gamma)u(q) + \gamma \max_{v \in \mathcal{V}}\left(v(q) - \max_{y \in \mathcal{M}}v(y)\right)\right)\\
&= \argmax_{q \in M}U(q \mid \mathcal{M})\\
&= c_2(M \mid \mathcal{M}).
\end{align} 

\textbf{Monotonicity:} Suppose $q$ is worse than $x$ and $x$ is worse than $y$. By Lemma \ref{lem:UE2U}, it suffices to show that $U(q \mid y) < U(q \mid x)$. By Lemma \ref{lem:V2worse}, $v(q) < v(x) < v(y)$ for all $v \in \mathcal{V}$. Thus,
\[\max_{v \in \mathcal{V}}\left(v(q) - v(y)\right) < \max_{v \in \mathcal{V}}\left(v(q) - v(x)\right).\]
This implies $U(q \mid y) < U(q \mid x)$.

\textbf{Quasiconvexity:} By Lemma \ref{lem:UE2U}, it suffices to show that 
\[U(q \mid \alpha x + (1-\alpha)y) \leq \max\left\{U(q \mid x), \UE(q \mid y) \right\}.\]
If $q$ is not worse than $x$, we have 
\[U(q \mid x) = (1-\gamma)u(q) \geq U(q \mid \alpha x + (1-\alpha)y),\]
and similarly if $q$ is not worse than $y$. Suppose that $q$ is worse than $x$ and $y$. We have
\begin{align}
\max_{v \in \mathcal{V}}\left(v(q) - v(\alpha x + (1-\alpha)y)\right) &\leq \alpha \max_{v \in \mathcal{V}}\left(v(q) - v(x) \right) + (1-\alpha)\max_{v \in \mathcal{V}}\left(v(q) - v(y) \right) \\
&\leq \max\left\{\max_{v \in \mathcal{V}}\left(v(q) - v(x) \right), \max_{v \in \mathcal{V}}\left(v(q) - v(y) \right) \right\}.
\end{align}
This implies the desired result. 

\textbf{Continuity:} By Lemma \ref{lem:UE2U}, it suffices to show that $U(q \mid q, x)$ is Lipschitz continuous. It suffices to show Lipschitz continuity in each argument. 

For the first argument: we need to show that $\frac{|U(q \mid q, x) - U(r \mid r, x)|}{d(q, r)}$ is bounded. If $q$ and $r$ are both worse than $x$, 
\begin{multline}
|U(q \mid q, x) - U(r \mid r, x)| \\\leq (1-\gamma)|u(q) - u(r)| + \gamma \left|\max_{v \in \mathcal{V}}\left(v(q) - v(x) \right)- \max_{v \in \mathcal{V}}\left(v(r) - v(x) \right)\right|.
\end{multline}
If neither $q$ nor $r$ is worse than $x$,
\[|U(q \mid x) - U(r \mid x)| = (1-\gamma)|u(q) - u(r)|.\]
If $q$ is worse than $x$ but $r$ is not, and if $U(q \mid x) \geq U(r \mid x)$,
\begin{align}
|U(q \mid x) - U(r \mid x)| &= (1-\gamma)\left(u(q) - u(r)\right) + \gamma\max_{v \in \mathcal{V}}\left(v(q) - v(x)\right)\\
&< (1-\gamma)\left(u(q) - u(r)\right)\\
&\leq (1-\gamma)|u(q) - u(r)|.
\end{align}.
If $q$ is worse than $x$ but $r$ is not, and if $U(q \mid x) \leq U(r \mid x)$,
\begin{align}
|U(q \mid x) - U(r \mid x)| &= (1-\gamma)\left(u(r) - u(q)\right) - \gamma\max_{v \in \mathcal{V}}\left(v(q) - v(x) \right)\\
&\leq (1-\gamma)\left(u(r) - u(q)\right) + \gamma\left(\max_{v \in \mathcal{V}}\left(v(r) - v(x)\right) - \max_{v \in \mathcal{V}}\left(v(q) - v(x)\right)\right)\\
&\leq (1-\gamma)|u(r) - u(q)| + \gamma\left|\max_{v \in \mathcal{V}}\left(v(r) - v(x)\right) - \max_{v \in \mathcal{V}}\left(v(q) - v(x)\right)\right|.
\end{align}
The case in which $r$ is worse than $x$ but $q$ is not can be handled in the same way as the previous two cases. Thus, in all cases,
\begin{multline}
\frac{|U(q \mid x) - U(r \mid x)|}{d(q, r)} \\\leq \gamma \frac{|u(q) - u(r)|}{d(q, r)} + \gamma \frac{|\max_{v \in \mathcal{V}}\left(v(r) - v(x)\right) - \max_{v \in \mathcal{V}}\left(v(q) - v(x)\right)|}{d(q, r)}.
\end{multline}
We have
\[\frac{|u(q) - u(r)|}{d(q, r)} = \frac{u}{\lVert u \rVert}\cdot \frac{q - r}{\lVert q - r \rVert}\lVert u \rVert \leq \max_{v \in \mathcal{V}}\lVert v \rVert.\]
It is without loss to assume $\max_{v \in \mathcal{V}}\left(v(q) - v(x) \right) \geq \max_{v \in \mathcal{V}}\left(v(r) - v(x) \right)$. Fix some 
\[v_q \in \argmax_{v \in \mathcal{V}}\left(v(q) - v(x) \right).\] 
We have
\begin{align}
\frac{|\max_{v \in \mathcal{V}}\left(v(q) - v(x) \right)- \max_{v \in \mathcal{V}}\left(v(r) - v(x) \right)|}{d(q, r)} &\leq \frac{v_q(q) - v_q(x) - \left(v_q(r) - v_q(x) \right)}{d(q, r)} \\
& = \frac{v_q(q) - v_q(r)}{d(q, r)}\\
&= \frac{v_q}{\lVert v_q \rVert}\cdot \frac{q - r}{\lVert q - r \rVert}\lVert v_q \rVert \\
&\leq \max_{v \in \mathcal{V}}\lVert v \rVert.
\end{align}
Thus,
\begin{align}
\frac{|U(q \mid x) - U(r \mid x)|}{d(q, r)} \leq \max_{v \in \mathcal{V}}\lVert v \rVert.
\end{align}
Since $\mathcal{V}$ is compact, the right-hand side exists. 

For the second argument: we need to show that $\frac{|U(q \mid x) - U(q \mid y)|}{d(x, y)}$ is bounded. If $q$ is worse than both $x$ and $y$, we have
\[|U(q \mid x) - U(q \mid y)| = \gamma\left|\max_{v \in \mathcal{V}}\left(v(q) - v(x)\right) - \max_{v \in \mathcal{V}}\left(v(q) - v(y)\right)\right|.\]
If $q$ is worse than neither $x$ nor $y$, we have
\[|U(q \mid x) - U(q \mid y)| = 0.\]
If $q$ is worse than $x$ but not $y$, then
\begin{align}
|U(q \mid x) - U(q \mid y)| &= -\gamma\max_{v \in \mathcal{V}}\left(v(q) - v(x)\right)\\
&\leq \gamma \left(\max_{v \in \mathcal{V}}\left(v(q) - v(y) \right) - \max_{v \in \mathcal{V}}\left(v(q) - v(x)\right)\right)\\
&\leq \gamma \left|\max_{v \in \mathcal{V}}\left(v(q) - v(x)\right) - \max_{v \in \mathcal{V}}\left(v(q) - v(y)\right)\right|.
\end{align}
The case in which $q$ is worse than $y$ but not $x$ can be handled in the same way. Thus, in all cases,
\[\frac{|U(q \mid x) - U(q \mid y)|}{d(x, y)} \leq \gamma\frac{|\max_{v \in \mathcal{V}}\left(v(q) - v(x)\right) - \max_{v \in \mathcal{V}}\left(v(q) - v(y)\right)|}{d(x,y)}.\]
It is without loss to assume that $\max_{v \in \mathcal{V}}\left(v(q) - v(x)\right) > \max_{v \in \mathcal{V}}\left(v(q) - v(y)\right)$. Fix some \[v_x \in \argmax_{v \in \mathcal{V}}\left(v(q) - v(x)\right).\] We have
\begin{align}
\frac{|\max_{v \in \mathcal{V}}\left(v(q) - v(x) \right)- \max_{v \in \mathcal{V}}\left(v(q) - v(y) \right)|}{d(x, y)} &\leq \frac{v_x(q) - v_x(x)- \left(v_x(q) - v_x(y)\right)}{d(x, y)} \\
&= \frac{v_x(y) - v_x(x)}{d(x, y)}\\
&\leq \frac{v_x}{\lVert v_x \rVert}\cdot \frac{y - x}{\lVert y - x\rVert}\lVert v_x \rVert\\
&\leq \max_{v \in \mathcal{V}}\lVert v \rVert.
\end{align}
Thus,
\[\frac{|U(q \mid x) - U(q \mid y)|}{d(x, y)} \leq \gamma \max_{v \in \mathcal{V}}\lVert v \rVert.\]

\subsection{Sufficiency}

\begin{lemma}
\label{lem:Vpref}
For any nonempty, closed, convex $V_\text{pref} \subset \mathcal{U}_\text{pref}$, any $\hat{\succsim} \in \mathcal{U}_\text{pref}$, and any $q \in \text{int}(\Delta(Z))$: if
\begin{equation}
\label{eq:congruence}
q \succ r \text{ for all } \succsim \in V_\text{pref} \quad \Longrightarrow \quad q \; \hat{\succ} \; r \text{ for all } r \in \Delta(Z),
\end{equation}
then $\hat{\succsim} \in V_\text{pref}$. 
\end{lemma}

\begin{proof}
If $|V_\text{pref}| = 1$, the result is trivial. For the rest of the proof, we assume that $|V_\text{pref}| > 1$. 

Let 
\[W(q) \equiv \left\{r \in \Delta(Z): q \succ r \text{ for all } v \in V_\text{pref}\right\}.\]
If there does not exist $\succsim^* \in \mathcal{U}_\text{pref}$ such that
\[r \in \bar{W}(q) \setminus \{q\} \quad \Longrightarrow \quad q \succ^* r \text{ for all } r \in \Delta(Z),\]
then there exists $r \in \Delta(Z)$ such that $r \in \bar{W}(q) \setminus \{q\}$ and $2q - r \in \bar{W}(q) \setminus \{q\}$. This implies $q \sim r$ for all $\succsim$ satisfying (\ref{eq:congruence}). In turn, this implies the existence of $z^* \in Z$ and $\rho_1 \in \Delta(Z \setminus \{z^*\})$ such that
\[\delta_{z^*} \sim \rho_1 \text{ for all } \succsim \text{ satisfying (\ref{eq:congruence})}.\]
It is without loss to assume that $z^* = z_n$, where $Z = \{z_1, \ldots, z_n\}$. Let $Z_1 \equiv \{z_1, \ldots, z_{n-1}\}$. For any lottery $r \in \Delta(Z)$, let
\[r_1 \equiv r + r(z_n)\left(\rho_1 - \delta_{z_n}\right).\]
We have that $r_1 \in \Delta(Z_1)$ and that $r_1 \sim r$ for all $\succsim$ satisfying (\ref{eq:congruence}). 

Now suppose there does not exist $\succsim^* \in \mathcal{U}_\text{pref}$ such that
\[r \in \bar{W}(q_1) \setminus \{q_1\} \quad \Longrightarrow \quad q_1 \succ^* r \text{ for all } r \in \Delta(Z_1).\]
We can repeat the above argument with $z_{n-1}$ in place of $z_n$. Since
\[z_n \sim \rho_1 \sim (\rho_1)_2 \in \Delta(Z_2) \text{ for all } \succsim \text{ satisfying (\ref{eq:congruence})},\]
it is without loss to assume that $\rho_1 \in \Delta(Z_2)$. We then have $\rho_1, \rho_2 \in \Delta(Z_2)$ such that $\rho_1 \sim \delta_{z_n}$ and $\rho_2 \sim \delta_{z_{n-1}}$ for all $\succsim$ satisfying (\ref{eq:congruence}). 

We claim that there exists $K \in \{0, \ldots, n-3\}$ satisfying two conditions. First, for each $k \in \{1, \ldots, K\}$, there exists $\rho_k \in \Delta(Z_K)$ such that $\rho_k \sim \delta_{z_{n-k+1}}$ for all $\succsim$ satisfying (\ref{eq:congruence}). Second, there exists $\succsim^*$ that satisfies the same indifference conditions and
\begin{equation}
\label{eq:K}
r \in \bar{W}(q_K) \setminus \{q_K\} \quad \Longrightarrow \quad q_K \succ^* r \text{ for all } r \in \Delta(Z_K).
\end{equation}
To see this, suppose that we have iterated back to $K = n-3$. We have $Z_{n-3} = \{z_1, z_2, z_3\}$. It there is still no $\succsim^*$ satisfying (\ref{eq:K}), then the restriction of $\bar{W}(q_{n-3})$ to $\Delta(Z_{n-3})$ must be a half-plane. This implies that all $\succsim \in V_\text{pref}$ agree on their restriction to $\Delta(Z_{n-3})$. This implies $|V_\text{pref}| = 1$, which we have ruled out. 

Now that we have established the existence of $\succsim^*$, we show that $\succsim^* \in V_\text{pref}$. By (\ref{eq:K}), 
\[\left\{r \in \Delta(Z_K): r \succsim^* q_K\right\} \setminus \{q_K\} \subseteq \bigcup_{\succsim \in V_\text{pref}}\left\{r \in \Delta(Z_K): r \succ q_K  \right\}.\]
Fix some $\epsilon > 0$ such that $B^K_\epsilon(q_K) \equiv \{r \in \Delta(Z_K): d(q_K, r) < \epsilon\} \subset \text{int}(\Delta(Z_K))$. We have
\[\left\{r \in \Delta(Z_K): r \succsim^* q_K\right\} \setminus B^K_\epsilon(q) \subseteq \bigcup_{\succsim \in V_\text{pref}}\left\{r \in \Delta(Z_K): r \succ q_K  \right\}.\]
By the Heine-Borel theorem, there exists some finite $W_\text{pref} \subseteq V_\text{pref}$ such that 
\begin{equation}
\label{eq:HB1}
\left\{r \in \Delta(Z_K): r \succsim^* q_K\right\} \setminus B^K_\epsilon(q) \subseteq \bigcup_{\succsim \in W_\text{pref}}\left\{r \in \Delta(Z_K): r \succ q_K  \right\}.
\end{equation}
We claim that
\begin{equation}
\label{eq:HB2}
\left\{r \in \Delta(Z_K): r \succsim^* q_K\right\} \setminus \{q_K\} \subseteq \bigcup_{\succsim \in W_\text{pref}}\left\{r \in \Delta(Z_K): r \succ q_K  \right\}.
\end{equation}
Suppose not. Then, there exists $r \in B^K_\epsilon(q_K) \setminus \{q_K\}$ such that $r \succsim^* q_K$, but $q_K \succsim r$ for all $\succsim \in W_\text{pref}$. For some $\lambda > 0$, we have $r + \lambda(r - q_K) \in \Delta(Z_K) \setminus B^K_\epsilon(q_K)$. We have $r + \lambda(r-q_K) \succsim^* q_K$, but $q_K \succsim r + \lambda(r - q_K)$ for all $\succsim \in W_\text{pref}$. This contradicts (\ref{eq:HB1}).

Assign $\succsim^*$ a representation $v^*$ such that $v^*(q_K) = 0$. Assign each $\succsim \in W_\text{pref}$ a representation $v$ such that $v(q_K) = 0$, and let $W$ denote the resulting set of utilities. By Farkas' lemma, we can find $\alpha > 0$ such that $\alpha \in \text{co}(W)$ provided there does not exist $\phi \in \mathbb{R}^n$ such that $\phi'v^* < 0$ and $\phi'v \geq 0$ for all $v \in W$. Suppose there does exist some such $\phi$. Since $\rho_k \sim \delta_{z_{n-k+1}}$ for all $k \in \{1, \ldots, K\}$ and all $\succsim \in W_\text{pref} \cup \{\succsim^*\}$, it is without loss to assume that $\phi(k) = 0$ for all $k \in \{K+1, \ldots, n\}$. Consider the case $\sum_{\{k: \phi(k) > 0\}}\phi(k) \geq \sum_{\{k: \phi(k) < 0\}}(-\phi(k))$. We can always normalize $\phi$ by dividing each $\phi(k)$ by $\sum_{\{k: \phi(k) > 0\}}\phi(k)$. (Since $\phi \neq 0$, this term must be strictly positive.) By definition of $\phi$,
\begin{align}
&\sum_{k: \phi(k) > 0}\phi(k) v^*(z_k) < \sum_{k: \phi(k) < 0}(-\phi(k)) v^*(z_k) \\
&\sum_{k: \phi(k) > 0}\phi(k) v(z_k) \geq \sum_{k: \phi(k) < 0}(-\phi(k)) v(z_k) \text{ for all } v \in W.
\end{align}
Since the normalization ensures that $\sum_{\{k: \phi(k) > 0\}}\phi(k) = 1$, the left-hand side is the valuation of a lottery in $\Delta(Z_K)$, which we label $\phi^+$. Since $v^*(q_K) = v(q_K) = 0$ for all $v \in W$, we have
\begin{align}
&v^*(\phi^+) < \sum_{k: \phi(k) < 0}(-\phi(k)) v^*(z_k) + \left(1 - \sum_{k: \phi(k) < 0}(-\phi(k))\right)v^*(q_K) \\
&v^*(\phi^-) \geq \sum_{k: \phi(k) < 0}(-\phi(k)) v(z_k) + \left(1 - \sum_{k: \phi(k) < 0}(-\phi(k))\right)v(q_K) \text{ for all } v \in W.
\end{align}
Now the right-hand side is also the valuation of a lottery in $\Delta(Z_K)$, which we label $\phi^-$. We have
\begin{align}
&v^*(\phi^+) < v^*(\phi^-) \\
&v(\phi^+) \geq v(\phi^-) \text{ for all } v \in W.
\end{align}
For $\epsilon > 0$ sufficiently small, we have
\begin{align}
&v^*(q_K) < v^*(q_K + \epsilon(\phi^- - \phi^+)) \\
&v(q_K) \geq v(q_K + \epsilon(\phi^- - \phi^+)) \text{ for all } v \in W.
\end{align}
This contradicts (\ref{eq:HB2}). Conclude that $\alpha v^* \in \text{co}(W)$ for some $\alpha > 0$, so $\succsim^* \in \text{co}(W_\text{pref})$. Since $\text{co}(W_\text{pref}) \subseteq V_\text{pref}$, we have $\succsim^* \in V_\text{pref}$. 

Now take any $\hat{\succsim}$ that satisfies (\ref{eq:congruence}). Assign $\hat{\succsim}$ a representation $\hat{v}$. For each $n \in \mathbb{N}$, let
\[v_n \equiv \frac{1}{n}v^* + \left(1-\frac{1}{n}\right)\hat{v}.\]
Let $\succsim_n$ denote the preference represented by $v_n$. Since $\hat{v}(\rho_k) = \hat{v}(z_{n-k+1})$ and $v^*(\rho_k) = v^*(z_{n-k+1})$ for all $k \in \{1, \ldots, K\}$, we have $v_n(\rho_k) = v_n(z_{n-k+1})$ for all $k \in \{1, \ldots, K\}$. Additionally, since $\hat{\succsim}$ satisfies (\ref{eq:congruence}) and since
\[r \in \bar{W}(q_K) \setminus \{q_K\} \quad \Longrightarrow \quad v^*(q_K) \succ^* v^*(r)\]
for all $r \in \Delta(Z_K)$, we have 
\[r \in \bar{W}(q_K) \setminus \{q_K\} \quad \Longrightarrow \quad v_n(q_K) \succ^* v_n(r)\]
for all $r \in \Delta(Z_K)$ and all $n$. By the previous argument, $\succsim_n \in V_\text{pref}$ for all $n$. Since $v_n \to \hat{v}$ and $\hat{v}$ represents $\hat{\succsim}$, we have $\hat{\succsim} \in V_\text{pref}$. 
\end{proof}

\begin{lemma}
\label{lem:convexworse}
\begin{enumerate} \
    \item If $r$ is not worse than $q$, then for any $p$, $\alpha r + (1-\alpha)p$ is not worse than $\alpha q + (1-\alpha)p$. 
    \item For any interior $q$: if $r$ is worse than $q$, then for any $p$, $\alpha r + (1-\alpha)p$ is worse than $\alpha q + (1-\alpha)p$.
    \end{enumerate}
\end{lemma}

\begin{proof}
\textbf{First part:} Suppose $r$ is not worse than $q$. We can choose $\tilde{r}$ arbitrarily close to $r$ such that, for some $y \in \text{co}(q, \tilde{r})$ and some $(M, \mathcal{M})$, either $q \in c_2(M \mid \mathcal{M})$ or $q \in c_2(M \setminus \{y\} \mid \mathcal{M} \setminus \{y\})$, but not both. Let $\tilde{r}_\alpha \equiv \alpha \tilde{r} + (1-\alpha)p$, and likewise for $q$ and $y$. Since $y \in \text{co}(q, \tilde{r}) \setminus \{q\}$, we have $y_\alpha \in \text{co}(q_\alpha, \tilde{r}_\alpha) \setminus \{q_\alpha\}$. Let $M_\alpha \equiv \alpha M + (1-\alpha)\{p\}$, and likewise for $\mathcal{M}$. By Linearity, $q_\alpha$ is in $c_2(M_\alpha \mid \mathcal{M}_\alpha)$ or $c_2(M_\alpha \setminus \{y_\alpha\} \mid \mathcal{M}_\alpha \setminus \{y_\alpha\})$, but not both. Conclude that $r_\alpha$ is not worse than $q_\alpha$. 

\textbf{Second part:} Suppose $r$ is worse than $q$. Let $q_\alpha \equiv \alpha q + (1-\alpha)p$, and likewise for $r$. Fix $\lambda \in (0, 1)$ such that $(1+\lambda)q - \lambda q_\alpha \in \Delta(Z)$. Let $\beta = \lambda/(1+\lambda)$. We have that 
\begin{align}
&\beta q_\alpha + (1-\beta)((1+\lambda)q - \lambda q_\alpha) = q \\
&\beta r_\alpha + (1-\beta)((1+\lambda)q - \lambda q_\alpha) = (1-\beta\alpha)q + (\beta\alpha) r \in \text{co}(q, r).
\end{align}
Suppose $r_\alpha$ is not worse than $q_\alpha$. By the first part, $(1-\beta\alpha)q + \beta\alpha r$ is not worse than $q$. Thus, there exists $y \in \text{co}(q, r) \setminus \{q\}$ that is not worse than $q$. This contradicts the definition of ``worse than.''
\end{proof}

\begin{lemma}
\label{lem:self}
For any $r \in S$, $\UE(r \mid r) = u(r)$.
\end{lemma}

\begin{proof}
By Existence, $\UE(r \mid r)$ exists. Suppose that $r$ is worse than $r$. By Monotonicity, $\UE(r \mid r) < \UE(r \mid r)$---contradiction. Since $r$ is not worse than $r$ and since $r \in c_2(r \mid r)$, we have $\UE(r \mid r) \geq u(r)$. If $\UE(r \mid r) > u(r)$, then there exists $q$ such that $u(q) > u(r)$ and $r \in c_2(q, r \mid q, r)$. This contradicts the definition of $u$. Thus, $\UE(r \mid r) = u(r)$. 
\end{proof}

\begin{lemma}
\label{lem:regretfree}
For any $r, x \in S$: if $r$ is not worse than $x$, then $\UE(r \mid x) = u(r)$. 
\end{lemma}

\begin{proof}
By Existence, $\UE(r \mid x)$ exists. By Lemma \ref{lem:convexworse}, $r$ is not worse than any member of $\text{co}(r, x)$. Since $r \in c_2(r \mid r, x)$, we have $\UE(r \mid x) \geq u(r)$. 

Suppose $\UE(r \mid x) > u(r)$. Then, there exists $q$ such that $u(q) > u(r)$, $q$ is not worse than any member of $\text{co}(r, x)$, and $r \in c_2(q, r \mid q, r, x)$. 

For some $\alpha \in (0, 1]$, $\alpha q + (1-\alpha)r \in S$. Let $q_\alpha \equiv \alpha q + (1-\alpha)r$, and likewise for $x$. By Existence, $\UE(q_\alpha \mid y)$ and $\UE(r \mid y)$ exist for all $y \in \text{co}(q_\alpha, r, x_\alpha)$. By Lemma \ref{lem:convexworse}, $q_\alpha$ is not worse than any member of $\text{co}(q_\alpha, r, x_\alpha)$. We have already seen that this implies $\min_{y \in \text{co}(q_\alpha, r, x_\alpha)}\UE(q_\alpha \mid y) \geq u(q_\alpha)$. Since $\UE(r \mid r) = u(r)$ by Lemma \ref{lem:self}, we have
\[\min_{y \in \text{co}(q_\alpha, r, x_\alpha)}\UE(r \mid y) \leq u(r) < u(q_\alpha) \leq \min_{y \in \text{co}(q_\alpha, r, x_\alpha)}\UE(q_\alpha \mid y).\]
By Rationalization, $r \notin c_2(q_\alpha, r \mid q_\alpha, r, x_\alpha)$. By Linearity, $r \notin c_2(q, r \mid q, r, x)$, which is a contradiction. Conclude that $\UE(r \mid x) = u(r)$.  
\end{proof}

\begin{lemma}
\label{lem:UEbelow}
For any $r, x \in S$, $\UE(r \mid x) \leq u(r)$.
\end{lemma}

\begin{proof}
Suppose not. Then there exists $q$ such that $u(q) > u(r)$, $q$ is not worse than any member of $\text{co}(r, x)$, and $r \in c_2(q, r \mid q, r, x)$. For some $\alpha \in (0, 1]$, $q_\alpha \equiv \alpha q + (1-\alpha)r \in S$ and $x_\alpha \equiv \alpha x + (1-\alpha)r \in S$. By Linearity, $r \in c_2(q_\alpha, r \mid q_\alpha, r, x_\alpha)$. By Lemma \ref{lem:convexworse}, $q_\alpha$ is not worse than any member of $\text{co}(r, x_\alpha)$. By Existence, for all $y \in \text{co}(q_\alpha, r, x_\alpha)$, both $\UE(q_\alpha \mid y)$ and $\UE(r \mid y)$ exist. By Rationalization and $r \in c_2(q_\alpha, r \mid q_\alpha, r, x_\alpha)$, we have 
\[\min_{y \in \text{co}(q_\alpha, r, x_\alpha)}\UE(r \mid y) \geq \min_{y \in \text{co}(q_\alpha, r, x_\alpha)}\UE(q \mid y).\]
By Lemma \ref{lem:regretfree} and the fact that $q_\alpha$ is not worse than any member of $\text{co}(r, x_\alpha)$, the right-hand side equals $u(q)$. Since $u(q) > u(r)$, we have
\[\min_{y \in \text{co}(q_\alpha, r, x_\alpha)}\UE(r \mid y) > u(r).\]
But this cannot hold since $\UE(r \mid r) = u(r)$ by Lemma \ref{lem:self}.
\end{proof}

\begin{lemma}
\label{lem:regret}
For any $r, x \in S$: if $r$ is worse than $x$, then $\UE(r \mid x) < u(r)$.
\end{lemma}

\begin{proof}
By Lemma \ref{lem:convexworse}, $r$ is worse than $\alpha r + (1-\alpha)x$, and $\alpha r + (1-\alpha)x$ is worse than $x$. By Monotonicity, $\UE(r \mid \alpha r + (1-\alpha)x) > \UE(r \mid x)$. By Lemma \ref{lem:UEbelow}, $u(r) \geq \UE(r \mid \alpha r + (1-\alpha)x)$. Conclude that $u(r) > \UE(r \mid x)$. 
\end{proof}

\begin{lemma}
\label{lem:utilworse}
For any $r, x \in S$: if $r$ is worse than $x$, then $u(x) > u(r)$. 
\end{lemma}

\begin{proof}
Suppose that $r$ is worse than $x$ but $u(r) > u(x)$. We have $x \in c_2(x \mid x)$. By definition of $u$, we have $x \notin c_2(x, r \mid x, r)$. This contradicts the definition of ``worse than.'' 

Now suppose that $r$ is worse than $x$ and $u(r) = u(x)$. By Lemma \ref{lem:regret}, we have $\UE(r \mid x) < u(r) = u(x)$. By Continuity, there exists $\tilde{r}$ such that $\UE(\tilde{r} \mid x) < u(x) = u(r) < u(\tilde{r})$. By Lemma \ref{lem:regretfree}, $\tilde{r}$ is worse than $x$. But since $u(\tilde{r}) > u(x)$, this cannot be. 
\end{proof}

\begin{lemma}
\label{lem:closeworse}
For each $q \in S$, there exists $p$ arbitrarily close to $q$ such that $p$ is worse than $q$, and there exists $p'$ arbitrarily close to $q$ such that $q$ is worse than $p'$.
\end{lemma}

\begin{proof}
By assumption, there exist $(M, \mathcal{M})$ and $y$ such that $c_2(M \mid \mathcal{M}) \neq c_2(M \mid \mathcal{M} \setminus \{y\})$. There exists $\alpha \in (0, 1]$ such that $\mathcal{M}_\alpha \equiv \alpha \mathcal{M} + (1-\alpha)\{q\} \subset S$. Let $y_\alpha \equiv \alpha y + (1-\alpha)q$, and let $M_\alpha \equiv \alpha M + (1-\alpha)\{q\}$. By Linearity, we have $c_2(M_\alpha \mid \mathcal{M}_\alpha) \neq c_2(M_\alpha \mid \mathcal{M}_\alpha \setminus \{y_\alpha\})$. Suppose that no member of $M_\alpha$ is worse than any member of $\text{co}(\mathcal{M}_\alpha)$. By Lemma \ref{lem:regretfree}, we have $\UE(r \mid x) = u(r)$ for all $r \in M_\alpha$ and all $x \in \text{co}(\mathcal{M}_\alpha)$. By Rationalization,
\[c_2(M_\alpha \mid \mathcal{M}_\alpha) = \argmax_{M_\alpha}u = c_2(M_\alpha \mid \mathcal{M}_\alpha \setminus \{y_\alpha\}).\]
This is a contradiction, so there exist $r \in M_\alpha$ and $x \in \text{co}(\mathcal{M}_\alpha)$ such that $r$ is worse than $x$.

Choose any $p \in S$ such that $q - p = \epsilon(x - r)$ for some $\epsilon \in (0, 1)$. Since $x \in S$, there exists $\lambda \in (0, 1/2)$ such that $(1+\lambda)x - \lambda q \in S$. Let $\alpha = 1/(1+\lambda)$. We have
\begin{align}
&\alpha((1+\lambda)x - \lambda q) + (1-\alpha)q = x\\
&\alpha((1+\lambda)x - \lambda q) + (1-\alpha)p = \left(1-\epsilon\frac{\lambda}{1-\lambda}\right)x + \left(\frac{\lambda}{1-\lambda}\epsilon\right) r \in \text{co}(x, r) \setminus \{x\}.
\end{align}
Suppose that $p$ is not worse than $q$. By Lemma \ref{lem:convexworse}, there exists $y \in \text{co}(x, r) \setminus \{x\}$ that is not worse than $x$. Since $r$ is worse than $x$, this contradicts Lemma \ref{lem:convexworse}. Conclude that $p$ is worse than $q$. By Lemma \ref{lem:convexworse}, the same is true for any $y \in \text{co}(p, q) \setminus \{q\}$. 

A parallel argument establishes the existence of $p'$ such that $q$ is worse than $p'$. By Lemma \ref{lem:convexworse}, $q$ is worse than any $y \in \text{co}(q, p') \setminus \{q\}$. 
\end{proof}

\begin{lemma}
\label{lem:mix}
For any $q, x \in S$: $\UE(q \mid \alpha x + (1-\alpha)q) = \alpha \UE(q \mid x) + (1-\alpha)u(q)$.
\end{lemma}

\begin{proof}
First, suppose that $q$ is not worse than $x$. By Lemma \ref{lem:convexworse}, $q$ is not worse than $\alpha x + (1-\alpha)q$. By Lemma \ref{lem:regretfree}, we have
\[\UE(q \mid \alpha x + (1-\alpha)q) = u(q) = \alpha \UE(q \mid x) + (1-\alpha)u(q).\]

Now suppose that $q$ is worse than $x$. By Linearity and Lemma \ref{lem:convexworse}, 
\[\UE(q \mid \alpha x + (1-\alpha)q) \geq \alpha \UE(q \mid x) + (1-\alpha)u(q).\] 
Toward a contradiction, suppose the inequality holds strictly. Let $r^*$ be a member of 
\begin{multline}
\argmax\{u(r): r \text{ is interior and not worse than any member of } \text{co}(q, x), \\ \text{ and } q \in c_2(q, r \mid q, r, x)\}.
\end{multline}
Let $x_\alpha \equiv \alpha x + (1-\alpha)q$, and likewise for $r^*$. Let $p^*_\alpha$ be a member of 
\begin{multline}
\argmax\{u(r): r \text{ is interior and not worse than any member of } \text{co}(q, x_\alpha), \\\text{ and } q \in c_2(q, r \mid q, r, x_\alpha)\}.
\end{multline}
By assumption, $u(r^*_\alpha) < u(p^*_\alpha)$. 

Let
\[p^*_\beta \equiv \left(\frac{\beta}{\alpha}\right)p^*_\alpha + \left(1-\frac{\beta}{\alpha} \right)q.\] Fix $\beta \in (0, \alpha]$ such that $r^*_\beta, p^*_\beta \in S$. Since $u(r^*_\alpha) < u(p^*_\alpha)$, we have $u(r^*_\beta) < u(p^*_\beta)$. 

We show that $\{q, r^*\} = c_2(q, r^* \mid q, r^*, x)$. Since $r^*$ is not worse than any member of $\text{co}(q, x)$, Lemma \ref{lem:convexworse} implies that $r^*_\beta$ is not worse than any member of $\text{co}(q, x_\beta)$. Suppose $r^* \notin c_2(q, r^* \mid q, r^*, x)$. By Linearity, $r^*_\beta \notin c_2(q, r^*_\beta \mid q, r^*_\beta, x_\beta)$. By Rationalization and Lemma \ref{lem:regretfree}, $\min_{y \in \text{co}(q, r^*_\beta, x_\beta)}\UE(q \mid y) > u(r^*_\beta)$. By Lemma \ref{lem:closeworse}, we can choose $\tilde{r}_\beta$ arbitrarily close to $r^*_\beta$ such that $r^*_\beta$ is worse than $\tilde{r}_\beta$. 
If $\tilde{r}_\beta$ is sufficiently close to $r^*_\beta$, then it will satisfy the following three conditions. First, by Existence, $\UE(q \mid y)$ and $\UE(\tilde{r}_\beta \mid y)$ exist for all $y \in \text{co}(q, \tilde{r}_\beta, x_\beta)$. Second, by Continuity, $\min_{y \in \text{co}(q, \tilde{r}_\beta, x_\beta)}\UE(q \mid y) \geq u(\tilde{r}_\beta)$. This implies $q \in c_2(q, \tilde{r}_\beta \mid q, \tilde{r}_\beta, x_\beta)$. Third, 
\[\tilde{r} \equiv \frac{1}{\beta}\tilde{r}_\beta - \frac{1-\beta}{\beta}q \in \text{int}(\Delta(Z)).\] Since $q \in c_2(q, \tilde{r}_\beta \mid q, \tilde{r}_\beta, x_\beta)$, Linearity implies $q \in c_2(q, \tilde{r} \mid q, \tilde{r}, x)$. Suppose that $\tilde{r}$ is worse than some $y \in \text{co}(q, x)$. By Lemma \ref{lem:convexworse}, $\tilde{r}_\beta$ is worse than $y_\beta \in \text{co}(q, x_\beta)$. By Monotonicity, $\UE(r^*_\beta \mid \tilde{r}_\beta) > \UE(r^*_\beta \mid y_\beta)$. Since $r^*_\beta$ is not worse than any member of $\text{co}(q, x_\beta)$, Lemma \ref{lem:regretfree} implies $\UE(r^*_\beta \mid y_\beta) = u(r^*_\beta)$. Thus, $\UE(r^*_\beta \mid \tilde{r}_\beta) > u(r^*_\beta)$. This contradicts Lemma \ref{lem:UEbelow}. Conclude that $\tilde{r}$ is not worse than any $y \in \text{co}(q, x)$. Since $q \in c_2(q, \tilde{r} \mid q, \tilde{r}, x)$, we have $\UE(q \mid x) \geq u(\tilde{r})$. Since $r^*_\beta$ is worse than $\tilde{r}_\beta$, Lemma \ref{lem:utilworse} implies $u(r^*_\beta) < u(\tilde{r}_\beta)$. This implies $u(r^*) < u(\tilde{r})$, so $\UE(q \mid x) > u(r^*)$. This contradicts the definition of $r^*$. Conclude that $\{q, r^*\} = c_2(q, r^* \mid q, r^*, x)$. By Linearity, $\{q, r^*_\beta\} = c_2(q, r^*_\beta \mid q, r^*_\beta, x_\beta)$. By Rationalization and Lemma \ref{lem:regretfree}, $\min_{y \in \text{co}(q, r^*_\beta, x_\beta)}\UE(q \mid y) = u(r^*_\beta)$. 

We break the rest of the proof into two cases. \textbf{Case 1:}
\[x_\beta \notin \argmin_{y \in \text{co}(q, r^*_\beta, x_\beta)}\UE(q \mid y).\]

First, we show that each member of $\argmin_{y \in \text{co}(q, r^*_\beta, x_\beta)}\UE(q \mid y)$ can be written $\lambda r^*_\beta + (1-\lambda)x_\beta$ for some $\lambda \in (0, 1]$. Suppose otherwise, so we have $\lambda q + (1-\lambda)y^* \in \argmin_{y \in \text{co}(q, r^*_\beta, x_\beta)}\UE(q \mid y)$ for some $\lambda \in (0, 1]$ and some $y^* \in \text{co}(r^*_\beta, x_\beta)$. Since $q$ is worse than $x$, Lemma \ref{lem:convexworse} implies that $q$ is worse than $x_\beta$. By Lemma \ref{lem:regret}, $\min_{y \in \text{co}(q, r^*_\beta, x_\beta)}\UE(q \mid y) < u(q)$. Thus, $\UE(q \mid \lambda q + (1-\lambda)y^*) < u(q)$. By Lemma \ref{lem:regretfree}, $q$ is worse than $\lambda q + (1-\lambda)y^*$. By Lemma \ref{lem:convexworse}, $q$ is worse than $y^*$. Applying Lemma \ref{lem:convexworse} again, $\lambda q + (1-\lambda)y^*$ is worse than $y^*$. By Monotonicity, $\UE(q \mid \lambda q + (1-\lambda)y^*) > \UE(q \mid y^*)$. This contradicts the assumption about $\lambda q + (1-\lambda)y^*$. 

Next, we show that there exists $\epsilon > 0$ such that  $B_\epsilon(r_\beta^*) \subset S$ and no member of $B_\epsilon(r_\beta^*)$ is worse than any member of $\text{co}(q, x_\beta)$. Suppose otherwise, so there exist lotteries arbitrarily close to $r^*_\beta$ that are worse than some member of $\text{co}(q, x_\beta)$. Fix any $\tilde{r}_\beta \in S$ such that $\tilde{r}_\beta$ is worse than $\lambda q + (1-\lambda)x_\beta$ for some $\lambda \in [0, 1]$. Since $q$ is worse than $x_\beta$, Lemma \ref{lem:convexworse} implies that $\lambda q + (1-\lambda)x_\beta$ is worse than $x_\beta$. By Monotonicity, $\UE(\tilde{r}_\beta \mid \lambda q + (1-\lambda)x_\beta) > \UE(\tilde{r}_\beta \mid x_\beta)$. By Lemma \ref{lem:UEbelow}, $u(\tilde{r}_\beta) \geq \UE(\tilde{r}_\beta \mid \lambda q + (1-\lambda)x_\beta)$. Thus, $u(\tilde{r}_\beta) > \UE(\tilde{r}_\beta \mid x_\beta)$. By Lemma \ref{lem:regretfree}, $\tilde{r}_\beta$ is worse than $x_\beta$. Thus, there exist lotteries arbitrarily close to $r^*_\beta$ that are worse than $x_\beta$. 

Fix any $\lambda \in (0, 1]$ such that $\lambda r^*_\beta + (1-\lambda)x_\beta \in \argmin_{y \in \text{co}(q, r^*_\beta, x_\beta)}\UE(q \mid y)$. We have that $u(q) > \UE(q \mid \lambda r^*_\beta + (1-\lambda)x_\beta)$. Recall that we can find $\tilde{r}_\beta$ arbitrarily close to $r^*_\beta$ such that $\tilde{r}_\beta$ is worse than $x_\beta$. By Continuity, we can ensure $u(q) > \UE(q \mid \lambda \tilde{r}_\beta + (1-\lambda)x_\beta)$ by choosing $\tilde{r}_\beta$ sufficiently close to $r^*_\beta$. By Lemma \ref{lem:regretfree}, $q$ is worse than $\lambda \tilde{r}_\beta + (1-\lambda)x_\beta$. Since $\tilde{r}_\beta$ is worse than $x_\beta$, Lemma \ref{lem:convexworse} implies that $\lambda \tilde{r}_\beta + (1-\lambda)x_\beta$ is worse than $x_\beta$. By Monotonicity, $\UE(q \mid \lambda \tilde{r}_\beta + (1-\lambda)x_\beta) > \UE(q \mid x_\beta)$. Since this holds for $\tilde{r}_\beta$ arbitrarily close to $r^*_\beta$, Continuity implies $\UE(q \mid \lambda r^*_\beta + (1-\lambda)x_\beta) \geq \UE(q \mid x_\beta)$. This contradicts the assumption that $x_\beta \notin \argmin_{y \in \text{co}(q, r^*_\beta, x_\beta)}\UE(q \mid y)$. 

By Lemma \ref{lem:closeworse}, we can find $\hat{r}_\beta$ arbitrarily close to $r^*_\beta$ such that $\hat{r}_\beta$ is worse than $r^*_\beta$. We have that $u(q) > \UE(q \mid \lambda r^*_\beta + (1-\lambda)x_\beta)$. By Continuity, we can ensure $u(q) > \UE(q \mid \lambda \hat{r}_\beta + (1-\lambda)x_\beta)$ by choosing $\hat{r}_\beta$ sufficiently close to $r^*_\beta$. Fix $\lambda \in [0, 1]$ such that $\lambda \hat{r}_\beta + (1-\lambda)x_\beta \in \argmin_{y \in \text{co}(q, \hat{r}_\beta, x_\beta)}\UE(q \mid y)$. If $\lambda = 0$, we have
\[\min_{y \in \text{co}(q, \hat{r}_\beta, x_\beta)}\UE(q \mid y) = \UE(q \mid x_\beta) > \min_{y \in \text{co}(q, r^*_\beta, x_\beta)}\UE(q \mid y).\]
Now suppose $\lambda > 0$. Since $\hat{r}_\beta$ is worse than $r^*_\beta$, Lemma \ref{lem:convexworse} implies that $\lambda \hat{r}_\beta + (1-\lambda)x_\beta$ is worse than $\lambda r^*_\beta + (1-\lambda)x_\beta$. By Monotonicity, $\UE(q \mid \lambda \hat{r}_\beta + (1-\lambda)x_\beta) > \UE(q \mid \lambda r^*_\beta + (1-\lambda)x_\beta)$. This implies
\[\min_{y \in \text{co}(q, \hat{r}_\beta, x_\beta)}\UE(q \mid y) > \UE(q \mid \lambda r^*_\beta + (1-\lambda)x_\beta) \geq \min_{y \in \text{co}(q, r^*_\beta, x_\beta)}\UE(q \mid y).\]
Thus, regardless of the value of $\lambda$, 
\begin{equation}
\label{eq:hatr_rstar}
\min_{y \in \text{co}(q, \hat{r}_\beta, x_\beta)}\UE(q \mid y) > \min_{y \in \text{co}(q, r^*_\beta, x_\beta)}\UE(q \mid y).
\end{equation}

By definition of $r^*$, we have $u(r^*) = \UE(q \mid x)$. Since $\UE(q \mid x) < u(q)$, we have $u(r^*) < u(q)$. This implies $u(r^*_\beta) < u(q)$. Since $\hat{r}_\beta$ is worse than $r^*_\beta$ and $q$ is worse than $x_\beta$, Lemma \ref{lem:utilworse} implies $u(\hat{r}_\beta) < u(r^*_\beta) < u(q) < u(x_\beta)$. Thus, we can find $\bar{r}_\beta \in \text{co}(\hat{r}_\beta, x_\beta)$ such that $u(r^*_\beta) < u(\bar{r}_\beta) < u(q)$. By choosing $\hat{r}_\beta$ sufficiently close to $r^*_\beta$, we can require $\bar{r}_\beta \in B_\epsilon(r^*_\beta)$. This ensures that $\bar{r}_\beta$ is not worse than any member of $\text{co}(q, x_\beta)$. Since $r^*$ is interior, we can also require that
\[\bar{r} \equiv \frac{1}{\beta}\bar{r}_\beta - \frac{1-\beta}{\beta}q \in \text{int}(\Delta(Z)).\]
By Lemma \ref{lem:convexworse}, $\bar{r}$ will not be worse than any member of $\text{co}(q, x)$. 

Since $u(q) > u(\bar{r}_\beta) > u(\hat{r}_\beta)$, Lemma \ref{lem:utilworse} implies that $q$ is not worse than any $y \in \text{co}(\hat{r}_\beta, \bar{r}_\beta)$. By Lemma \ref{lem:self}, $u(q) = \UE(q \mid y)$ for all $y \in \text{co}(\hat{r}_\beta, \bar{r}_\beta)$. Since $u(q) > \UE(q \mid x_\beta)$, no member of $\text{co}(\hat{r}_\beta, \bar{r}_\beta)$ can belong to $\argmin_{y \in \text{co}(q, \hat{r}_\beta, x_\beta)}\UE(q \mid y)$. Thus, 
\[\min_{y \in \text{co}(q, \hat{r}_\beta, x_\beta)}\UE(q \mid y) = \min_{y \in \text{co}(q, \bar{r}_\beta, x_\beta)}\UE(q \mid y).\]
Combining this with (\ref{eq:hatr_rstar}), we get
\[\min_{y \in \text{co}(q, \bar{r}_\beta, x_\beta)}\UE(q \mid y) > \min_{y \in \text{co}(q, r^*_\beta, x_\beta)}\UE(q \mid y).\]
Since $\min_{y \in \text{co}(q, r^*_\beta, x_\beta)}\UE(q \mid y) = u(r^*_\beta)$, we have
\[\min_{y \in \text{co}(q, \bar{r}_\beta, x_\beta)}\UE(q \mid y) > u(r^*_\beta).\]
By choosing $\bar{r}_\beta$ so that $u(\bar{r}_\beta)$ is sufficiently close to $u(r^*_\beta)$, we can require
\[\min_{y \in \text{co}(q, \bar{r}_\beta, x_\beta)}\UE(q \mid y) \geq u(\bar{r}_\beta).\]
By Rationalization, $q \in c_2(q, \bar{r}_\beta \mid q, \bar{r}_\beta, x_\beta)$. By Linearity, $q \in c_2(q, \bar{r} \mid q, \bar{r}, x)$. Since $\bar{r}$ is interior and is not worse than any member of $\text{co}(q, x)$, we have $\UE(q \mid x) \geq u(\bar{r}) > u(r^*)$. This contradicts the definition of $u(r^*)$. 

\textbf{Case 2:}
\[x_\beta \in \argmin_{y \in \text{co}(q, r^*_\beta, x_\beta)}\UE(q \mid y).\]
Since $\min_{y \in \text{co}(q, r^*_\beta, x_\beta)}\UE(q \mid y) = u(r^*_\beta)$, we have 
\[u(r^*_\beta) = \UE(q \mid x_\beta) \geq \min_{y \in \text{co}(q, p^*_\beta, x_\beta)}\UE(q \mid y).\]
Since $p^*_\alpha$ is not worse than any member of $\text{co}(q, x_\alpha)$ and since $\alpha \leq \beta$, Lemma \ref{lem:convexworse} implies that $p^*_\beta$ is not worse than any member of $\text{co}(q, x_\beta)$. By Lemma \ref{lem:regretfree}, $\min_{y \in \text{co}(q, p^*_\beta, x_\beta)}\UE(p^*_\beta \mid y) = u(p^*_\beta)$. Since $u(p^*_\beta) > u(r^*_\beta)$, we have
\[u(p^*_\beta) > \min_{y \in \text{co}(q, p^*_\beta, x_\beta)}\UE(q \mid y).\]
By Rationalization, $q \notin c_2(q, p^*_\beta \mid q, p^*_\beta, x_\beta)$. By Linearity, $q \notin c_2(q, p^*_\alpha \mid q, p^*_\alpha, x_\alpha)$. This contradicts the definition of $p^*_\alpha$.
\end{proof}

\begin{lemma}
\label{lem:shift}
For any $p, q, r, s \in S$: if $p - s = q - x$, then $\UE(p \mid s) - u(p) = \UE(q \mid x) - u(q)$.
\end{lemma}

\begin{proof}
Fix any $y$ in
\begin{multline}
\argmax\{u(r): r \text{ is interior and is not worse than any member of } \text{co}(q, x), \\\text{ and } q \in c_2(q, r \mid q, r, x).
\end{multline}
We have $\UE(q \mid x) = u(y)$. Fix any $\alpha \in (0, 1]$ such that $p + \alpha(y - q) \in S$, $p + \alpha(x -q) \in S$, and $p + \frac{\alpha}{1-\alpha}(p - q) \in \Delta(Z)$. Since $y$ is not worse than $\beta q + (1-\beta)x$ for any $\beta \in [0, 1]$, Lemma \ref{lem:convexworse} implies that $\alpha y + (1-\alpha)(p + \frac{\alpha}{1-\alpha}(p - q))$ is not worse than $\alpha(\beta q + (1-\beta)x) + (1-\alpha)(p + \frac{\alpha}{1-\alpha}(p - q))$. Rearranging these lotteries, we have that $p + \alpha(y - q)$ is not worse than $\beta p + (1-\beta)(p + \alpha(x-q))$ for any $\beta \in [0, 1]$. By Linearity and the fact that $q \in c_2(q, y \mid q, y, x)$, we have that 
\[p \in c_2(p, p + \alpha(y - q) \mid p, p + \alpha(y - q), p + \alpha(x-q)).\] 
Thus, $\UE(p \mid p + \alpha(x-q)) \geq u(p + \alpha(y - q))$. This can be rewritten 
\[\UE(p \mid \alpha (p + x - q) + (1-\alpha)p) \geq u(p) - \alpha u(q) + \alpha \UE(q \mid x).\]
Applying Lemma \ref{lem:mix} gives
\[\UE(p \mid p + x - q) \geq \UE(q \mid x) + u(p) - u(q).\]
Since $p - s = q - x$, we have
\[\UE(p \mid s) - u(p) \geq \UE(q \mid x) - u(q).\]
A symmetric argument establishes the opposite inequality. 
\end{proof}

Fix some $q \in S$. By Lemma \ref{lem:closeworse}, there exists $x \in S$ such that $q$ is worse than $x$. By Quasiconvexity and Continuity, the set $L(q, x) \equiv \{y \in S: \UE(q \mid y) \leq \UE(q \mid x)\}$ is closed and convex. By Lemma \ref{lem:closeworse}, there exists $x'$ arbitrarily close to $x$ such that $x'$ is worse than $x$. By choosing $x'$ sufficiently close to $x$, we can ensure that $x' \in S$. Since $q$ is worse than $x$, we have $\UE(q \mid x) < u(q)$ by Lemma \ref{lem:regret}. By Continuity, we have $\UE(q \mid x') < u(q)$ provided $x'$ is close enough to $q$. By Lemma \ref{lem:regretfree}, we have that $q$ is worse than $x'$. By Monotonicity, $\UE(q \mid x') > \UE(q \mid x)$, so $x' \notin L(q, x)$. Since we can find $x'$ arbitrarily close to $x$ that does not belong to $L(q, x)$, we have that $x$ is on the boundary of $L(q, x)$. Fix any hyperplane $H_x$ that supports $L(q, x)$ at $x$. We show that $H_x$ does not contain $q$. By Monotonicity, every lottery $x' \in S$ such that $x$ is worse than $x'$ belongs to $L(q, x)$. Since $q$ is worse than $x$, Lemma \ref{lem:convexworse} implies that $q$ is worse than $\alpha x + (1-\alpha)q$ for all $\alpha \in (0, 1]$. By Lemma \ref{lem:regret}, $\UE(q \mid \alpha x + (1-\alpha)q) < u(q)$. For $\alpha$ sufficiently small, $x + \alpha(x - q) \in S$. By Lemma \ref{lem:shift}, $\UE(x \mid x + \alpha(x - q)) < u(x)$. By Lemma \ref{lem:regretfree}, $x$ is worse than $x + \alpha(x-q)$. By Monotonicity, $\UE(q \mid x) > \UE(q \mid x + \alpha(x-q))$. By Continuity, $x + \alpha(x-q)$ is in the interior of $L(q, x)$. since $H_x$ supports $L(q, x)$, we have $x + \alpha(x-q) \notin H_x$. Since $x \in \text{co}(x + \alpha(x-q), q)$ and $x \in H_x$, we have that $q \notin H_x$. 
Since $q \notin H_x$, there exists an expected-utility preference $\succsim_{H_x}$ that has indifference curve $H_x$ through $x$ and that strictly prefers $x$ to $q$. The representation of $\succsim_{H_x}$ is unique up to a positive affine transformation. Let $w_{H_x}$ denote the representation that satisfies
\begin{align}
&w_{H_x}(q) = u(q) \\
&w_{H_x}(x) - w_{H_x}(q) = u(q) - \UE(q \mid x).
\end{align}

Let
\[\hat{\mathcal{W}} \equiv \{w_{H_x} \in \mathcal{U}: x \in S, \; q \text{ is worse than } x,\; H_x \text{ supports } L(q, x) \text{ at } x\}.\]
The following lemmas establish useful properties of $\hat{\mathcal{W}}$.

\begin{lemma}
\label{lem:correctranking}
For any $y \in S$ and any $w_{H_x} \in \hat{\mathcal{W}}$: if $q$ is worse than $y$, then $w_{H_x}(q) < w_{H_x}(y)$.
\end{lemma}

\begin{proof}
Suppose $q$ is worse than $y$. Since $x \in S$, there exists $\lambda > 0$ such that $x + \lambda(y-q) \in S$. By Lemma \ref{lem:regret}, $\UE(q \mid y) < u(q)$. By Lemma \ref{lem:mix}, $\UE(q \mid \lambda y + (1-\lambda)q) < u(q)$. By Lemma \ref{lem:shift}, $\UE(x \mid x + \lambda(y-q)) < u(x)$. By Lemma \ref{lem:regretfree}, $x$ is worse than $x + \lambda(y-q)$. Since $q$ is worse than $x$, Monotonicity implies that $\UE(q \mid x) > \UE(q \mid x + \lambda(y - q))$. By Continuity, for all $\tilde{x}$ sufficiently close to $x + \lambda(y-q)$, we have $\UE(q \mid x) > \UE(q \mid \tilde{x})$. Thus, $x + \lambda(y-q)$ is in the interior of $L(q, x)$. By definition of $w_{H_x}$, we have $w_{H_x}(x + \lambda(y-q)) > w_{H_x}(x)$. This implies $w_{H_x}(y) > w_{H_x}(q)$.
\end{proof}

\begin{lemma}
\label{lem:mixwithq}
For any $w_{H_x} \in \hat{\mathcal{W}}$ and any $\alpha \in (0, 1)$: for any $p \in L(q, \alpha x + (1-\alpha)q)$, we have $w_{H_x}(p) \geq w_{H_x}(\alpha x + (1-\alpha)q) \leq w_{H_x}(p)$. 
\end{lemma}

\begin{proof}
Let $x_\alpha \equiv \alpha x + (1-\alpha)q$. We show that $w_{H_x}(p) \geq w_{H_x}(x_\alpha)$ for all $p \in L(q, x_\alpha)$. 

Fix any $p \in L(q, x_\alpha)$. Fix $\lambda \in (0, 1]$ such that $x + (\lambda/\alpha)(x_\alpha - p) \in S$ and $x_\alpha + \lambda(p - x_\alpha) \in S$. Since $L(q, x_\alpha)$ is convex,
$\lambda p + (1-\lambda)x_\alpha \in L(q, x_\alpha)$. By definition of $L(q, x_\alpha)$, we have
\[\UE(q \mid x_\alpha + \lambda(p - x_\alpha)) \leq \UE(q \mid x_\alpha).\]
By Lemma \ref{lem:mix},
\begin{align}
\UE(q \mid x_\alpha + \lambda(p - x_\alpha)) &= \UE\left(q \mid \alpha\left(x + \frac{\lambda}{\alpha}(p - x_\alpha)\right) + (1-\alpha)q\right)\\
&\alpha \UE\left(q \mid x + \frac{\lambda}{\alpha}(p - x_\alpha)\right) + (1-\alpha)u(q).
\end{align}
Also by Lemma \ref{lem:mix},
\[\UE(q \mid x_\alpha) = \alpha \UE(q \mid x) + (1-\alpha)u(q).\]
Thus,
\[\UE\left(q \mid x + \frac{\lambda}{\alpha}(p - x_\alpha)\right) \leq \UE(q \mid x),\]
so $x + (\lambda/\alpha)(p - x_\alpha) \in L(q, x)$. By definition of $w_{H_x}$, 
\[w_{H_x}\left(x + \frac{\lambda}{\alpha}(p - x_\alpha)\right) \geq w_{H_x}(x).\]
Since $w_{H_x}$ is linear and $\lambda > 0$, we have
\[w_{H_x}(p) \geq w_{H_x}(x_\alpha).\]
\end{proof}

\begin{lemma}
\label{lem:min}
For each $w_{H_x}, w_{H_y} \in \hat{\mathcal{W}}$, 
\begin{equation}
\label{eq:min1}
u(q) - \UE(q \mid x) = w_{H_x}(x) - w_{H_x}(q) \leq w_{H_y}(x) - w_{H_y}(q).
\end{equation}
\end{lemma}

\begin{proof}
We break the argument into two cases.

\textbf{Case 1:} Suppose that $\UE(q \mid x) \leq \UE(q \mid y)$. By Lemma \ref{lem:mix}, there exists $\alpha \in (0, 1]$ such that 
\[\UE(q \mid \alpha x + (1-\alpha)q) = \UE(q \mid y).\]
We have
\begin{align}
w_{H_x}(x) - w_{H_x}(q) &= \frac{1}{\alpha}\left(w_{H_x}(\alpha x + (1-\alpha)q) - w_{H_x}(q)\right) \\
&= \frac{1}{\alpha}\left(u(q) - \UE(q \mid \alpha x + (1-\alpha)q)\right) \\
&= \frac{1}{\alpha}\left(u(q) - \UE(q \mid y)\right) \\
&= \frac{1}{\alpha}\left(w_{H_y}(y) - w_{H_y}(q)\right) \\
&\leq \frac{1}{\alpha}\left(w_{H_y}(\alpha x + (1-\alpha)q) - w_{H_y}(q)\right)\\
&= w_{H_y}(x) - w_{H_y}(q).
\end{align}
second equality uses the definition of $w_{H_x}$ and Lemma \ref{lem:mix}, the third uses the definition of $\alpha$, the fourth uses the definition of $w_{H_y}$, and the fifth uses the fact that $\alpha x + (1-\alpha)q \in L(q, y)$. 

\textbf{Case 2:} Suppose that $\UE(q \mid y) \leq \UE(q \mid x)$. By Lemma \ref{lem:mix}, there exists $\alpha \in (0, 1]$ such that
\[\UE(q \mid \alpha y + (1-\alpha)q) = \UE(q \mid x).\]
We have
\begin{align}
w_{H_x}(x) - w_{H_x}(q) &= u(q) - U(q \mid x) \\
&= u(q) - \UE(q \mid \alpha y + (1-\alpha)q) \\
&= \alpha\left(u(q) - \UE(q \mid y)\right) \\
&= \alpha\left(w_{H_y}(y) - w_{H_y}(q) \right)\\
&= w_{H_y}(\alpha y + (1-\alpha)q) - w_{H_y}(q) \\
&\leq w_{H_y}(x) - w_{H_y}(q),
\end{align}
where the first equality uses the definition of $w_{H_x}$, the second uses the definition of $\alpha$, the third uses Lemma \ref{lem:mix}, the fourth uses the definition of $W_{H_y}$, and the sixth uses Lemma \ref{lem:mixwithq} and the fact that $x \in L(q, \alpha y + (1-\alpha)q)$. 
\end{proof}

\begin{lemma}
\label{lem:bounded}
$\hat{\mathcal{W}}$ is bounded.
\end{lemma}

\begin{proof}
Fix any $w_{H_x} \in \hat{\mathcal{W}}$. Let $H_q \equiv H_x + \{q\} - \{x\}$. Let $\hat{w}_{H_x}$ denote the utility that represents the same preference as $w_{H_x}$, but has $\hat{w}_{H_x}(x) - \hat{w}_{H_x}(q) = d(H_q, H_x)$. For all $y \in \Delta(Z)$, we have
\[w_{H_x}(y) - w_{H_x}(q) = \frac{u(q) - \UE(q \mid x)}{d(H_q, H_x)}\left(\hat{w}_{H_x}(y) - \hat{w}_{H_x}(q)\right) = \frac{u(q) - \UE(q \mid x)}{d(H_q, H_x)}d(H_q, H_y).\]
Since the simplex is bounded, it suffices to show that
\[\frac{u(q) - \UE(q \mid x)}{d(H_q, H_x)}\]
is bounded (holding $q$ fixed but allowing $x$ and $H_x$ to vary). 

There exists $\alpha \in (0, 1]$ such that $\text{proj}(\alpha x + (1-\alpha)q \mid H_q) \in S$. Since $H_q$ is the indifference curve of $w_{H_x}$ through $q$, we have that
\[w_{H_x}(q) = w_{H_x}\left(\text{proj}(\alpha x + (1-\alpha)q \mid H_q)\right).\] By Lemma \ref{lem:correctranking}, $q$ is not worse than $\text{proj}(\alpha x + (1-\alpha)q \mid H_q)$. By Lemma \ref{lem:regretfree}, 
\[\UE(q \mid \text{proj}(\alpha x + (1-\alpha)q \mid H_q)) = u(q).\] 
Using this fact and Lemma \ref{lem:mix}, we have
\begin{align}
\frac{u(q) - \UE(q \mid x)}{d(H_q, H_x)} &= \frac{u(q) - \UE(q \mid \alpha x + (1-\alpha)q)}{d(\alpha x + (1-\alpha)q, H_q)} \\
&= \frac{\UE(q \mid \text{proj}(\alpha x + (1-\alpha)q \mid H_q)) - \UE(q \mid \alpha x + (1-\alpha)q)}{d(\alpha x + (1-\alpha)q, \text{proj}(\alpha x + (1-\alpha)q \mid H_q))}.
\end{align}
Since $\UE$ is Lipschitz continuous in its second argument, the final quantity above is bounded.
\end{proof}

Let $\mathcal{W}$ denote the closed convex hull of $\hat{\mathcal{W}}$. By Lemma \ref{lem:bounded}, $\mathcal{W}$ is compact. Since $\hat{\mathcal{W}}$ satisfies (\ref{eq:min1}), we have
\begin{equation}
\label{eq:min2}
u(q) - \UE(q \mid x) = \min_{w \in \mathcal{W}}\left(w(x) - w(q) \right) \text{ for all } x \in S \text{ such that } q \text{ is worse than } x.
\end{equation}

We show that 
\begin{equation}
\label{eq:agree}
\{x \in S: w(x) > w(q) \text{ for all } w \in \mathcal{W}\} = \{x \in S: q \text{ is worse than } x\}.
\end{equation}
First, take any $x \in S$ such that $q$ is worse than $x$. By Lemma \ref{lem:regret}, $u(q) - \UE(q \mid x) > 0$. By (\ref{eq:min2}), $\min_{w \in \mathcal{W}}\left(w(x) - w(q)\right) > 0$. 

Now take any $y \in S$ such that $q$ is not worse than $y$. By Lemma \ref{lem:closeworse}, there exists $x^* \in S$ such that $q$ is worse than $x^*$. There exists $\bar{x} \in \text{co}(y, x^*)$ such that $\bar{x}$ is on the boundary of $\{x \in S: q \text{ is worse than } x\}$. By Lemmas \ref{lem:regret} and \ref{lem:regretfree}, 
\[\{x \in S: q \text{ is worse than } x\} = \{x \in S: u(q) > \UE(q \mid x)\}.\]
By Continuity, the right-hand side is an open set. Since $\bar{x}$ is on the boundary of this set, it is not the case that $q$ is worse than $\bar{x}$. By Lemma \ref{lem:regretfree}, $\UE(q \mid \bar{x}) = u(q)$. There exists a sequence $x_i \to \bar{x}$ such that $q$ is worse than $x_i$ for all $i$. By construction of $\mathcal{W}$, there also exists a sequence $w_{H_{x_i}} \in \mathcal{W}$ such that $w_{H_{x_i}}(x_i) - w_{H_{x_i}}(q) = u(q) - \UE(q \mid x_i)$ for all $i$. Since $\mathcal{W}$ is compact, we can pass to a convergent subsequence of the $w_{H_{x_i}}$. Let $\bar{w}$ denote the limit. By Continuity,
\begin{align}
\bar{w}(\bar{x}) - \bar{w}(q) &= \lim_{i \to \infty}\left(w_{H_{x_i}}(x_i) - w_{H_{x_i}}(q)\right) \\
&= \lim_{i \to \infty}\left(u(q) - \UE(q \mid x_i)\right)\\
&= u(q) - \UE(q \mid \bar{x})\\
&= 0.
\end{align}
Conclude that there exists $\bar{w} \in \mathcal{W}$ such that $\bar{w}(q) \geq \bar{w}(\bar{x})$. Since $\bar{x} \in \text{co}(y, x^*)$ and since $\bar{w}(x^*) > \bar{w}(q) \geq \bar{w}(\bar{x})$, we have that $\bar{w}(q) \geq \bar{w}(y)$. Conclude that (\ref{eq:agree}) holds. 

\begin{lemma}
\label{lem:nonsing}
$\mathcal{W}$ contains representations of at least two distinct preferences, and does not contain a constant utility. 
\end{lemma}

\begin{proof}
Suppose that $\mathcal{W}$ contains representations of only one preference. By (\ref{eq:agree}) and Lemma \ref{lem:utilworse}, $q$ is worse than $x \in S$ if and only if $u(q) < u(x)$. Fix any $x \in S$ such that $q$ is worse than $x$. By Lemma \ref{lem:regret}, $\UE(q \mid x) < u(q)$. By definition of $\UE$, there exists $r$ such that $u(q) > u(r)$ and $r$ is not worse than $q$. This is a contradiction. 

Suppose that $|\mathcal{W}|$ contains a constant utility. By (\ref{eq:agree}), there is no $x \in S$ such that $q$ is worse than $x$. But this contradicts Lemma \ref{lem:closeworse}. 
\end{proof}

We show that $\alpha u \in \mathcal{W}$ for some $\alpha > 0$. By (\ref{eq:agree}) and Lemma \ref{lem:Vpref}, every $\succsim \in \mathcal{U}$ that satisfies
\begin{equation}
\label{eq:correctranking}
q \text{ is worse than } x \quad \Longrightarrow \quad x \succ q
\end{equation}
has a representation in $\mathcal{W}$. By Lemma \ref{lem:utilworse}, the preference represented by $u$ satisfies (\ref{eq:correctranking}). Thus, $\alpha u + \beta \in \mathcal{W}$ for some $\alpha > 0$. Since $w(q) = u(q)$ for all $w \in \mathcal{W}$, we have $\beta = 0$. 

Let
\[\mathcal{V} \equiv \frac{1}{\alpha}\mathcal{W}.\] Since $\alpha u \in \mathcal{W}$, we have $u \in \mathcal{V}$. Since $\mathcal{W}$ is compact and convex, so is $\mathcal{V}$. By Lemma \ref{lem:nonsing}, $\mathcal{V}$ contains representations of at least two distinct preferences and does not contain a constant utility. By (\ref{eq:agree}), $\mathcal{V}$ satisfies
\begin{equation}
\label{eq:agree2}
\{x \in S: v(x) > v(q) \text{ for all } v \in \mathcal{V}\} = \{x \in S: q \text{ is worse than } x\}.
\end{equation}
By (\ref{eq:min2}),
\[u(q) - \UE(q \mid x) = \alpha \min_{v \in \mathcal{V}}\left(v(x) - v(q)\right) \text{ for all } x \in S \text{ such that } q \text{ is worse than } x.\]
Let $\gamma \equiv \frac{\alpha}{1+\alpha}$. Since $\alpha > 0$, we have $\gamma \in (0, 1)$. We also have
\begin{equation}
\label{eq:min3}
u(q) - \UE(q \mid x) = \frac{\gamma}{1-\gamma}\min_{v \in \mathcal{V}}\left(v(x) - v(q)\right) \text{ for all } x \in S \text{ such that } q \text{ is worse than } x.
\end{equation}

We show that $(\gamma, u, \mathcal{V})$ represents $c_2$. Suppose that $x^* \in c_2(M \mid \mathcal{M})$. For some $\alpha \in (0, 1]$, $\alpha \mathcal{M} + (1-\alpha)\{q\} \subset S$. Let $x^*_\alpha \equiv \alpha x^* + (1-\alpha)q$, and likewise for $M$ and $\mathcal{M}$. By Linearity, $x_\alpha^* \in c_2(M_\alpha \mid \mathcal{M}_\alpha)$. By Rationalization, for all $x_\alpha \in M_\alpha$, we have
\begin{equation}
\label{eq:choice}
\min_{y \in \text{co}(\mathcal{M}_\alpha)}\UE(x^*_\alpha \mid y) \geq \min_{y \in \text{co}(\mathcal{M}_\alpha)}\UE(x_\alpha \mid y).
\end{equation}
Fix any $x_\alpha \in M_\alpha$. Suppose that $x_\alpha$ is not worse than any $y \in \text{co}(\mathcal{M}_\alpha)$. By Lemma \ref{lem:regretfree},
\begin{equation}
\label{eq:noregretcase}
\min_{y \in \text{co}(\mathcal{M}_\alpha)}\UE(x_\alpha \mid y) = u(x_\alpha).
\end{equation}
For any $y \in \text{co}(\mathcal{M}_\alpha)$, for some $\beta \in (0, 1]$, we have $q + \beta(y - x_\alpha)$ in $S$. By Lemma \ref{lem:mix}, we have $\UE(x_\alpha \mid (1-\beta)x_\alpha + \beta y) = u(x_\alpha)$. By Lemma \ref{lem:shift}, $\UE(q \mid q + \beta(y - x_\alpha)) = u(q)$. By Lemma \ref{lem:regret}, $q$ is not worse than $q + \beta(y - x_\alpha)$. By (\ref{eq:agree2}), there exists $v \in \mathcal{V}$ such that $v(q) \geq v(q + \beta(y - x_\alpha))$. Thus, there exists $v \in \mathcal{V}$ such that $v(x_\alpha) \geq v(y)$. Since $x_\alpha \in \mathcal{M}_\alpha$, we have
\[\min_{y \in \text{co}(\mathcal{M}_\alpha)}\max_{v \in \mathcal{V}}\left(v(x_\alpha) - v(y)\right) = 0.\]
Combining with (\ref{eq:noregretcase}), we have
\begin{equation}
\label{eq:noregretcase2}
\min_{y \in \text{co}(\mathcal{M}_\alpha)}\UE(x_\alpha \mid y) = u(x_\alpha) + \frac{\gamma}{1-\gamma}\min_{y \in \text{co}(\mathcal{M}}\max_{v \in \mathcal{V}}\left(v(x_\alpha) - v(y)\right).
\end{equation}
Now suppose that $x_\alpha$ is worse than some $y \in \text{co}(\mathcal{M}_\alpha)$. Fix any $y^* \in \text{co}(\mathcal{M}_\alpha)$ such that $x_\alpha$ is worse than $y^*$. For some $\beta \in (0, 1]$, we have $q + \beta(y^* - x_\alpha) \in S$. By Lemma \ref{lem:mix},
\begin{equation}
\label{eq:alphabeta}
\UE(x_\alpha \mid \beta y^* + (1-\beta)x_\alpha) = \beta \UE(x_\alpha \mid y^*) + (1-\beta)u(x_\alpha).
\end{equation}
Since $\UE(x_\alpha \mid y^*) < u(x_\alpha)$, we have $\UE(x_\alpha \mid \beta y^* + (1-\beta)x_\alpha) < u(x_\alpha)$. 
By Lemma \ref{lem:shift}, 
\begin{equation}
\label{eq:alphabeta2}
\UE(q \mid q + \beta(y^* - x_\alpha)) - u(q) = \UE(x_\alpha \mid \beta y^* + (1-\beta)x_\alpha) - u(x_\alpha).
\end{equation}
Since $\UE(x_\alpha \mid \beta y^* + (1-\beta)x_\alpha) < u(x_\alpha)$, we have $U(q \mid q + \beta(y^*- x_\alpha)) < u(q)$. By Lemma \ref{lem:regretfree}, $q$ is worse than $q + \beta(y^*-x_\alpha)$. Combining (\ref{eq:alphabeta}) and (\ref{eq:alphabeta2}), 
\begin{equation}
\label{eq:alphabeta3}
\UE(x_\alpha \mid y^*) = \frac{1}{\beta}\left(\UE(q \mid q + \beta(y^*-x_\alpha))-u(q)\right) + u(x_\alpha).
\end{equation}
By (\ref{eq:min3}) and the fact that $q$ is worse than $q + \beta(y^*-x_\alpha)$, 
\begin{align}
\UE(q \mid q + \beta(y^* - x_\alpha)) - u(q) &= \frac{\gamma}{1-\gamma}\max_{v \in \mathcal{V}}\left(v(q) - v(q + \beta(y^* - x_\alpha))\right) \\
&= \beta \frac{\gamma}{1-\gamma}\max_{v \in \mathcal{V}}\left(v(x_\alpha) - v(y^*)\right).
\end{align}
Substituting into (\ref{eq:alphabeta3}), 
\[\UE(x_\alpha \mid y^*) = u(x_\alpha) + \frac{\gamma}{1-\gamma}\max_{v \in \mathcal{V}}\left(v(x_\alpha) - v(y^*)\right).\]
Since $y^*$ was an arbitrary member of $\{y \in \text{co}(\mathcal{M}_\alpha): x_\alpha \text{ is worse than } y\}$, we have
\[\min_{y \in \text{co}(\mathcal{M}_\alpha)}\UE(x_\alpha \mid y) = u(x_\alpha) + \frac{\gamma}{1-\gamma}\min_{y \in \text{co}(\mathcal{M}_\alpha)}\max_{v \in \mathcal{V}}\left(v(x_\alpha) - v(y)\right).\]
This is the same as (\ref{eq:noregretcase2}), so (\ref{eq:noregretcase2}) holds regardless of whether $x_\alpha$ is worse than some $y \in \text{co}(\mathcal{M}_\alpha)$. Substituting into (\ref{eq:choice}), 
\begin{multline}
(1-\gamma)u(x^*_\alpha) + \gamma \min_{y \in \text{co}(\mathcal{M}_\alpha)}\max_{v \in \mathcal{V}}(v(x^*_\alpha) - v(y)) \\\geq (1-\gamma)u(x_\alpha) + \gamma \min_{y \in \text{co}(\mathcal{M}_\alpha)}\max_{v \in \mathcal{V}}(v(x_\alpha) - v(y)).
\end{multline}
By the Minimax Theorem (which uses the fact that $\mathcal{V}$ is compact and convex), this is equivalent to
\begin{multline}
(1-\gamma)u(x^*_\alpha) + \gamma\max_{v \in \mathcal{V}}\left(v(x^*_\alpha) - \max_{y \in \mathcal{M}_\alpha}v(y)\right) \\\geq (1-\gamma)u(x_\alpha) + \gamma\max_{v \in \mathcal{V}}\left(v(x_\alpha) - \max_{y \in \mathcal{M}_\alpha}v(y)\right).
\end{multline}
Finally, by linearity of $u$ and each $v$, this is equivalent to
\[(1-\gamma)u(x^*) + \gamma\max_{v \in \mathcal{V}}\left(v(x^*) - \max_{y \in \mathcal{M}}v(y)\right) \geq (1-\gamma)u(x) + \gamma\max_{v \in \mathcal{V}}\left(v(x) - \max_{y \in \mathcal{M}}v(y)\right).\]
A parallel argument establishes: for any $\hat{x} \notin c_2(M \mid \mathcal{M})$, there exists $x \in M$ such that
\[(1-\gamma)u(\hat{x}) + \gamma\max_{v \in \mathcal{V}}\left(v(\hat{x}) - \max_{y \in \mathcal{M}}v(y)\right) < (1-\gamma)u(x) + \gamma\max_{v \in \mathcal{V}}\left(v(x) - \max_{y \in \mathcal{M}}v(y)\right).\] Conclude that $(\gamma, u, \mathcal{V})$ represents $c_2$. 

\section{Proof of Theorem \ref{thm:id}}\label{app:id}

Since the preference represented by material utility is pinned down by
\[q \succsim_u x \quad \Longleftrightarrow \quad q \in c_2(q, x \mid q, x),\]
material utility is pinned down up to a positive affine transformation. Thus, there exist $\alpha > 0$ and $\beta_u \in \mathbb{R}$ such that $u_2 = \alpha u_1 + \beta_u$. 

Fix any interior $q$. By Lemma \ref{lem:V2worse}, we have
\[\{r \in \Delta(Z): r \text{ is worse than } q\} = \{r \in \Delta(Z): v_i(q) > v_i(r) \text{ for all } v \in \mathcal{V}_i\} \text{ for i = 1, 2}.\]
By Lemma \ref{lem:Vpref}, a preference $\succsim$ has a representation in $\mathcal{V}_i$ if and only if $\succsim$ satisfies
\[r \text{ is worse than } q \quad \Longrightarrow \quad q \succsim r.\]
Thus, $\mathcal{V}_1$ and $\mathcal{V}_2$ represent the same set of preferences, $\mathcal{V}_\text{pref}$. We have
\[\mathcal{V}_\text{pref} = \{\succsim \in \mathcal{U}: q \succsim r \text{ for any } r \text{ that is worse than } q\}.\]

Now we show that, for any $q, x \in \Delta(Z)$ such that $q$ is worse than $x$, the two representations agree on the set
\[L(q, x) \equiv \{y \in \Delta(Z): U(q \mid q, x) \geq U(q \mid q, y)\}.\]
Fix any $y$ such that $U_1(q \mid q, x) \geq U_1(q \mid q, y)$. Fix $r \in S$. For all $\epsilon \in (0, 1]$, let $y_\epsilon \equiv \epsilon y + (1-\epsilon)r$, and likewise for $q$ and $x$. For $\epsilon$ sufficiently small, we have $q_\epsilon, x_\epsilon, y_\epsilon \in S$. Since $U_1(q_\epsilon \mid q_\epsilon, x_\epsilon) \geq U_1(q_\epsilon \mid q_\epsilon, y_\epsilon)$, Lemma \ref{lem:UE2U} implies $u_1(\ME(q_\epsilon \mid x_\epsilon)) \geq u_1(\ME(q_\epsilon \mid y_\epsilon))$. Since $u_1$ and $u_2$ represent the same preference, we have $u_2(\ME(q_\epsilon \mid x_\epsilon)) \geq u_2(\ME(q_\epsilon \mid y_\epsilon))$, which implies $U_2(q_\epsilon \mid q_\epsilon, x_\epsilon) \geq U_2(q_\epsilon \mid q_\epsilon, y_\epsilon)$ by Lemma \ref{lem:UE2U}. In turn, this implies $U_2(q \mid q, x) \geq U_2(q \mid q, y)$. 

\begin{lemma}
Suppose that $c_2$ has a rationalization representation $(\gamma, u, \mathcal{V})$. For any $q \in \text{int}(\Delta(Z))$, the set of preferences with representations in $\mathcal{V}$ is the closure of the set of preferences that satisfy
\begin{equation}
\label{eq:tangency}
\text{ for some } x \in \text{int}(\Delta(Z)): \quad y \in L(q, x) \quad \Longrightarrow \quad y \succsim x. 
\end{equation}
\end{lemma}

\begin{proof}
First, take any $\succsim$ that has a representation $v \in \mathcal{V}$. Suppose that there exists $p \in \text{int}(\Delta(Z))$ such that $p \in \min_{L(q, x)}v$. We have
\[y \in L(q, x) \quad \Longrightarrow \quad y \succsim p.\]
We claim that $L(q, x) = L(q, p)$. It suffices to show that $U(q \mid q, x) = U(q \mid q, p)$. Since $p \in L(q, x)$, we have $U(q \mid q, x) \geq U(q \mid q, p)$. Suppose that the inequality holds strictly. By Lemma \ref{lem:closeworse_rep} and interiority of $p$, there exists $\underline{p}$ arbitrarily close to $p$ such that $\underline{p}$ is worse than $p$. By choosing $\underline{p}$ sufficiently close to $p$, we can ensure $U(q \mid q, x) \geq U(q \mid q, \underline{p})$. Thus, $\underline{p} \in L(q, x)$. Since $v \in \mathcal{V}$ and $\underline{p}$ is worse than $p$, we have $v(\underline{p}) < v(p)$. This contradicts the definition of $p$. Conclude that $L(q, x) = L(q, p)$, so $\succsim$ satisfies (\ref{eq:tangency}). 

Suppose that $\hat{\succsim} \in \mathcal{V}_\text{pref}$ does not satisfy (\ref{eq:tangency}). Choose any representation $\hat{v}$ for $\hat{\succsim}$. For all $n \in \mathbb{N}$, let
\[L_n(q) \equiv \left\{x \in \Delta(Z): u(q) - U(q \mid q, x) \geq \frac{1}{n}\right\}.\]
Since each $L_n(q)$ is compact, $\argmin_{L_n(q)}\hat{v}$ is nonempty for each $n$. For each $n$, fix some $x_n \in \argmin_{L_n(q)}\hat{v}$. By the previous argument and the assumption that $\hat{\succsim}$ does not satisfy (\ref{eq:tangency}), we have $x_n \notin \text{int}(\Delta(Z))$ for all $n$. Pass to a convergent subseqence of $\{x_n\}$, and let $\bar{x}$ denote the limit. Since each $x_n$ is on the boundary of $\Delta(Z)$, so is $\bar{x}$. Since $q \in \text{int}(\Delta(Z))$, we have $\bar{x} \neq q$. We show that $\hat{v}(q) = \hat{v}(\bar{x})$. Fix any $p$ such that $u(q) - U(q \mid p, q) > 0$. Let
\[\alpha_n \equiv \frac{1}{n\left(u(q) - U(q \mid q, p) \right)}.\]
Let $p_n \equiv \alpha_n p + (1-\alpha_n)q$. We have that
\[u(q) - U(q \mid q, p_n) = \alpha_n\left(u(q) - U(q \mid q, p)\right) = \frac{1}{n}.\]
Thus, $p_n \in L_n(q)$. Since $x_n \in \argmin_{L_n(q)}\hat{v}$, we have
\[\hat{v}(q) < \hat{v}(x_n) \leq \hat{v}(p_n).\]
Since $p_n \to q$ and $x_n \to \bar{x}$, we have $\hat{v}(q) = \hat{v}(\bar{x})$. Thus, $q \; \hat{\sim} \; \bar{x}$. Since $\hat{\succsim} \in \mathcal{V}_\text{pref}$, we have that $q \notin W(\bar{x})$. But since $q$ is worse than $x_n$ for all $n$, we have that $q + (1/n)(\bar{x} - x_n)$ is worse than $\bar{x}$ for all $n$ sufficiently large, so $q \in \bar{W}(\bar{x})$. 

We can use the same arguments from the proof of Lemma \ref{lem:Vpref} to show that there exists $K \in \{0, \ldots, n-3\}$ that satisfies two conditions. First, for each $k \in \{1, \ldots, K\}$, there exists $\rho_k \in \Delta(z_1, \ldots, z_{n-K})$ such that
\[\rho_k \sim z_{n-k+1} \text{ for all } \succsim \in \mathcal{V}_\text{pref}.\]
Second, there exists $\succsim^*$ such that
\begin{equation}
\label{eq:prefstar1}
x \in \bar{W}(q_K) \setminus \{q_K\} \quad \Longrightarrow \quad q_K \succ^* x \text{ for all } x \in \Delta(Z_K),
\end{equation}
where $Z_K = \{z_1, \ldots, z_{n-K}\}$ and 
\[q_K \equiv q + \sum_{k = 1}^K q(z_{n-k+1})\left(\rho_k - \delta_{z_{n-k+1}}\right).\]
For any $x \in \Delta(Z_K)$, let $L_K(q_K, x)$ denote the restriction of $L(q_K, x)$ to $\Delta(Z_K)$. Suppose $\succsim^*$ does not satisfy (\ref{eq:tangency}) on $\Delta(Z_K)$. The previous argument implies the existence of $x \in \Delta(Z_K)$ such that $x \in \bar{W}(q_K) \setminus \{q_K\}$ and $q_K \sim^* x$. But the definition of $\succsim^*$ rules this out. Conclude that there exists $x^* \in \text{int}(\Delta(Z_K))$ such that
\begin{equation}
\label{eq:prefstar2}
y \in L_K(q_K, x^*) \quad \Longrightarrow \quad y \succsim^* x^*.
\end{equation}
For some $\epsilon > 0$, 
\[x^{**} \equiv x + \epsilon \sum_{k = 1}^K \left( \delta_{z_{n-k+1}}-\rho_k\right) \in \text{int}(\Delta(Z)).\]
Notice that $x^{**} \in \text{int}(\Delta(Z))$. 
We claim that 
\[y \in L(q, x^{**}) \quad \Longrightarrow \quad y \succsim^* x^{**}.\]
Fix some $y \in L(q, x^{**})$, and let
\[y_K \equiv y + \sum_{k = 1}^K y(z_{n-k+1})\left(\rho_k - \delta_{z_{n-k+1}}\right).\] We have that $U(q \mid y) \geq U(q \mid x^{**})$. Since $q \sim q_K$, $y \sim y_K$, and $x^* \sim x^{**}$ for all $\succsim \in \mathcal{V}_\text{pref}$, we have that $U(q_K \mid y_K) \geq U(q_K \mid x^*)$. That is, $y_K \in L(q_K, x^*)$. By (\ref{eq:prefstar2}), $y_K \succsim^* x^*$. Since $\succsim^* \in \mathcal{V}_\text{pref}$, we have $y_K \sim^* y$ and $x^* \sim^* x^{**}$, so $y \succsim^* x^{**}$.

Fix any representation $v^*$ for $\succsim^*$, and let
\[v_n \equiv \frac{1}{n}v^* + \left(1-\frac{1}{n} \right)\hat{v}.\] 
Let $\succsim_n$ denote the preference represented by $v_n$. Since $\succsim^*$ satisfies (\ref{eq:prefstar1}), so does each $\succsim_n$. By the previous argument, each $\succsim_n$ satisfies (\ref{eq:tangency}). Since $v_n \to v$ and $v$ represents $\hat{\succsim}$, we have that $\hat{\succsim}$ is a limit of preferences satisfying (\ref{eq:tangency}). 
\end{proof}

Fix any $q \in S$ and any $\succsim \in \mathcal{V}_\text{pref}$ satisfying (\ref{eq:tangency}). Fix $\epsilon$ such that $x_\epsilon \equiv \epsilon x + (1-\epsilon)q \in S$. We show that
\[y_\epsilon \in L(q, x_\epsilon) \quad \Longrightarrow \quad y_\epsilon \succsim x_\epsilon.\]
Fix any $y_\epsilon \in L(q, x_\epsilon)$, and suppose that $x_\epsilon \succ y_\epsilon$. Since $L(q, x_\epsilon)$ is convex, we have $\alpha y_\epsilon + (1-\alpha)x_\epsilon \in L(q, x_\epsilon)$ for all $\alpha \in (0, 1]$. We also have $x_\epsilon \succ \alpha y_\epsilon + (1-\alpha)x_\epsilon$. Thus, it is without loss to assume that $y_\epsilon$ is close enough to $x_\epsilon$ that 
\[y \equiv \frac{1}{\epsilon}y_\epsilon-\frac{1-\epsilon}{\epsilon}q \in \Delta(Z).\]
Since $y_\epsilon \in L(q, x_\epsilon)$, we have $y \in L(q, x)$. Since $x_\epsilon \succ y_\epsilon$, we have $x \succ y$. This contradicts the definition of $\succsim$. Thus, $\mathcal{V}_\text{pref}$ is the closure of the set of preferences that satisfy 
\begin{equation}
\label{eq:tangency2}
\text{for some } x \in S: \quad y \in L(q, x) \quad \Longrightarrow \quad y \succsim x.
\end{equation}

Fix any $\succsim^* \in \mathcal{V}_\text{pref}$ that satisfies
\[y \in L(q, x^*) \quad \Longrightarrow \quad y \succsim x^*\]
for some $x^* \in S$. Fix any menu $X$ such that $x^*$ is in the interior of $\text{co}(X)$ relative to the indifference curve of $\succsim^*$ through $x^*$. Since no $x \in X$ is in the interior of $L(q, x^*)$, all representations satisfy
\[U(q \mid q, X) = \min_{x \in X}U(q \mid q, x) = U(q \mid q, x^*).\]
Suppose that, in some representation, no utility representing $\succsim^*$ belongs to \[\argmax_{v \in \mathcal{V}}\left(v(q) - v(x^*)\right).\] For any $\hat{v}$ that does belong to the argmax, we have
\[\max_{v \in \mathcal{V}}(v(q) - v(x^*)) = \hat{v}(q) - \hat{v}(x^*) > \hat{v}(q) - \max_{x \in X}\hat{v}(x).\]
For any $\hat{v}$ outside the argmax, we have
\[\max_{v \in \mathcal{V}}(v(q) - v(x^*)) > \hat{v}(q) - \hat{v}(x^*) \geq \hat{v}(q) - \max_{x \in X}\hat{v}(x).\]
Putting both cases together,
\[\max_{v \in \mathcal{V}}\left(v(q) - v(x^*)\right) > \max_{v \in \mathcal{V}}\left(v(q) - \max_{x \in X}v(x)\right),\]
which implies $U(q \mid q, x^*) > U(q \mid q, X)$. This is a contradiction. Thus, in each representation $i$, there exists $v^*_i \in \mathcal{V}_i$ that represents $\succsim^*$ and belongs to 
\[\argmax_{v \in \mathcal{V}_i}\left(v(q) - v(x^*) \right).\] 
Each $v^*_i$ is non-redundant. We have
\begin{align}
\frac{\gamma_2}{1-\gamma_2}\left(v_2^*(x^*) - v_2^*(q)\right) &= u_2(q) - \frac{1}{1-\gamma_2}U_2(q \mid x^*) \\
&= u_2(q) - u_2(\ME(q \mid x^*)) \\
&= \alpha\left(u_1(q) - u_1(\ME(q \mid x^*)) \right) \\
&= \alpha\left(u_1(q) - \frac{1}{1-\gamma_1}U_1(q \mid x^*) \right) \\
&= \alpha \frac{\gamma_1}{1-\gamma_1}\left(\left(v_1^*(x^*) - v_1^*(q)\right)\right)
\end{align}
where the second and fourth equalities use Lemma \ref{lem:UE2U}. Since $v_1^*$ and $v_2^*$ represent the same preference, this implies
\[v_2^* = \alpha \frac{\gamma_1(1-\gamma_2)}{(1-\gamma_1)\gamma_2}v_1^* + \beta^* \text{ for some } \beta^* \in \mathbb{R}.\]

Now consider an arbitrary $\succsim^* \in \mathcal{V}_\text{pref}$. We showed that there exists $\{\succsim_n\}$ converging to $\succsim^*$ such that each $\succsim_n$ satisfies (\ref{eq:tangency2}). By the previous argument, each $\succsim_n$ has representations $v^n_1 \in \underline{\mathcal{V}_1}$ and $v^n_2 \in \underline{\mathcal{V}_2}$ such that
\[\text{ for some } x^n \in S, \text{ for both } i \in \{1, 2\}: \quad v^n_i \in \argmax_{v \in \mathcal{V}_1}\left(v(q) - v(x^n) \right).\]
and 
\[v_2^n = \alpha \frac{\gamma_1(1-\gamma_2)}{(1-\gamma_1)\gamma_2}v_1^n + \beta^n \text{ for some } \beta^n \in \mathbb{R}.\]
Since $\mathcal{V}_1$ and $\mathcal{V}_2$ are compact, it is without loss to assume that $v_1^n$ and $v_2^n$ converge. Let $\bar{v}_1$ and $\bar{v}_2$ denote the respective limits, and let $\bar{\beta} \equiv \lim_{n \to \infty}\beta^n$. Both $\bar{v}_1$ and $\bar{v}_2$ represent $\succsim^*$, and 
\[\bar{v}_2 = \alpha\frac{\gamma_1(1-\gamma_2)}{(1-\gamma_1)\gamma_2}\bar{v}_1 + \bar{\beta}.\]
We show that $\bar{v}_1$ is non-redundant. Suppose there exist $\lambda \in (0, 1)$ and $\beta \in \mathbb{R}$ such that $\lambda \bar{v}_1 + \beta \in \mathcal{V}_1$. We can pass to a convergent subsequence of the $x^n$. For $n$ sufficiently large, we will have
\[\left(\lambda \bar{v}_1 + \beta\right)(q) - \left(\lambda \bar{v}_1 + \beta\right)(x^n) < v_1^n(q) - v_1^n(x^n).\]
This contradicts the definition of $v_1^n$. The same argument works for $\bar{v}_2$, so both $\bar{v}_1$ and $\bar{v}_2$ are non-redundant. Now consider any other $v_1 \in \underline{\mathcal{V}_1}$ and $v_2 \in \underline{\mathcal{V}_2}$ that represent $\succsim^*$. Since $v_1$ and $\bar{v}_1$ are both non-redundant, they differ by a constant at most. The same is true for $v_2$ and $\bar{v}_2$. Thus,
\begin{equation}
\label{eq:v2vsv1}
v_2 = \alpha \frac{\gamma_1(1-\gamma_2)}{(1-\gamma_1)\gamma_2}v_1 + \beta \text{ for some } \beta \in \mathbb{R}.
\end{equation}

Now we show that $\gamma_1 = \gamma_2$. Since $u \in \mathcal{V}_1$, we have that 
\[\alpha \frac{\gamma_1(1-\gamma_2)}{(1-\gamma_1)\gamma_2}u + \beta \in \mathcal{V}_2 \text{ for some } \beta \in \mathbb{R}.\]
Since $\alpha u + \beta_u \in \underline{\mathcal{V}_2}$, it must be that
\[\frac{\gamma_1(1-\gamma_2)}{(1-\gamma_1)\gamma_2} \geq 1.\]
This implies $\gamma_1 \geq \gamma_2$. Since $\alpha u + \beta_u \in \mathcal{V}_2$, we have that
\[\frac{\gamma_2(1-\gamma_1)}{(1-\gamma_2)\gamma_1}u + \beta \in \mathcal{V}_1 \text{ for some } \beta \in \mathbb{R}.\]
Since $u \in \underline{\mathcal{V}_1}$, it must be that
\[\frac{\gamma_2(1-\gamma_1)}{(1-\gamma_2)\gamma_1} \geq 1.\]
This implies $\gamma_2 \geq \gamma_1$. 

Finally, take any $v_1 \in \underline{\mathcal{V}_1}$. Since $v_1$ represents a preference in $\mathcal{V}_\text{pref}$, there exists $v_2 \in \underline{\mathcal{V}_2}$ that represents the same preference. Plugging $\gamma_1 = \gamma_2$ into (\ref{eq:v2vsv1}), we obtain
\[v_2 = \alpha v_1 + \beta \text{ for some } \beta \in \mathbb{R}.\]

\newpage
\section{Identification with unit demand and quasi-linear utility}\label{app:quasiliner_ID}

Consider an agent whose rationales have the form $\theta \mathbbm{1} - \tau$ for $\theta \in [0, \overline{\theta}]$, where $\mathbbm{1}$ is an indicator for whether the agent gets the object and $\tau$ is his net transfer.  The model has three free parameters; the material payoff from the object $\theta^*$, the highest rationale $\overline{\theta}$, and the weight on rationalization utility $\gamma$.

We can identify the model by observing the agent's choices in decision problems of this form:
\begin{enumerate}
    \item At time $1$, the agent chooses either to start the transaction or quit.
    \item The agent learns prices $p_1, p_2 \in \mathbb{R}_0^+$, and pays $p_1$ if he started the transaction.
    \item At time $2$, if the agent started the transaction, he can pay a further $p_2$ to finish the transaction and get the object.
\end{enumerate}
Consider a classical agent with $\gamma = 0$. Conditional on starting the transaction, the starting fee $p_1$ is sunk, and the highest $p_2$ that he would accept is equal to $\theta^*$. By contrast, for a rationalizer, the highest $p_2$ he would accept is equal to $\theta^* + \frac{\gamma}{1-\gamma} p_1$ if $p_1 < (1-\gamma)(\overline{\theta} - \theta^*)$ and equal to $(1-\gamma) \theta^* + \gamma \overline{\theta}$ otherwise. \Cref{fig:WTP_graph} illustrates. 
\begin{figure}[ht!]
    \centering
    \captionsetup{width = .65\textwidth}
    \includegraphics[width =.65\textwidth]{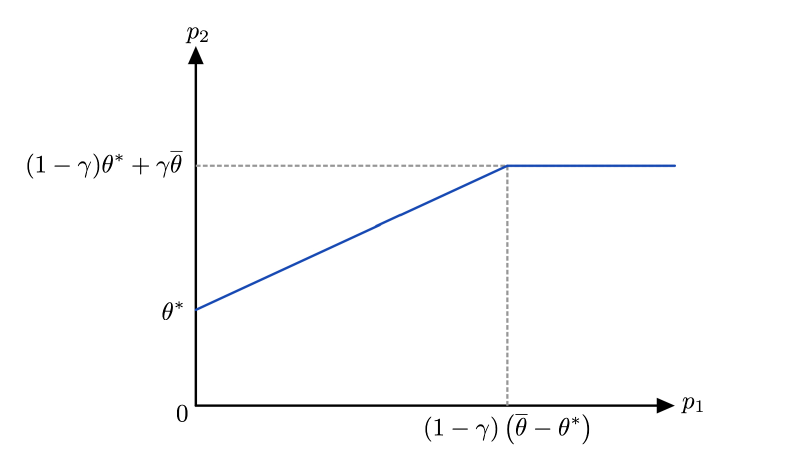}
    \caption{Highest $p_2$ that a rationalizer would accept, as a function of $p_1$.}
    \label{fig:WTP_graph}
\end{figure}

The relevant parameters are identified. $\theta^*$ is the intercept; $\gamma$ determines the slope of the increasing piece, and $\overline{\theta}$ is pinned down by the height of the plateau.  This graph is for an agent facing a take-it-or-leave-it offer at time-$2$, but a similar calculation would yield identification for an agent who learns $p_1$ and then faces a BDM mechanism \citep{becker1964measuring}.
\section{Identification with only on-path choices}\label{app:onpath_id}

One plausible objection to our identification procedure is that some pairs $(M,\mathcal{M})$ could be \textbf{off-path}, meaning that the agent would not choose $M$ from $\mathcal{M}$ at time $1$. (The off-path pairs depend on whether the agent is a na\"if, a sophisticate, or an empathetic sophisticate, as well as on the agent's prior.) If we rule out trembles, can we still elicit the counterfactual choice $c_2^s(M \mid \mathcal{M})$?

We can elicit $c_2^s(M \mid \mathcal{M})$ using only on-path choices, provided that we expand the state space to include a payoff-irrelevant horse race, and offer choices between race-contingent menus.\footnote{\citet[p. 90]{savage1972foundations} argued that some kinds of events are payoff-irrelevant: ``Consider, for example, a lottery in which numbered tickets are drawn from a drum. It seems clear that for an ordinary person the outcome of the lottery is utterly irrelevant to his life, except through the rules of the lottery itself."}

The construction is as follows: We expand the state space to be $S \times \{h, h'\}$, where $h$ and $h'$ represent different results for the horse race. By assumption, for all rationales $v$, we have $v^{s,h} = v^{s,h'} = v^s$.  A \textbf{race-contingent menu} is a function $N: \{h, h'\} \rightarrow \mathcal{K}_f(\Delta(Z))$. The agent now faces decision problems of this form:
\begin{enumerate}
    \item At $t=1$, the agent selects a race-contingent menu $N$ from a finite collection of such menus $\mathcal{N}$.
    \item The agent learns the state $s$ and the race result $\rho \in \{h,h'\}$.
    \item At $t = 2$, the agent chooses a lottery $\la$ from the menu $N(\rho)$.
\end{enumerate}

Suppose we wish to elicit $c_2^s(M \mid \mathcal{M})$ for some arbitrary pair $(M,\mathcal{M})$.  This is trivial if $M = \bigcup \mathcal{M}$, so suppose $M \subsetneq \bigcup \mathcal{M}$.  We construct two race-contingent menus, $N$ and $N'$, such that $N(h) = N'(h') = M$ and $N(h') = N'(h) = (\bigcup \mathcal{M}) \setminus M$.  We then offer the agent, at time $1$, the collection $\mathcal{N} = \{N,N'\}$.

By construction, $c_2^s(M \mid \mathcal{M})$ is equal to the agent's choice from $N(h)$ in state $(s,h)$ and also equal to the agent's choice from $N'(h')$ in state $(s,h')$, that is
\begin{equation}
    \begin{split}
        c_2^s(M \mid \mathcal{M})&\equiv\argmax_{\la \in {\color{red} M}}\left\{(1-\gamma)u^s(\la) + \gamma\max_{v^s \in \mathcal{V}^s}\left\{v^s(\la) - \max_{\hat{\la} \in {\color{red} \bigcup \mathcal{M}}} v^s(\hat{\la}) \right\}\right\} \\
        = c_2^{s,h}(N(h) \mid \mathcal{N}) &\equiv\argmax_{\la \in {\color{red} N(h)}}\left\{(1-\gamma)u^s(\la) + \gamma\max_{v^s \in \mathcal{V}^s}\left\{v^s(\la) - \max_{\hat{\la} \in {\color{red} N(h) \cup N'(h)}} v^s(\hat{\la}) \right\}\right\}\\
        = c_2^{s,h'}(N'(h') \mid \mathcal{N}) &\equiv\argmax_{\la \in {\color{red} N'(h')}}\left\{(1-\gamma)u^s(\la) + \gamma\max_{v^s \in \mathcal{V}^s}\left\{v^s(\la) - \max_{\hat{\la} \in {\color{red} N(h') \cup N'(h')}} v^s(\hat{\la}) \right\}\right\}.
    \end{split}
\end{equation}
If the agent chooses $N$, then we observe $c_2^s(M \mid \mathcal{M})$ in state $(s,h)$. Otherwise, he chooses $N'$, and we observe $c_2^s(M \mid \mathcal{M})$ in state $(s,h')$.

In summary, in the expanded model with race-contingent menus, the model primitives are identified as in \Cref{thm:id} even if we restrict the data to exclude off-path pairs $(N,\mathcal{N})$.

If we rule out contingent menus, is the model still identified by on-path behavior? Some further assumption is required. To illustrate, suppose that material utility $u$ does not depend on the state. Then the agent never makes an \textit{ex post} mistake, so on-path behavior does not identify $\gamma$ or $\mathcal{V}$. However, if we assume that the agent's \textit{ex ante} preferences over $\Delta(Z)$ are `sufficiently misaligned' with her \textit{ex post} preferences, then the model is still identified by on-path behavior. We omit this result for brevity.

\end{document}